\documentclass[11pt, letterpaper]{article}
\usepackage{fullpage}
\usepackage[T1]{fontenc}

\usepackage[colorlinks=true,pdfpagemode=none,linkcolor=blue,citecolor=blue]{hyperref}
\usepackage{amsmath,amsfonts,amsthm}
\usepackage[figure,boxed,lined]{algorithm2e}
\usepackage{color}
\usepackage{multirow}
\usepackage{enumerate}
\usepackage{graphicx}
\newtheorem{theorem}{Theorem}[section]
\newtheorem{corollary}[theorem]{Corollary}
\newtheorem{lemma}[theorem]{Lemma}
\newtheorem{observation}[theorem]{Observation}

\newtheorem{definition}[theorem]{Definition}

\newtheorem{invariant}[theorem]{Invariant}

\newcommand{\Cinf}{C_{\infty}}

\newcommand{\domX}{\mathcal{X}}
\newcommand{\bbR}{\mathbb{R}}
\newcommand{\bbZ}{\mathbb{Z}}
\newcommand{\onev}{\mathbf{1}}
\newcommand{\zerov}{\mathbf{0}}

\newcommand{\inorm}[1]{\left\|#1\right\|_{\infty}}

\newcommand{\tO}[1]{\widetilde{O}\left(#1\right)}

\newcommand{\lphi}{m}
\newcommand{\uphi}{M}

\renewcommand{\oe}{\bar{e}}

\newcommand{\energy}[2]{\mathcal{E}_{#1}(#2)}

\newcommand{\energymax}{\mathcal{E}^{\max}}

\newcommand{\eps}{\varepsilon}

\newcommand{\tf}{\widetilde{f}}
\newcommand{\hf}{\widehat{f}}

\newcommand{\ow}{\bar{w}}

\newcommand{\hrho}{\widehat{\rho}}

\newcommand{\valpha}{\boldsymbol{\mathit{\alpha}}}
\newcommand{\vphi}{\boldsymbol{\mathit{\phi}}}
\newcommand{\vdelta}{\boldsymbol{\mathit{\delta}}}

\newcommand{\vrho}{\boldsymbol{\mathit{\rho}}}

\renewcommand{\aa}{\boldsymbol{\mathit{a}}}
\newcommand{\bb}{\boldsymbol{\mathit{b}}}

\newcommand{\cc}{\boldsymbol{\mathit{c}}}
\newcommand{\dd}{\boldsymbol{\mathit{d}}}

\newcommand{\ff}{\boldsymbol{\mathit{f}}}
\renewcommand{\gg}{\boldsymbol{\mathit{g}}}
\newcommand{\tff}{\boldsymbol{\mathit{\widetilde{f}}}}
\newcommand{\tgg}{\boldsymbol{\mathit{\widetilde{g}}}}

\newcommand{\hff}{\boldsymbol{\mathit{\widehat{f}}}}
\newcommand{\hyy}{\boldsymbol{\mathit{\widehat{y}}}}

\renewcommand{\gg}{\boldsymbol{\mathit{g}}}
\newcommand{\hh}{\boldsymbol{\mathit{h}}}

\newcommand{\rr}{\boldsymbol{\mathit{r}}}

\renewcommand{\ss}{\boldsymbol{\mathit{s}}}

\newcommand{\uu}{\boldsymbol{\mathit{u}}}

\newcommand{\ww}{\boldsymbol{\mathit{w}}}
\newcommand{\oww}{\boldsymbol{\mathit{\bar{w}}}}
\newcommand{\xx}{\boldsymbol{\mathit{x}}}

\newcommand{\txx}{\boldsymbol{\widetilde{\mathit{x}}}}
\newcommand{\yy}{\boldsymbol{\mathit{y}}}
\newcommand{\tyy}{\boldsymbol{\widetilde{\mathit{y}}}}
\newcommand{\ty}{{\widetilde{{y}}}}
\newcommand{\zz}{\boldsymbol{\mathit{z}}}

\renewcommand{\AA}{\boldsymbol{\mathit{A}}}
\newcommand{\BB}{\boldsymbol{\mathit{B}}}
\newcommand{\CC}{\boldsymbol{\mathit{C}}}
\newcommand{\DD}{\boldsymbol{\mathit{D}}}

\newcommand{\PP}{\boldsymbol{\mathit{P}}}

\newcommand{\cinf}{c_{\infty}}
\renewcommand{\epsilon}{\varepsilon}

\newcommand{\tfstar}{{{\widetilde{f}}_{\!\star}}}

\newcommand{\CCstar}{\boldsymbol{\mathit{C}_{\!\star}}}

\newcommand{\Rstar}{R_{\star}}
\newcommand{\Rp}{R_{p}}
\newcommand{\Gstar}{G_{\!\star}}
\newcommand{\Estar}{E_{\!\star}}
\newcommand{\vstar}{v_{\star}}

\newcommand{\tffstar}{\boldsymbol{\mathit{\widetilde{f}}_{\!\star}}}\newcommand{\txxstar}{\boldsymbol{\mathit{\widetilde{x}}_{\!\star}}}

\newcommand{\tdd}{\boldsymbol{\mathit{\widetilde{d}}}}
\newcommand{\hdd}{\boldsymbol{\mathit{\widehat{d}}}}
\newcommand{\hbb}{\boldsymbol{\mathit{\widehat{b}}}}
\newcommand{\hxx}{\boldsymbol{\mathit{\widehat{x}}}}

\newcommand{\flog}{\widetilde{\log}}

\usepackage[dvipsnames]{xcolor}

\begin{document}

\clubpenalty=10000
\widowpenalty = 10000

\title{Circulation Control for Faster Minimum Cost Flow in Unit-Capacity Graphs}
\author{
Kyriakos Axiotis\thanks{MIT, \tt{kaxiotis@mit.edu}.}
\and
Aleksander M\k{a}dry\thanks{MIT, {\tt madry@mit.edu}.} 
\and
Adrian Vladu\thanks{Boston University, Department of Computer Science, {\tt avladu@mit.edu}.}}

\date{}
\maketitle
\begin{abstract}
We present an $m^{4/3+o(1)}\log W$-time algorithm for solving the minimum cost flow problem in graphs with unit capacity, where $W$ is the maximum absolute value of any edge weight. For sparse graphs, this improves over the best known running time for this problem and, by well-known reductions, also implies improved running times for the shortest path problem with negative weights, minimum cost bipartite $\boldsymbol{\mathit{b}}$-matching when $\|\boldsymbol{\mathit{b}}\|_1 = O(m)$, and recovers the running time of the  currently fastest algorithm for maximum flow in graphs with unit capacities (Liu-Sidford, 2020).

Our algorithm relies on developing an interior point method--based framework which acts on the space of circulations in the underlying graph. From the combinatorial point of view, this framework can be viewed as iteratively improving the cost of a suboptimal solution by pushing flow around circulations. These circulations are derived by computing a regularized version of the standard Newton step, which is partially inspired by previous work on the unit-capacity maximum flow problem (Liu-Sidford, 2019), and subsequently refined based on the very recent progress on this problem (Liu-Sidford, 2020). The resulting step problem can then be computed efficiently using the recent work on $\ell_p$-norm minimizing flows (Kyng-Peng-Sachdeva-Wang, 2019). We obtain our faster algorithm by combining this new step primitive with a customized preconditioning method, which aims to ensure that the graph on which these circulations are computed has sufficiently large conductance.

\end{abstract}

\thispagestyle{empty}
\newpage
\setcounter{page}{1}

\section{Introduction}

Finding the least costly way to route a demand through a network is a fundamental algorithmic primitive. Within the context of algorithmic graph theory it is captured as the \emph{minimum cost flow problem}, in which given a graph with costs on its arcs and a set of demands on its vertices, one needs to find a flow that routes the demand while minimizing its cost. This problem has received significant attention~\cite{ahuja1988network} and inspired the development of new algorithmic techniques. For example, Orlin's network simplex algorithm~\cite{orlin1997polynomial} offered an explanation of the excellent behavior that the simplex method exhibits in practice when applied to flow problems. More broadly, the recent progress on algorithms for the flow problems ~\cite{daitch2008faster,christiano2011electrical,madry2013navigating, lee2013new,sherman2013nearly,kelner2014almost,madry2016computing,peng2016approximate,cmsv17,sherman2017area,sidford2018coordinate,liu2019faster,liu2020faster} has been an instance of the general approach to graph algorithms that leverages the tools of continuous optimization, rather than classical combinatorial techniques. Also, there exist efficient reductions that enable us to leverage algorithms for the minimum cost flow problem to solve a host of other fundamental problems, including {the maximum flow problem}, {the minimum cost bipartite matching problem}, and {the shortest path problem with negative weights}.

\subsection{Our Contributions}

In this paper, we present an $m^{4/3+o(1)} \log W$-time algorithms for the minimum cost flow problem 
in graphs 
with unit capacities, where $W$ denotes the bound on the magnitude of the arc costs.
This improves upon the previously known $\widetilde{O}(m^{10/7} \log W)$ running time bound of Cohen et al.~\cite{cmsv17} and matches the running times of the recent algorithms due to Liu and Sidford~\cite{liu2019faster, liu2020faster} for the unit capacity maximum flow problem.\footnote{The initial version of this paper 
obtained a running time of $m^{11/8+o(1)}\log W$, which matched the running time of the then-fastest unit-capacity maximum flow algorithm due to Liu and Sidford \cite{liu2019faster}.  After this version was released \cite{axiotis2020circulation}, Liu and Sidford~\cite{liu2020faster} developed an improved running time of  $m^{4/3+o(1)}$ for the unit-capacity maximum flow problem. Their techniques turned out to be immediately adaptable to our minimum cost flow framework, and led to the current
$m^{4/3+o(1)}\log W$ running time for the unit-capacity minimum cost flow problem.}

Similarly to most of the relevant prior work, our algorithm at its core relies on an interior point method, but the variant of the interior point method we design and employ is 
directly attuned to the combinatorial properties of the graph. In particular, in contrast to~\cite{cmsv17}, we do not rely on a reduction to the bipartite perfect $\bb$-matching problem (which requires a  sophisticated analysis). Instead, our algorithm operates directly in the space of circulations of the original graph.

One can also draw an analogy between the network simplex method~\cite{orlin1997polynomial} and ours. The former navigated the corners of a feasible polytope and improved an existing suboptimal solution through pushing flow around cycles. In contrast, we iteratively improve our existing suboptimal solution by augmenting it with circulations, but navigate through the strict interior of the polytope, seeking to keep a specific condition called \textit{centrality}  satisfied.
Also, while in the network simplex case, the key difficulty is in finding the right  pivoting rule, our approach shifts the attention towards finding the right circulation to augment the flow with so as to maintain the centrality invariant.

A key ingredient of our approach is a custom preconditioning method, which enables us to control the flows we use to update the solution in each iteration. We derive a new way to tie the conductance of the graph to a certain guarantee on the flows computed in the preconditioned graph. This allows us to perform a better, tighter analysis of the quality of the preconditioner we use.

On a more technical level, our work provides a number of insights into the underlying interior point method. In particular, in our $m^{11/8+o(1)} \log W$-time algorithm (that we develop first), the progress steps we perform in order to reduce the duality gap of our current solution are cast as a refinement procedure, which simply attempts to correct a residual. This procedure is very similar to {iterative refinement}---widely used in the more restricted case of minimizing convex quadratic functions~\cite{wilkinson1994rounding, higham2002accuracy}. Also, in contrast to the classic approach for maintaining constraint feasibility during the interior point method update step---which relies on controlling the $\ell_2$ norm of the relative updates to the slack variables---we want to perform steps for which it is only guaranteed that these relative updates are small in $\ell_\infty$ norm. To this end, we employ a custom residual correction procedure that works by re-weighting the capacity constraints. (It is worth noting that a similar procedure has been used in~\cite{liu2019faster}.\footnote{While our analysis aims to enforce a small $\ell_2$ norm of the residual error, \cite{liu2019faster} seek to control the $\ell_4$ norm of the congestion vector. These techniques turn out to be largely equivalent.})

This paves the way for the final algorithm that has the further improved running time of $m^{4/3+o(1)} \log W$. 
As a matter of fact, the key bottleneck to obtaining a faster algorithm using the above approach is 
the need to ensure that the residual error in the solution obtained after performing a step bounded in $\ell_\infty$ norm 
can be reduced to zero. This requires increasing the weights on the constraint barriers, and these weight increases are exactly what limits the exponent in the running time to $11/8$.
The step problem we need to solve, however, is well conditioned within a local $\ell_\infty$ ball around the current iterate. Therefore, being able to certify that the point returned by solving the step problem optimally lies within this local $\ell_\infty$ ball, implies that we can efficiently find it using a direct optimization subroutine. This latter observation is the key insight in the very recent preprint by Liu-Sidford~\cite{liu2020faster} that enables them to improve the running time for maximum flow in graphs with unit capacity. We employ this insight in our setting in order to obtain an improved running time for unit-capacity minimum cost flow as well.

Finally, in order to guarantee that the $\ell_\infty$ norm of each progress step is indeed as small as needed, we employ a convex optimization subroutine with mixed $\ell_2$ and $\ell_p$ terms~\cite{kyng2019flows}, instead of solving a linear system of equations in each update step as is typically done. (Such subroutine was similarly used by Liu and Sidford~\cite{liu2019faster, liu2020faster}, in a slightly different form.)

\subsection{Previous Work}
Due to the size of the existing literature on the studied problems, we focus our discussion only on the works that are the most relevant to our results and refer the reader to~\cite{goldberg2017minimum} and Section 1.2 in~\cite{cmsv17} for a more detailed discussion.

In 2013, M\k{a}dry \cite{madry2013navigating} developed an algorithm that produces an optimal solution to the unit capacity maximum flow problem in $\widetilde{O}(m^{10/7})$ time and thus improves over a long standing running time barrier of $\widetilde{O}(n^{3/2})$ in the case of sparse graphs. 
An important characteristic of this approach was the careful tracking of an electrical energy quantity which allowed to control the step size.
The underlying approach was then simplified by providing a more direct correspondence between the update steps of the interior point method and computing augmenting paths via electrical flow computations~\cite{madry2016computing}. This framework has been also extended to a more general setting of unit capacity minimum cost flow~\cite{cmsv17}, achieving a running time of $\widetilde{O}(m^{10/7} \log W)$, where $W$ upper bounds the largest cost of a graph edge in absolute value.

In a different context, motivated by new developments involving regression problems~\cite{durfee2017, bubeck2018homotopy, adil2019iterative}, Kyng et al.~\cite{kyng2019flows} studied the $\ell_p$ regression problems on graphs, obtaining an algorithm which runs in $m^{1+o(1)}$ time for a range of large values of $p$. This algorithm's running time was subsequently further improved by Adil and Sachdeva~\cite{adil2020faster}. 

Liu and Sidford~\cite{liu2019faster} have recently obtained an improved algorithm for the unit capacity maximum flow problem with a running time of $m^{11/8+o(1)}$.
One of their key insights was that the work on $\ell_p$-regression problems
enables treating energy control as a self-contained problem in each iteration of the interior point method, rather than maintaining energy as a global potential over the whole course of the algorithm, which was the case in previous work. Then, in their recent follow-up work, Liu and Sidford~\cite{liu2020faster} strengthen the step problem primitive by directly optimizing a regularized log barrier function as opposed to performing a sequence of regularized Newton step.
This led to a running time of $m^{4/3+o(1)}$ for the unit capacity maximum flow problem.

\subsection{Organization of the Paper}
We begin with technical preliminaries in Section~\ref{sec:prelim}.
In Section~\ref{sec:vanilla_mcf}, we present our interior point framework specialized to minimum cost flow, and provide a basic analysis which yields an algorithm running in 
$\widetilde{O}(m^{3/2}\log W)$ time.
We further refine our framework in Section~\ref{sec:faster_mcf}, where we develop the key tools needed for our results, giving a faster, $m^{11/8+o(1)}\log W$-time algorithm for obtaining the solution to a slightly perturbed instance of the original minimum cost flow problem. In Section~\ref{sec:repair} we then show how to use existing combinatorial techniques to repair this perturbed instance. Finally, in Section~\ref{sec:improved_running_time}, we demonstrate how to combine the framework developed in the previous sections with an insight from the recent work of Liu and Sidford \cite{liu2020faster} to achieve the final running time of $m^{4/3+o(1)}\log W$.

\section{Preliminaries}
\label{sec:prelim}
In this section, we introduce some basic notation and definitions that we will need later.
\subsection{Basic Notation}
\paragraph{Vectors.} We use $\zerov$ and $\onev$ to represent the all-zeros and all-ones vectors, respectively. Given two vectors $\xx$ and $\yy$ of the same dimension, we use $\langle\xx,\yy\rangle$ to represent their inner product.
We apply scalar operations to vectors with the interpretation that they are applied point-wise, for example $\xx/\yy$ represents the vector whose $i^{th}$ entry is $x_i/y_i$. We use the inline notation $(\xx; \yy)$ to represent the concatenation of vectors $\xx$ and $\yy$.
\paragraph{Norms.}
Given a vector $\xx \in \bbR^n$ and a scalar $p \geq 1$, we write the $\ell_p$ norm of $\xx$ as $\|\xx\|_p = \left(\sum_{i=1}^n \vert x_i\vert^p\right)^{1/p}$. Using this definition we also obtain $\|\xx\|_{\infty} = \max_{i=1}^n \vert x_i \vert$. Throughout this paper we will be working especially with the $\ell_1$, $\ell_2$ and $\ell_\infty$ norms.
\paragraph{Graphs.}
Given a graph $G = (V,E)$ and a vertex $v\in V$, we will write $e\sim v$ to denote the set of edges $e\in E$ that are incident to $v$ in $G$, i.e.
the set of edges that have $v$ as an endpoint.
\paragraph{Asymptotic notation.} 
Given a parameter $m$ denoting the number of edges of a graph, we use
$\tO{c}$ to denote a quantity that is $O(c \log^k m)$ for some constant $k$.

\subsection{Minimum Cost Flow}
We denote by $G = (V, E, \cc)$ a directed graph with vertex set $V$, arc set $E$ and cost vector $\cc\in\mathbb{R}^{|E|}$. We denote by $m = \vert E \vert$ the number of arcs, and by $n=\vert V \vert$ the number of vertices in $G$. An arc $e$ of $G$  connects an ordered pair $(u,v)$, where $u$ is its \textit{tail} and $v$ is its \textit{head}. 
The basic notion of this paper is the notion of a flow. Given a graph $G$ we view a flow in $G$ as a vector $\ff \in \bbR^m$ that assigns a value to each arc of $G$. If this value is negative we interpret it as having a flow of $\vert f_e \vert$ flowing in the direction opposite to the arc orientation. This convention is especially useful when discussing flows in undirected graphs.

We will be working with flows in $G$ that satisfy a certain \textit{demand} $\dd \in \bbR^n$ such that $\sum_u d_u = 0$. 
We say that a flow $\ff$ satisfies or routes demand $\dd$ if it satisfies 
the flow conservation constraints with respect to the demands. That is:
\begin{equation}
\sum_{e\in E^+(u)} f_e - \sum_{e \in E^-(u)} f_e = d_u,\quad\textnormal{for all }u\in V\,.
\end{equation}
Here, $E^+(u)$ and $E^-(u)$ are the sets of arcs of $G$ that are entering $u$ and leaving $u$, respectively. Intuitively, these constraints enforce that the net balance of the total in-flow into vertex $u$ and the total out-flow leaving that vertex is equal to $d_u$. A flow for which the demand vector $\dd$ is zero everywhere is called a \emph{circulation}.

We say that a flow $\ff$ is \emph{feasible} (or that it respects capacities) in $G$ if it obeys the capacity constraints:
\begin{equation}
0 \leq f_e \leq u_e,\quad\textnormal{for all }e\in E\,,
\end{equation}
where $\uu \in \bbR^m$ is a vector of arc capacities.

The \emph{unit capacity minimum cost flow problem} is to find a flow $\ff\in\mathbb{R}^m$ that 
meets the \emph{unit capacity constraints} $0 \leq f_e \leq 1$ for all $e\in E$ and
routes the demand $\dd$, while minimizing the \emph{cost} $\sum\limits_{e\in E} c_e f_e$.

\paragraph{Cycle Basis.} A set of circulations $\mathcal{C}$ in $G$ is called a \emph{cycle basis} if 
any circulation in $G$ can be expressed as a linear combination of circulations in $\mathcal{C}$. 
If $G$ is connected, the dimension of a cycle basis of $G$ is $m-n+1$.

\section{Minimum Cost Flow by Circulation Improvement} \label{sec:vanilla_mcf}
In this section we present our (customized) interior point method--based framework for solving the minimum cost flow problem, setting the foundations for the faster algorithm
of Section~\ref{sec:faster_mcf}.

\subsection{LP Formulation and Interior Point Method}
\label{sec:lpform}
We first cast the minimum cost flow problem as a linear program that we then proceed to solve using an interior point method.

\paragraph{LP formulation.} It will be useful to consider the parametrization of a flow in terms of the 
circulation space of the graph.
The goal of this re-parametrization is to initialize the interior point method with an initial flow $\ff_0$ which routes the prescribed demand $\dd$, then iteratively improve it by adding circulations to get a flow which routes the same demand $\dd$ but has lower duality gap. 
It is noteworthy that the specific parametrization of the circulation space is irrelevant to the interior point method, due to its affine invariance. We will elaborate on this point later. For us it will be a useful tool for understanding the centrality condition 
arising from the interior point method
and applying more aggressive progress steps.

Given the (connected) underlying graph $G = (V,E)$, let $\CC \in \bbR^{m\times (m-n+1)}$ be a matrix whose columns encode the characteristic vectors of a basis for $G$'s circulation space.

In order to construct such a matrix, we let $C_1, C_2, ..., C_{m-n+1}$ be 
an arbitrary cycle basis for $G$, where we ignore the arc orientations. An easy way to produce one is to consider a spanning tree $T \subseteq G$. For each arc $(u,v) \in E$ which is not in $T$, consider the unique path in $T$ connecting $v$ and $u$. The arcs on this path along with the arc $(u,v)$ determine a cycle in the basis.
More specifically, consider the set of arcs of $G$ present in $C_i$, sorted according to the order in which they are visited along the cycle, starting with the off-tree arc $(u,v)$, then continuing with those witnessed along the tree path from $v$ to $u$. If an arc $e\in E$ has the opposite orientation to the one corresponding to the traversal of the cycle, we represent it as $\oe$, otherwise we write it just as $e$. 

Now, letting $C_i$ consist of a subset of arcs in $E$, each of which appears either with its original orientation $e$, or the opposite orientation $\oe$, we write the $i^{th}$ column of matrix $\CC$ as follows.

\[
\CC_{e,i} = 
\begin{cases}
1\,,  &\textnormal{if } e \in C_i\,, \\
-1\,, &\textnormal{if } \oe \in C_i\,, \\
0\,, &\textnormal{otherwise.}
\end{cases}
\]

We can now use $\CC$ to represent any circulation in 
$G$. 
Given any $\xx \in \bbR^{m-n+1}$ we have that $\ff = \CC\xx$ is a circulation. 
Furthermore the sign of each coordinate $f_e$, $e = (u,v)\in E$, shows whether $\ff_e$ is a flow that runs in the same direction as $e$ or vice-versa, i.e. $f_e > 0$ if $\ff$ carries flow from $u$ to $v$, and similarly $f_e < 0$ if $\ff$ carries flow from $v$ to $u$.
On the other hand, for any circulation $\ff\in\mathbb{R}^m$ there exists an $\xx\in\mathbb{R}^{m-n+1}$ such that 
$\ff = \CC \xx$ (or in other words the image of $\CC$ is the space of circulations).

Now let $\ff_0$ be a flow in $G$ such that for each arc $e\in E$, $0 < (\ff_0)_e < 1$, and $\ff_0$ routes the demand $\dd$. 
The minimum cost flow problem can be cast as the following linear program:
\begin{align}
&\min \,\langle \cc, \CC \xx \rangle  \label{eq:circ_lp}\\
&\zerov \leq \ff_0 + \CC \xx \leq \onev\,. \notag
\end{align}
We see that the objective value of this linear program differs by a term of $\langle \cc,\ff_0\rangle$ from the original objective. We did not include 
it here, since it is a constant. 
It is useful to also consider its dual: 
\begin{align}
&\max \, -\langle \onev - \ff_0, \yy^+\rangle - \langle \ff_0, \yy^-\rangle\label{eq:circ_dual}\\
& \CC^\top\left(\yy^+-\yy^-\right) = - \CC^\top c\notag\\
& \yy^+,\yy^- \geq \zerov\,.\notag
\end{align}

The objective we are left to solve simply suggests that in order to find the minimum cost flow in the graph with unit capacities, we equivalently have to find the minimum cost circulation in the residual graph under shifted capacity constraints.
This carries a significant similarity with the \textit{network simplex} algorithm~\cite{orlin1997polynomial}, which has been used in the past as a specialization of the simplex method to the minimum cost flow problem. It essentially consisted of maintaining a solution routing the prescribed demand $\dd$, and iteratively improving it by pushing flow around a cycle, while satisfying capacity constraints.
Rather than performing such updates, which always maintain a flow on the boundary of polytope corresponding to the set of feasible solutions, the interior point method maintains a more sophisticated condition on these intermediate solutions.
Another similar approach can be found in~\cite{kelner2013simple}, where updates are iteratively pushed around cycles in order to solve Laplacian linear systems.

Like these methods, our approach will be to 
repeatedly improve the cost of the solution by pushing augmenting circulations. Crucially, maintaining a solution centrality condition,
stemming from interior point methods,
will allow us to make significant progress during each augmentation step.
Figure~\ref{fig:ipm} is the high-level procedure for this algorithm. It consists of an
initialization procedure $\textsc{Initialize}$ which is described in Section~\ref{sec:init},
repeated circulation augmentations using procedure $\textsc{Augment}$ as described in Section~\ref{sec:correct_res},
and finally a standard procedure $\textsc{Repair}$ to round the 
returned solution with low duality gap to an optimal integral one (described in Section~\ref{sec:repair}).
The faster algorithm of Section~\ref{sec:faster_mcf} will also follow the same format,
but the $\textsc{Augment}$ and $\textsc{Repair}$ routines will be more sophisticated.

\begin{figure}[h!]
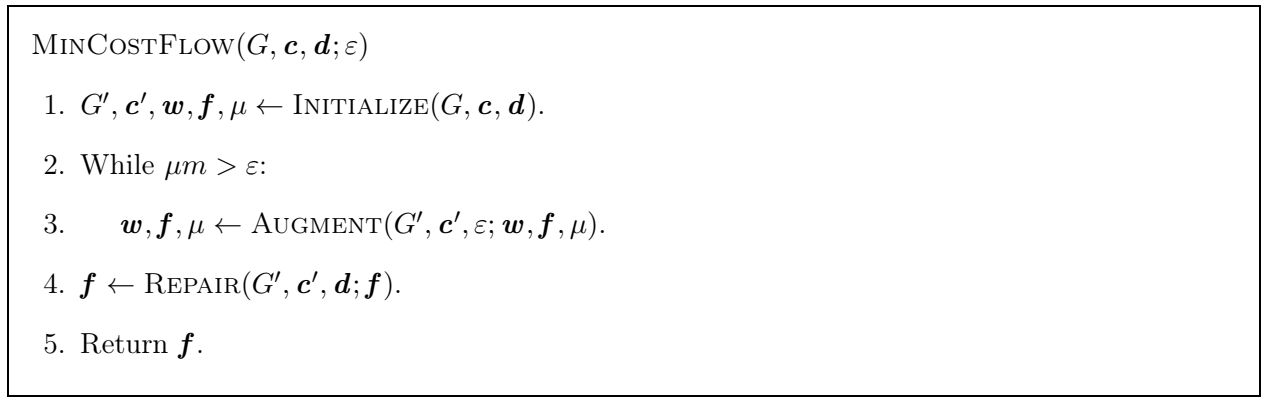

\frame{
\begin{minipage}{\textwidth}
\vspace{10pt}
\hspace{5pt}
$\textsc{MinCostFlow}(G,\cc,\dd;\epsilon)$
\begin{enumerate}

\item {$G', \cc', \ww, \ff, \mu \leftarrow \textsc{Initialize}(G, \cc, \dd)$.}

\item {While $\mu m > \epsilon$:}

\item{\ \ \ \ $\ww, \ff, \mu \leftarrow \textsc{Augment}(G', \cc', \epsilon; \ww, \ff, \mu)$.}

\item{$\ff \leftarrow \textsc{Repair}(G', \cc', \dd; \ff)$.}

\item{Return $\ff$.}

\end{enumerate}
\vspace{5pt}
\end{minipage}
}
\caption{Minimum Cost Flow by Circulation Improvement}
\label{fig:ipm}
\end{figure}

\paragraph{Barrier Formulation.} 
In order to apply an interior point method on (\ref{eq:circ_lp}), we need to replace the feasibility constraints by a convex barrier function.
We seek a nearly optimal solution, i.e. one that has small duality gap. 
The vanilla interior point method consists of iteratively finding the optimizer $\xx_{\mu}$ for a family of functions parametrized by $\mu > 0$
\begin{align}
\underset{\xx\in\mathbb{R}^{m-n+1}}{\min}\, F_\mu^{\ww}(\xx) = \frac{1}{\mu} \cdot \langle \cc, \CC \xx \rangle
- \sum_{e\in E}
 \left(  w_e^+ \cdot \log(\onev - \ff_0 - \CC \xx)_e + w_e^- \cdot \log(\ff_0 + \CC \xx)_e   \right)\,.
 \label{eq:barrier_formulation}
\end{align}
where $w_e^+, w_e^- > 0$ are weights on the flow capacity constraints.
In order to find the optimizer $\xx_{\mu}$, one performs Newton method on $F_{\mu}^{\ww}$, after warm starting with $\xx_{\mu(1+\delta)}$ for some $\delta > 0$. 

While classical methods maintain $\ww = \onev$ at all times, this extra parameter has been introduced 
in previous work in order to allow the method to make progress more aggressively.
To simplify notation we define the slack vector $\ss=(\ss^+;\ss^-)$ as
\begin{align}
\ss^+ &= \onev - \ff_0 - \CC \xx\,, \\
\ss^- &= \ff_0 + \CC \xx\,,
\end{align}
representing the upper slack (i.e. the distance of the current flow $\ff = \ff_0 + \CC\xx$ to the upper capacity constraint of $\ff \leq \onev$) and the lower slack  (i.e. the distance from the current flow to the lower capacity constraint $\zerov \leq \ff$). We will use the vector $\ww = (\ww^+; \ww^-)$ to represent the weights for the two sets of barriers that we are using.

\subsection{Optimality and Duality Gap} \label{sec:optimality}

In order to describe the method and analyze it, it is important to understand the optimality conditions for $F_\mu^{\ww}$. We say that a vector $\xx$ which minimizes $F_\mu^{\ww}$ is \textit{central} (or satisfies \textit{centrality}).
This condition is described in the following lemma, whose proof can be found in Appendix~\ref{sec:centrality_condition_proof}.

\begin{lemma}\label{lem:centrality_condition}
The vector $\xx$ is a minimizer for $F_{\mu}^{\ww}$ if and only if
\begin{align}
\CC^\top \left(\frac{\ww^+}{\ss^+} - \frac{\ww^-}{\ss^-}\right) = - \frac{\CC^\top \cc}{\mu}\,. \label{eq:centrality}
\end{align}
Furthermore the vector $\yy=(\yy^+;\yy^-)$
with $\yy^+=\mu \cdot  \frac{ \ww^+ }{\ss^+}$,
$\yy^-=\mu \cdot  \frac{ \ww^- }{\ss^-}$
is a feasible dual vector, and the duality gap of the primal-dual solution $(\xx, \yy)$ is exactly $\mu \| \ww\|_1$.
\end{lemma}

Maintaining the centrality condition (\ref{eq:centrality}) will be the key challenge in obtaining a faster interior point method for this linear program.
This emphasizes the fact that the aim of this method is to construct a feasible set of slacks 
$\ss^+ = \onev - \ff_0 - \CC \xx  > \zerov$ 
and $\ss^- = \ff_0 + \CC \xx  > \zerov$ 
such that $\CC^\top \left(\frac{\ww^+}{\ss^+}-\frac{\ww^-}{\ss^-}\right) = -\frac{\CC^\top \cc}{\mu}$ for a very small $\mu>0$. 
It is important to note that even though the existence of such an $\xx$ needs to be guaranteed, it is not necessary to explicitly maintain it.
This will be apparent in the definition below.

\begin{definition}[$\mu$-central flow]
Given weights $\ww=(\ww^+;\ww^-)\in\mathbb{R}_{>0}^{2m}$, 
a flow $\zerov < \ff < \onev$ is called \emph{$\mu$-central} with respect to $\ww$ 
if for some cycle basis matrix $\CC\in\mathbb{R}^{m\times (m-n+1)}$,
\begin{align} \label{eq:central_flow}
\CC^\top \left(\frac{\ww^+}{\onev - \ff} - \frac{\ww^-}{\ff}\right) = - \frac{\CC^\top \cc}{\mu}
\end{align}
for some $\mu > 0$.
We will call the parameter $\mu$ the \emph{centrality} of $\ff$ with respect to $\ww$.
\end{definition}
It should be noted that the precise choice of cycle basis $\CC$ 
in the above definition 
is irrelevant, as the property is invariant under the choice of cycle basis.

\subsection{Initialization} \label{sec:init}

The initialization procedure description
and analysis is standard and thus
deferred to Section~\ref{sec:init_appendix}. From now on we assume that
$G$ is the graph produced by the procedure in Section~\ref{sec:init_appendix}, together with
a $\mu$-central flow with $\mu \leq  2\|\cc\|_2$.

\subsection{The Augmentation Procedure and Routing the Residual} \label{sec:correct_res}
The $\textsc{Augment}$ procedure can be seen in Figure~\ref{fig:step_problem}.
It consists of computing an augmenting flow $\tff$ by solving a linear system, augmenting the current solution by that flow,
and finally 
calling a correction procedure in order to enforce the centrality of the solution.

\begin{figure}
\frame{
\begin{minipage}{\textwidth}
\vspace{10pt}
\hspace{5pt}
$\textsc{Augment}(G,\cc,\epsilon; \ww, \ff, \mu)$
\begin{itemize}
\item Given $\ff$ : $\mu$-central flow with respect to weights $\ww$.
\item Returns $\ff''$ : $\mu'$-central flow with respect to weights $\ww'$.
\end{itemize}

\begin{enumerate}

\item {
Let $\tff$ be a
solution to 
(\ref{eq:flow_linsys}), where 
\begin{align*}
&\hh = 
\delta\left(\frac{\ww^+}{\onev - \ff} - \frac{\ww^-}{\ff}\right)\,.\end{align*}
}
\item{ Compute the congestion vector $\vrho=(\vrho^+;\vrho^-)$ as $\rho_e^+ = \frac{\tf_e}{1-f_e}$, $\rho_e^- = \frac{-\tf_e}{f_e}$, cf. (\ref{eq:rhodef}).}
\item{Augment flow $\ff' = \ff + \tff$.}

\item{ Correct residual given by
	$\hh' = -\left(\frac{\ww^+}{\onev-\ff'} - \frac{\ww^-}{\ff'} + \frac{\cc}{\mu/(1+\delta)}\right)$
using 
	Lemma~\ref{lem:correction} and get new weights $\ww'$, flow $\ff''$, and centrality parameter $\mu'$.} 

\item{ Return $\ww', \ff'', \mu'$.}

\end{enumerate}
\vspace{5pt}
\end{minipage}
}
\caption{Circulation improvement step}
\label{fig:step_problem}
\end{figure}

We think of the interior point method as iteratively augmenting a feasible 
flow $\zerov < \ff < \onev$
satisfying $\CC^\top\left( \frac{\ww^+}{\onev - \ff} - \frac{\ww^-}{\ff}\right) = -\frac{\CC^\top \cc}{\mu}$ into another 
flow $\ff'$ such that 
$\CC^\top\left( \frac{\ww^+}{\onev - \ff'} - \frac{\ww^-}{\ff'}\right) = -\frac{\CC^\top \cc}{\mu'}$ 
where $\mu' = \mu/(1+\delta)$ for some $\delta > 0$. The magnitude of $\delta$ for which we are able to do so dictates the rate of the convergence of the method, since in $O(\delta^{-1})$ such iterative steps, the parameter $\mu$ is reduced to half, and hence the duality gap also reduces by a factor of $1/2$, per Lemma~\ref{lem:centrality_condition}. 

We can interpret such an iterative step as a residual-fixing procedure. Given a flow $\ff$ with slacks
$\ss^+ = \onev - \ff > \zerov, \ss^- = \ff > \zerov$, we consider the residual
\[
\CC^\top \left(\frac{\ww^+}{\ss^+} - \frac{\ww^-}{\ss^-} + \frac{\cc}{\mu'}\right) = \nabla F_{\mu'}^{\ww}(\xx)\,,
\]
which is exactly the amount by which the target condition (\ref{eq:centrality}) fails to be satisfied. 
We denote the current residual as $\nabla F_\mu^{\ww}(\xx) = -\CC^\top \hh$ for some $\hh \in \bbR^{m}$. The goal of the iterative step is therefore to provide a feasible 
update rule to the flow such that the residual shrinks in some metric. To do so, we must first define a few useful notions.

\begin{definition}[Energy of the residual]
\label{def:energy}
Given a residual $-\CC^\top \hh$ for some $\hh \in \bbR^m$ and positive vectors $\ww=(\ww^+;\ww^-)\in\bbR^{2m}, \ss=(\ss^+;\ss^-) \in \bbR^{2m}$ we define the energy to route $\hh$ with resistances determined by $(\ww, \ss)$ as
\begin{align} \label{eq:dual_energy}
\energy{\ww,\ss}{\hh} = \min_{\tyy: \CC^\top (\tyy + \hh) = \zerov} \frac{1}{2}
\sum\limits_{e\in E} \left(\frac{w_e^+}{(s_e^+)^2} + \frac{w_e^-}{(s_e^-)^2}\right)^{-1} (\ty_e)^2\,.
\end{align}
\end{definition}
The following lemma gives equivalent formulations for the energy which will be useful for our analysis. Its proof can be found in Appendix~\ref{sec:lem_equiv_energy_proof}.
\begin{lemma}[Equivalent energy formulations] 
\label{lem:equiv_energy}
Given $\hh\in\bbR^m$ and $\ww,\ss \in \bbR^{2m}$, one can write:
\begin{align}\label{eq:energy}
\energy{\ww,\ss}{\hh} = \underset{\tff=\CC\txx}{\max}\, 
\langle \hh, \tff \rangle - \frac{1}{2} \sum_{e\in E} \left(\frac{w_e^+}{(s_e^+)^2} + \frac{w_e^-}{(s_e^-)^2}\right)(\tf_e)^2\,.
\,
\end{align}
Equivalently, the following conditions are satisfied:
\begin{align} \label{eq:flow_linsys}
\tff & = \CC \txx\,,\\ 
\CC^\top \left(\frac{\ww^+}{(\ss^+)^2} + \frac{\ww^-}{(\ss^-)^2}\right) \tff &= \CC^\top \hh\,.\notag
\end{align}
The latter equality can be re-stated in terms of the congestion vector 
\begin{align}
\vrho:=\left(\vrho^+=\frac{\tff}{\ss^+};\vrho^-=\frac{-\tff}{\ss^-}\right) \label{eq:rhodef}
\end{align}
as
\begin{align}
\CC^\top \left(\frac{\ww^+ \vrho^+}{\ss^+} - \frac{\ww^- \vrho^-}{\ss^-}\right) &= \CC^\top \hh \,.\label{eq:linsys}
\end{align}

The energy can then be written as 
\[
\energy{\ww,\ss}{\hh} = \frac{1}{2}\sum_{e\in E} \left(w_e^+ (\rho_e^+)^2+w_e^-(\rho_e^-)^2\right) = \frac{1}{2} \sum\limits_{e\in E} \left(\frac{w_e^+}{(s_e^+)^2} + \frac{w_e^-}{(s_e^-)^2}\right) (\tf_e)^2\,.
\]
\end{lemma}

Using Lemma~\ref{lem:equiv_energy} we can prove that
under certain conditions
we can update the flow while simultaneously maintaining its feasibility and 
reducing the residual. Let us first define a residual correction step.

\begin{definition}[Residual correction]
\label{def:residual_correction}
Let $\ww, \ss \in \bbR^{2m}$ where $\ww > \zerov$ and 
$\ff = \ff_0 + \CC \xx$ is a feasible flow with slacks
$\ss^+ = \onev - \ff > 0$,
$\ss^- = \ff > 0$,
and let $\nabla F_{\mu}^{\ww}(\xx)=-\CC^\top \hh$ be the corresponding residual. A residual correction step is defined as an update to the $\xx$ vector, and implicitly to the flow vector $\ff$ via:
\begin{align}
\xx' &= \xx + \txx \,, \label{eq:vlss1}\\
\ff' &= \ff + \CC \txx = \ff + \tff\,,\label{eq:vlss2}
\end{align}
where $\txx$ 
is the solution to the linear system
\begin{align}
\tff &= \CC \txx\,, \\
\CC^\top \left(\frac{\ww^+}{(\ss^+)^2} + \frac{\ww^-}{(\ss^-)^2}\right)\tff & = \CC^\top \hh\,.\end{align}
\end{definition}

Since $\tff$ is a circulation, this shows that 
the residual correction steps of the vanilla interior point method
preserve the demand $\dd$ by adding an augmenting circulation to the current flow $\ff$.
It is also important to ensure that such updates do not break the LP feasibility constraints, i.e. $\zerov \leq \ff \leq \onev$ at all times.
This will be made true by appropriately scaling the residual, thus enforcing $\|\vrho\|_\infty < 1/4$, or equivalently  
$\left|\tf_e \right| < \frac{1}{4} \min\left\{ 1-f_e, f_e\right\}$ for all $e\in E$. 

It is worth noting that the flow $\tff$ which corresponds to solving the linear system from (\ref{eq:vlss1}-\ref{eq:vlss2}) can be computed in $\widetilde{O}(m)$ time using a fast Laplacian solver~\cite{spielman2004nearly, kelner2013simple, cohen2014solving, peng2014efficient}. This may not be immediately obvious given the cycle basis formulation. But, in fact, this follows very easily from writing the linear system solve as a convex quadratic minimization problem. We do not go into further detail here, since we will elaborate more on this topic in Section~\ref{sec:faster_mcf}.

In order to analyze the algorithm, we use energy as a potential function.

\begin{lemma}[Energy after residual correction]
\label{lem:energy_contract}
Let $\ww\in\bbR^{2m}$ where $\ww > \zerov$, and $\ff = \ff_0 + \CC\xx$ be a flow vector with slacks $\ss>\zerov$.
Let $\cc \in \bbR^m$ and the residual $-\CC^\top \hh = \nabla F_{\mu}^{\ww}(\xx) = \CC^\top \left(\frac{\ww^+}{\ss^+} - \frac{\ww^-}{\ss^-} + \frac{\cc}{\mu}\right)$ for some $\mu > 0$
and $\vrho\in\bbR^{2m}$ be the corresponding congestion vector. 
Then a residual correction step produces a new flow $\ff' = \ff_0 + \CC\xx'$ with slacks $\ss'$ 
and residual $-\CC^\top \hh'= \nabla F_{\mu}^{\ww}(\xx') = \CC^\top \left(\frac{\ww^+}{(\ss^+)'} - \frac{\ww^-}{(\ss^-)'} + \frac{\cc}{\mu}\right)$
such that
\begin{align*}
\energy{\ww,\ss'}{\hh'} 
\leq \frac{1}{2} \sum\limits_{e\in E} \left(w_e^+ (\rho_e^+)^4 + w_e^-(\rho_e^-)^4\right)\,.
\end{align*}
\end{lemma}

The proof can be found in Appendix~\ref{sec:energy_contract_proof}. As a corollary, we show that if the energy required to route the residual is small to begin with, after performing a step from Lemma~\ref{lem:energy_contract}, it quickly contracts. 
\begin{corollary}
\label{cor:energy_contract}
Let $\ww\in\bbR^{2m}$ where $\ww \geq \onev$, and $\ff = \ff_0 + \CC\xx$ be a flow vector with slacks $\ss>\zerov$ and residual $\nabla F_{\mu}^{\ww}(\xx) = -\CC^\top \hh$.
Updating $\xx$ and $\ff$ via a residual correction step as in Lemma~\ref{lem:energy_contract} yields a new flow $\ff'=\ff_0+\CC\xx'$
with slacks $\ss'>\zerov$ and residual
$\nabla F_{\mu}^{\ww}(\xx') = -\CC^\top \hh'$
such that 
\begin{align*}
\energy{\ww,\ss'}{\hh'} \leq 2\cdot \energy{\ww,\ss}{\hh}^2\,.
\end{align*}
\end{corollary}
We defer the proof to Appendix~\ref{sec:cor_energy_contract_proof}.

Corollary~\ref{cor:energy_contract} shows that with a good initialization, residual correction steps decrease the energy of the residual very fast. 
In other words, starting from a $\mu$-central flow $\ff$
one can quickly obtain a $\mu'$-central flow $\ff'$ 
with a smaller duality gap, i.e. with $\mu' < \mu$. 
This is the workhorse of the 
vanilla interior point method
that we proceed to analyze in Section~\ref{sec:vanilla_ipm}.

Before that, we introduce an additional residual correction step that ensures that our solution is always \emph{exactly} central.
Since residual correction reduces the residual very fast, i.e. energy gets reduced to $\eps$ in $O(\log\log \eps^{-1})$ steps, we can intuitively think of it as a method which effectively removes the residual in $\tO{1}$ steps. To make this intuition rigorous, we first reduce the residual energy to $m^{-O(1)}$, then we force optimality conditions by changing the weights $\ww$. While this perfect correction step is generally unnecessary, it will make the description and analysis of our algorithm somewhat cleaner since we will always be able to assume exact centrality.

\begin{lemma}[Perfect correction]\label{lem:fine_correction}
Let 
$\ww\in\bbR^{2m}$ be
a set of weights 
such that $\onev \leq \ww$, and let $\ff = \ff_0 + \CC\xx$ be a flow vector with slacks $\ss>\zerov$ and residual $\nabla F_{\mu}^{\ww}(\xx) = -\CC^\top \hh$
such that 
$\energy{\ww, \ss}{\hh} \leq \eps \leq 1/100$.
Then one can compute weights $\ww' \in \bbR^{2m}$ such that 
$\ww \leq \ww'$ and $\|\ww'-\ww\|_1 \leq \|\ww\|_1 \cdot 4\sqrt{\epsilon}$ for which the residual 
$\nabla F_{\mu'}^{\ww'}(\xx) = \zerov$, where $\mu' \leq \mu(1+2\sqrt{\epsilon})$.
\end{lemma}
The complete proof can be found in Appendix~\ref{sec:fine_correction_proof}. This lemma shows that after performing residual correction until the energy becomes smaller than $m^{-20}/4$, we can slightly increase the weights from $\ww$ to $\ww'$ such that for the new objective, 
$\xx$ exactly satisfies the optimality condition, as in Equation (\ref{eq:centrality}). The effect of this perfect correction is an extremely small increase in the sum of weights and the current duality gap. While this step is not essential, it enables us to ensure that for all the essential points in our analysis we are able to assume \emph{exact} centrality, which comes at a negligible expense, but makes all of our proofs much cleaner. This is summarized in the following lemma, whose proof can be found
in Appendix~\ref{sec:lem_correction_proof}:

\begin{lemma}\label{lem:correction}
Suppose that $\nabla F_{\mu}^{\ww}(\xx) = -\CC^\top \gg$ for some vector $\gg$, and 
\[ \energy{\ww,\ss}{\gg} \leq 1/4 \,.\]
Then using $O(\log\log \|\ww\|_1)$ iterations of a vanilla residual correction step, we can obtain a new instance
with weights $\ww'\geq\ww$ and $\mu' \leq \mu(1+\frac{1}{2}\|\ww\|_1^{-11})$ such that 
\[
\nabla F_{\mu'}^{\ww'}(\xx') = \zerov\,.
\]
and $\|\ww'-\ww\|_1 \leq \|\ww\|_1^{-10}$.
\end{lemma}

\subsection{Vanilla Interior Point Method}\label{sec:vanilla_ipm}

At this point we are ready to describe a basic interior point method, which requires $\tO{m^{1/2}}$  iterations. The following lemma shows that once we have a $\mu$-central flow, we can scale down $\mu$ by a significant factor such that the new residual can be routed with low energy. The proof appears in Appendix~\ref{sec:dumb_predictor_proof}.

\begin{lemma}\label{lem:dumb_predictor}
Let $\ff = \ff_0 + \CC\xx$ be a $\mu$-central flow with respect to weights $\ww\in\bbR^{2m}$,
and with slacks $\ss$. Let $\delta = \frac{1}{(2\|\ww\|_1)^{1/2}}$, $\mu' = \mu/(1+\delta)$, and the corresponding  residual $\nabla F_{\mu'}^{\ww}(\xx) = -\CC^\top \hh'$.  Then one has that 
\begin{align*}
\energy{\ww, \ss}{\hh'} \leq 1/4\,.
\end{align*}
\end{lemma}

Together with Corollary~\ref{cor:energy_contract} this enables us to recover a simple analysis of the classical $\tO{m^{1/2}}$ iteration bound. This is shown in the following lemma, whose proof appears in Appendix~\ref{sec:vanilla_mcf_final_proof}.

\begin{lemma}\label{lem:vanilla_mcf_final}
Given a $\mu^0$-central flow with respect to weights $\onev$ and $\mu^0 = m^{O(1)}$,
we can obtain a minimum cost flow solution with duality gap at most $\epsilon = 1 / m^{O(1)}$
using $\tO{m^{1/2}}$ calls to the residual correction procedure (Definition~\ref{def:residual_correction}). 
\end{lemma}

As previously discussed, each iteration of the interior point point method can be implemented in $\widetilde{O}(m)$ time using fast Laplacian solvers. This carries over to an algorithm with a total running time of $\widetilde{O}(m^{3/2}\log W)$, matching that of previous classical algorithms.

As we saw in the proof above the major shortcoming of this method consists of only being able to scale down the parameter $\mu$ by $1+1/\Omega(m^{1/2})$ in every step, the reason being that while advancing from $\mu$ to a smaller $\mu'$ we only rely on Lemma~\ref{lem:dumb_predictor} to certify the fact that the new residual can be corrected.

Instead, in what follows we choose to employ a stronger method to produce a new point which satisfies centrality and has a smaller parameter $\mu$. To do so we employ a more sophisticated procedure for producing the new iterate, which also forces some  more drastic changes in the weights $\ww$.

Finally, we make an important observation concerning the vectors that need to be maintained throughout the algorithm.
\begin{observation}
A linear system solving oracle which returns the vector $\CC\txx$ rather than $\txx$ when solving the linear system described in Equations (\ref{eq:rhodef}), (\ref{eq:linsys})
suffices to execute the algorithm. This is because we never need to maintain the explicit solution $\xx$
but rather only the flow $\ff = \CC \xx$ and corresponding slacks, as
the interior point method works in the space of slacks.
\end{observation}

\section{A Faster Algorithm for Minimum Cost Flow} 
\label{sec:faster_mcf}

Our improved algorithm will be based on the interior point method framework that was developed in Section~\ref{sec:vanilla_mcf}.
The main bottleneck for the running time of that algorithm stems from the fact
that the augmenting circulation
we compute might not allow us to decrease the duality gap by more than a factor of $1+1/\Omega(\sqrt{m})$, as otherwise it is generally impossible to guarantee that the circulation will never congest some edges by more than the available capacity. Hence the iteration bound of $\widetilde{O}(m^{1/2})$, common to standard interior point methods.

We alleviate this difficulty by adding an $\ell_p$ regularization term for the augmenting flow in the objective~(\ref{eq:energy}), similarly to~\cite{liu2019faster}. In~\cite{liu2019faster}, the authors follow the idea of~\cite{madry2016computing} by computing augmenting $s$-$t$ flows. A crucial ingredient is the fact that the congestion of these resulting augmenting flows is then immediately bounded by using a result from~\cite{madry2016computing} which states that as long as there is enough $s$-$t$ residual capacity,
these flows come together with an electrical potential embedding, where no edge is too stretched.

However, this property is specific to the $s$-$t$ maximum flow problem.
To apply a similar argument for the minimum cost flow problem, 
one would need to guarantee that \emph{all} cuts of the graph
have sufficient residual capacity, which is not automatically enforced as in the case of $s$-$t$ max flow.
In order to enforce this cut property,
we further regularize our objective in a different way.
We do this by temporarily superimposing a star on top of our graph, thus obtaining an augmented graph. 
This transformation improves the conductance properties of the graph, ensuring
that there is enough residual capacity in all cuts of the graph.

In Section~\ref{sec:regularized_newton}, we describe the regularized step problem and outline the guarantees of the solution.
In particular, the bias introduced by the regularizers
implies that the augmenting flow is not a circulation anymore, and that we have introduced an additional residual
for our solution in the barrier objective.
We bound the magnitude of both of these perturbations
and ``undo'' them at a later stage.
Finally, we present our electrical stretch guarantee, which serves as the crucial ingredient in both preserving
feasibility and maintaining centrality.

In Section~\ref{sec:parameters} we state our choice of regularization parameters and their consequences.

Even though the electrical stretch guarantee suffices for all purposes if the interior point method 
barrier terms are unweighted, as soon as weights come in the guarantee is affected. In particular,
for any edge whose forward and backward weights are too imbalanced, the electrical stretch and
congestion bounds that we obtain loosen. In Section~\ref{sec:weight_invariants} we deal with this issue
by ensuring that the forward and backward weights for each edge are always relatively balanced, while
introducing an additional demand perturbation.

In Section~\ref{sec:execute} we provide the full view of the algorithm, which consists of 
combining all the ingredients of the previous sections, together with a residual routing 
scheme that includes both vanilla centering steps and constraint re-weighting to obtain an
$\ell_\infty$-based interior point method rather than an $\ell_4$-based one,
as achieved by the vanilla algorithm.

As we mentioned, the solution obtained by the interior point method is for a minimum cost flow problem 
with a slightly perturbed demand. In Section~\ref{sec:repair}, we outline an approach given in \cite{cmsv17}
that given this solution, one can turn it into an optimal solution for the original demand, as long as
the total demand perturbation is small.

\begin{figure}
\frame{
\begin{minipage}{\textwidth}
\vspace{10pt}
\hspace{5pt}
$\textsc{ModifiedAugment}(G,\cc,\epsilon; \ww, \ff, \mu)$
\begin{itemize}
\item Given $\ff$ : $\mu$-central flow with respect to weights $\ww$.
\item Returns $\ff''$ : $\mu'$-central flow with respect to weights $\ww'''$ and with perturbed demand.
\end{itemize}
\begin{enumerate}

\item $\ww, \ff \leftarrow \textsc{BalanceWeights}(G;\ww,\ff).$

\item Augment graph $G$ to $G_\star$.

\item {
Let $\tff$ be a minimizer to (\ref{eq:regularized_newton})
	where 
$h_e = \delta\left(\frac{w_e^{+}}{1-f_e} - \frac{w_e^-}{f_e}\right)$ for each $e\in E$
and $\Delta\hh$ be the residual perturbation. 
}

\item{ Augment flow $\ff' \leftarrow \ff + \tff$. }

\item{ Compute congestion vector $\vrho=(\vrho^+;\vrho^-)$ as $\rho_e^+ = \frac{\tf_e}{1-f_e}$, $\rho_e^- = \frac{-\tf_e}{f_e}$.}

\item{Compute new slacks $(\ss^+)' = \onev - \ff'$ and $(\ss^-)' = \ff'$.}

\item{ Correct residual for congested edges:
\begin{align*}
& (w_e^+)' = \begin{cases}
 w_e^+ + \frac{(s_e^+)'}{(s_e^-)'} \cdot w_e^- (\rho_e^-)^2\,, & \text{if $\left|\rho_e^-\right| \geq C_\infty$\,,}\\
w_e^+\,, & \text{otherwise,}
\end{cases}\\
& (w_e^-)' = 
\begin{cases}
w_e^- + \frac{(s_e^-)'}{(s_e^+)'} \cdot w_e^+ (\rho_e^+)^2\,, & \text{if $\left|\rho_e^+\right| \geq C_\infty$\,,}\\
w_e^-\,, & \text{otherwise.}
\end{cases}
\end{align*}
}

\item{ Correct perturbed residual 
	given by $\hh' = -\left(
			\frac{(\ww^+)'}{(\ss^+)'} - \frac{(\ww^-)'}{(\ss^-)'} + \frac{\cc}{\mu/(1+\delta)}
		+ \Delta\hh\right)$ using 
	Lemma~\ref{lem:restore_centrality_from_low_energy}
	and get new weights $\ww''$, flow $\ff''$, and centrality parameter $\mu'$.}

\item{Compute new slacks $(\ss^+)'' = \onev - \ff''$ and $(\ss^-)'' = \ff''$.}

\item{ Adjust weights to restore exact centrality:
\begin{align*}
& (w_e^+)''' = (w_e^+)'' + \max\left\{0, -(s_e^+)'' \cdot \Delta h_e\right\}\,,\\
& (w_e^-)''' = (w_e^-)'' + \max\left\{0, (s_e^-)'' \cdot \Delta h_e\right\}\,.
\end{align*}
}
\item{ Return $\ww''', \ff'', \mu'$.}

\end{enumerate}
\vspace{5pt}
\end{minipage}
}
\caption{Modified circulation improvement step}
\label{fig:modified_augment}
\end{figure}

\subsection{Regularized Newton Step} \label{sec:regularized_newton}
The initialization procedure from Section~\ref{sec:init} produces a solution with large duality gap, i.e. $O(\mu m)$ where $\mu \leq 2\|\cc\|_2$. Our goal will be to 
reduce this by gradually lowering the parameter $\mu$, while maintaining centrality. 
While in general to achieve this we require solving a sequence of linear systems of equations (as we saw in Section~\ref{sec:vanilla_ipm}), here we choose to solve a slightly perturbed linear system.

In order to do so, we modify the optimization problem from (\ref{eq:energy}) by adding two regularization terms, which will force the produced solution to be well-behaved.
In addition, we allow the newly produced flow $\tff$, which we will use  to update the current solution, to not be a circulation, as long as the demand it routes is small in $\ell_1$ norm. While this breaks the structure of the problem we are solving, it only does so mildly -- therefore once the interior point method has finished running we can repair the broken demand using combinatorial techniques.

\paragraph{Mixed Objective.}
To specify the regularized objective, we first augment the graph $G$ with $O(m)$ extra edges, which are responsible for routing a subset of the flows
that would otherwise force the output of the objective to be too degenerate.

\begin{definition}[Weighted degree]
Given a graph $G(V,E)$ and a weight vector $\ww=(\ww^+;\ww^-)\in\mathbb{R}^{2m}$, the \emph{weighted degree} of $v\in V$ in $G$ with respect to $\ww$ is
defined as 
$d^{\ww}_{v} = \sum_{e \sim v} (w_e^+ + w_e^-)$.
\end{definition}

\begin{definition}
Given a graph 
$G=(V,E)$ 
we define the augmented graph 
$\Gstar = (V \cup \{\vstar\}, \Estar)$, 
where $\Estar = E \cup E'$ and $E'$ 
is obtained by constructing $\lceil  d^{\ww}_v \rceil$ parallel edges
$(v,\vstar)$
for each $v\in V$.

Furthermore, if $\CC$ is a cycle basis for $G$, we let $\CCstar$ be a cycle basis for $\Gstar$ obtained by appending columns to $\CC$, i.e.
\begin{align}\label{eq:Cstardef}
\CCstar = \left[\begin{array}{cc} \CC & \PP_1 \\ \zerov & \PP_2 \end{array}\right]\,.
\end{align}
\end{definition}
We observe that $|E'| = \sum\limits_{v\in V} \lceil d^{\ww}_v\rceil \leq \sum\limits_{v\in V} \left(d^{\ww}_v + 1\right) \leq 3\|\ww\|_1$.
We can now write the regularized objective.

\begin{definition}\label{def:reg_obj}
Given a vector $\hh$, we define the regularized objective as
\begin{align}\label{eq:regularized_newton}
\max_{\tff = \CCstar \txx}\, 
\sum_{e \in E} h_e \cdot \tf_e
&- \frac{1}{2} \sum_{e \in E} (\tf_e)^2\cdot \left( \frac{ w_e^+}{(s_e^+)^2} + \frac{w_e^-}{(s_e^-)^2}\right)\notag \\
&- \frac{\Rstar}{2} \sum_{e \in E'} (\tf_e)^2
- \frac{\Rp}{p} \sum_{e\in E\cup E'} (\tf_e)^p\,,
\end{align}
where
$p > 2$ is an even positive integer, and $\Rstar$, $\Rp$ are some appropriately chosen non-negative scalars.\end{definition}
While this objective might seem difficult to handle, the fact that we are solving a problem on graphs makes it feasible for our purposes. In particular, the works of~\cite{kyng2019flows, adil2020faster} show that this objective can be solved to high precision in time $O(m^{1+o(1)})$, whenever $p$ is sufficiently large. We will make this statement more rigorous, but for simplicity let us for now assume that we can solve (\ref{eq:regularized_newton}) exactly. In Appendix~\ref{sec:solver_error} 
we will show how to handle the solver error.

Let us now understand the effect of the augmenting edges $E'$. Since they allow routing some of the flow through $\vstar$, if we look at the restriction of $\tff$ to the edges of $G$ we see that it stops being a circulation. Let $\tdd$ be the demand routed by the restriction of $\tff$ to $G$. We will see that $\tff$ satisfies optimality conditions for an objective similar to (\ref{eq:regularized_newton}) among all flows that route the demand $\tdd$ in $G$.

Before that, we give a useful lemma that, given a residual $-\CC^\top \hh$, can be used to certify an upper bound on the energy required to route it.
We capture this via the following definition.
\begin{definition}\label{def:emax}
Given a vector $\hh$, weights $\ww$, and slacks $\ss$, we define
\begin{align}
\energymax(\hh,\ww,\ss) = 
\frac{1}{2} \sum_{e \in E} h_e^2 \cdot \left( \frac{ w_e^+}{(s_e^+)^2} + \frac{w_e^-}{(s_e^-)^2}\right)^{-1}\,.
\end{align}
\end{definition}
\begin{lemma}
Given weights $\ww$, slacks $\ss$,
and a residual $-\CC^\top \hh$, we have that \[ \energy{\ww,\ss}{\hh} \leq \energymax(\hh,\ww,\ss)\,. \]
Furthermore, if $\hh = \delta \left(\frac{\ww^+}{\ss^+}-\frac{\ww^-}{\ss^-}\right)$,
we have that 
\[ \energymax(\hh,\ww,\ss) \leq \frac{1}{2}\delta^2 \left\Vert\ww\right\Vert_1\,. \]
\label{lem:emax_upperbound}
\end{lemma}
The proof of this lemma is given in Appendix~\ref{sec:emax_upperbound_proof}.
We are now ready to state the lemma that gives guarantees for the restriction of $\tf$ to $G$.

\begin{lemma}[Optimality in the non-augmented graph]
\label{lem:opt_non_aug}
Let $\tffstar = \CCstar \txxstar$ be the optimizer of the regularized objective from (\ref{eq:regularized_newton}), and let $\tff$ be its restriction to the edges of $G$.
Let $\tdd$ be the demand routed by $\tff$ in $G$.
Then $\tff$ optimizes the objective
\begin{align}\label{eq:mixed_obj}
\max_{\substack{\tff :\\\tff \textnormal{ routes $\tdd$ in G}}}\, \left\langle \hh, \tff \right\rangle - \frac{1}{2} \sum_{e \in E} (\tf_e)^2\cdot \left( \frac{ w_e^+}{(s_e^+)^2} + \frac{w_e^-}{(s_e^-)^2}\right)
- \frac{\Rp}{p} \sum_{e\in E} (\tf_e)^p\,.
\end{align}
Furthermore
\begin{align}\label{eq:first_order_opt_mixed_obj}
\CC^\top  \left( \frac{ \ww^+}{(\ss^+)^2} + \frac{\ww^-}{(\ss^-)^2}  \right) \cdot \tff = \CC^\top \left(\hh + \Delta\hh\right)\,,
\end{align}
where $\Delta\hh = -\Rp (\tff)^{p-1}\,,$
and
\begin{align}
&\|\tdd\|_1 
\leq 
 \left( \frac{6 \|\ww\|_1 \cdot \energymax(\hh,\ww,\ss) }{\Rstar} \right)^{1/2}\,, \\
&\|\tffstar\|_p 
\leq 
\left(\frac{p\cdot \energymax(\hh,\ww,\ss)}{\Rp}\right)^{1/p}\,.
\end{align}
Finally, the energy required to route the perturbed residual can be bounded
by the energy required to route the original residual:
\begin{align}\label{eq:energy_new_residual}
\frac{1}{2} \sum\limits_{e\in E} (\tf_e)^2 \left(\frac{w_e^+}{(s_e^+)^2} + \frac{w_e^-}{(s_e^-)^2}\right)
\leq 
4\cdot \energymax(\hh,\ww,\ss)\,.
\end{align}
\end{lemma}
\begin{proof}
Let $\tffstar = \tff + \tff'$ where $\tff'$ is the restriction of $\tffstar$ to the edges incident to $\vstar$. 
By computing the first order derivative in $\txxstar$, optimality conditions for (\ref{eq:regularized_newton}) imply that for any circulation $\gg$ in $G_\star$ one has that
\begin{align}
\left\langle \gg, \left[\begin{array}{c} \hh  - \tff \cdot \left( \frac{\ww^+}{(\ss^+)^2}+\frac{\ww^-}{(\ss^-)^2} \right) - \Rp \cdot (\tff)^{p-1}  \\ 
- \Rstar \cdot \tff' 
-\Rp  \cdot (\tff')^{p-1} \end{array}\right] \right\rangle = 0\,.
\end{align}

Therefore, restricting ourselves to circulations supported only in the non-preconditioned graph $G$, one has that for any circulation in $\gg' = \CC\zz$ in $G$:
\begin{align}
\left\langle \gg', \hh  - \tff \cdot \left( \frac{\ww^+}{(\ss^+)^2}+\frac{\ww^-}{(\ss^-)^2} \right) - \Rp \cdot (\tff)^{p-1}   \right\rangle = 0\,,
\end{align}
and equivalently
\begin{align}
\left\langle \zz, \CC^\top \left( \hh  - \tff \cdot \left( \frac{\ww^+}{(\ss^+)^2}+\frac{\ww^-}{(\ss^-)^2} \right) - \Rp \cdot (\tff)^{p-1}  \right) \right\rangle = 0\,.
\end{align}
Since this holds for any test vector $\zz$, it must be that the second term in the inner product is $0$. Rearranging, it yields the identity from (\ref{eq:first_order_opt_mixed_obj}).

Now we verify that this is the first order optimality condition for the objective in (\ref{eq:mixed_obj}). Parametrizing the flows that route $\dd$ via $\tff = \CC \xx + \tff_{\dd}$ where $\tff_{\dd}$ is an arbitrary flow which routes $\dd$ in $G$, and setting the derivative with respect to $\xx$ equal to $\zerov$, we obtain exactly (\ref{eq:first_order_opt_mixed_obj}).

Let us proceed to bound the norm of the demand routed by $\tff$.
Consider the value of the objective in (\ref{eq:regularized_newton}) 
after truncating it to only the first two terms, which we can write as:
\begin{align}
&\sum_{e \in E} h_e \cdot \tf_e
- \frac{1}{2} \sum_{e \in E} 
(\tf_e)^2\cdot 
\left( \frac{ w_e^+}{(s_e^+)^2} + \frac{w_e^-}{(s_e^-)^2}\right) \label{eq:trunc}\\
&\leq 
\frac{1}{2} \sum_{e\in E} 
h_e^2 \cdot \left( \frac{ w_e^+}{(s_e^+)^2} + \frac{w_e^-}{(s_e^-)^2}\right)^{-1} \\
&= \energymax(\hh,\ww,\ss) \,, \label{eq:upper_bound_obj}
\end{align}
where we used 
the fact that $\langle \aa,\bb\rangle \leq \frac{1}{2} \|\aa\|^2 + \frac{1}{2}\|\bb\|^2$. 

Note that the value of the regularized objective 
(\ref{eq:regularized_newton})
is at least $0$ since we can always substitute $\txx = \zerov$ and obtain exactly $0$. By re-arranging,\begin{align}\label{eq:rearranged_regularized}
&\frac{\Rstar}{2} \sum_{e \in E'} (\tf_e')^2\\
& \leq 
\sum_{e \in E} h_e \cdot \tf_e
- \frac{1}{2} \sum_{e \in E} (\tf_e)^2\cdot \left( \frac{ w_e^+}{(s_e^+)^2} + \frac{w_e^-}{(s_e^-)^2}\right)
- \frac{\Rp}{p} \sum_{e\in E\cup E'} (\tfstar)_e^p
\label{eq:partial_regularized}\\
& \leq \sum_{e \in E} h_e \cdot \tf_e
- \frac{1}{2} \sum_{e \in E} (\tf_e)^2\cdot \left( \frac{ w_e^+}{(s_e^+)^2} + \frac{w_e^-}{(s_e^-)^2}\right) \\
&\leq \energymax(\hh,\ww,\ss) \,, \label{eq:upper_bound_energy}
\end{align}
where we also used the fact that the last term of (\ref{eq:partial_regularized}) is non-positive
and (\ref{eq:upper_bound_obj}).
Therefore (\ref{eq:upper_bound_energy}) 
enables us to upper bound
\begin{align}
\sum_{e \in E'}(\tf_e')^2\leq
\frac{2}{\Rstar} 
\energymax(\hh,\ww,\ss) \,,
\end{align}
which implies that
\begin{align}
\sum_{e\in E'} \left| \tf_e'\right| 
& \leq|E'|^{1/2}
\cdot
\left( \sum_{e\in E'} 
(\tf_e')^2
\right)^{1/2}\\
&\leq \left( 3\|\ww\|_1 \cdot \frac{2}{\Rstar}
\energymax(\hh,\ww,\ss)
\right)^{1/2}\\
&= \left( \frac{6 \|\ww\|_1 \cdot 
\energymax(\hh,\ww,\ss)}{\Rstar}
\right)^{1/2}\,,
\end{align}
a quantity that upper bounds the demand perturbation.
Using a similar argument we can upper bound $\left\Vert \tffstar\right\Vert_p$. We have
\begin{align}\label{eq:rearranged_regularized_2}
&\frac{\Rp}{p} \sum_{e\in E\cup E'} (\tf_\star)_e^p\\
& \leq 
\sum_{e \in E} h_e \cdot \tf_e
- \frac{1}{2} \sum_{e \in E} (\tf_e)^2\cdot \left( \frac{ w_e^+}{(s_e^+)^2} + \frac{w_e^-}{(s_e^-)^2}\right)
- \frac{\Rstar}{2} \sum_{e \in E'} (\tf_e')^2
\label{eq:partial_regularized_2}\\
& \leq \sum_{e \in E} h_e \cdot \tf_e
- \frac{1}{2} \sum_{e \in E} (\tf_e)^2\cdot \left( \frac{ w_e^+}{(s_e^+)^2} + \frac{w_e^-}{(s_e^-)^2}\right) \\
&\leq \energymax(\hh,\ww,\ss) \,, \label{eq:upper_bound_energy_2}
\end{align}
thus concluding that
\begin{align}
\left\Vert \tffstar\right\Vert_p
= \left(\sum\limits_{e\in E\cup E'} (\tfstar)_e^p\right)^{1/p}
\leq \left(\frac{p\cdot \energymax(\hh,\ww,\ss)}{R_p}\right)^{1/p}\,.\label{eq:fstar_pnorm}
\end{align}
Finally, once more using the same argument, we have that 
\begin{align*}
& \frac{1}{4} \sum_{e \in E} (\tf_e)^2\cdot \left( \frac{ w_e^+}{(s_e^+)^2} + \frac{w_e^-}{(s_e^-)^2}\right)\\
& \leq
\sum_{e \in E} h_e \cdot \tf_e
-\frac{1}{4} \sum_{e \in E} (\tf_e)^2\cdot \left( \frac{ w_e^+}{(s_e^+)^2} + 
	\frac{w_e^-}{(s_e^-)^2}\right)
- \frac{\Rstar}{2} \sum_{e \in E'} (\tf_e')^2
-\frac{\Rp}{p} \sum_{e\in E\cup E'} (\tf_\star)_e^p\\
& \leq
\sum_{e \in E} h_e \cdot \tf_e
-\frac{1}{4} \sum_{e \in E} (\tf_e)^2\cdot \left( \frac{ w_e^+}{(s_e^+)^2} + 
	\frac{w_e^-}{(s_e^-)^2}\right)\\
& \leq 2\cdot \energymax(\hh,\ww,\ss)\,,
\end{align*}
therefore
\begin{align*}
\frac{1}{2} \sum_{e \in E} (\tf_e)^2\cdot \left( \frac{ w_e^+}{(s_e^+)^2} + \frac{w_e^-}{(s_e^-)^2}\right)
\leq 4 \cdot \energymax(\hh,\ww,\ss)\,.
\end{align*}
\end{proof}

Finally, we present an important property of the solution of the regularized Newton step, which will be crucial for obtaining the final result.
\begin{lemma}
\label{lem:prec_lemma}
Let $\tffstar$ be the solution of the regularized objective (\ref{eq:regularized_newton}) and $\tff$ its restriction on $G$, and suppose that $\|\ww\|_1 \geq 3$. Then one has that over the edges $e\in E$:
\begin{equation}
\begin{aligned}
\left\vert \left( 
\frac{w^+_e}{(s^+_e)^2}+\frac{w^-_e}{(s^-_e)^2}  
+
\Rp \cdot \tf_e^{p-2}
\right) \tf_e - h_e \right\vert
&\leq \hat{\gamma} \,,
\end{aligned}
\label{eq:precond_guarantee}
\end{equation}
where 
\begin{align*}
\hat{\gamma} = \left(\Rstar
+\Rp \cdot \|\tff_\star\|_\infty^{p-2}\right)^{1/2} 
\cdot \left\| \frac{ \hh }{  \sqrt{(\ww^+ + \ww^-)\left( \frac{\ww^+}{(\ss^+)^2}+\frac{\ww^-}{(\ss^-)^2} \right) }} \right\|_\infty \cdot {32 \log \|\ww\|_1} \,.
\end{align*}
Furthermore, this implies that
\begin{align}
\left(\frac{w_e^+}{(s_e^+)^2} + \frac{w_e^-}{(s_e^-)^2}\right) \cdot \left\vert\tf_e  \right\vert
&\leq 
\left\vert h_e \right \vert + \hat{\gamma}\,.
\label{eq:precond_guarantee2}
\end{align}
\end{lemma}
Since the proof is technical, we defer it to Section~\ref{sec:prec_proof}.

\subsection{Choice of Regularization Parameters} 
\label{sec:parameters}

Now we state our choice of regularization parameters $R_p$, $R_\star$ and their consequences. In particular, they affect the $\ell_\infty$ norm of the flow $\tff$ we obtain by solving the regularized objective, and the preconditioning guarantee (Lemma~\ref{lem:prec_lemma}) which will be essential to obtain the correct trade-off between iteration complexity and demand perturbation.

\begin{definition} \label{def:reg_parameters}
Given a weight vector $\ww$, we will use the following values for the parameters $p$, $R_p$ and $R_\star$:
\begin{align*}
p &= \min\left\{ k \in 2\bbZ : k \geq (\log m)^{1/3} \right\} \,,\\
R_p &= p\cdot \left(10^6\cdot \delta^2 \|\ww\|_1
 \cdot {\log\|\ww\|_1}\right)^{p+1}\,,\\
R_\star &= 3\cdot \delta^2\|\ww\|_1^2\,.
\end{align*}
\end{definition}
This immediately yields the following useful corollaries.

\begin{corollary}\label{cor:step_running_time}
For our specific choice of $p$, the regularized objective from (\ref{eq:regularized_newton}) can be solved to 
high precision in $m^{1+o(1)}$ time. We can, furthermore, assume an exact solution.
\end{corollary}
\begin{proof}
We use Theorem~\ref{thm:mixedsolver} which was proved in~\cite{kyng2019flows,adil2020faster}. The objective in (\ref{eq:regularized_newton}) matches exactly the type handled there.  For our choice of $p$, the running time is 
\[
2^{O(p^{3/2})} m^{1+O(1/\sqrt{p})}
= 2^{O{\sqrt{\log n}}} m^{1+1/(\log^{1/6} n)} = m^{1+o(1)}\,.
\]
Furthermore, the resulting solution has a quasi-polynomially small error $2^{-(\log m)^{O(1)}}$, which can be neglected in our analysis, per the discussion in Section~\ref{sec:solver_error}.
\end{proof}

The proofs of the following two corollaries appear in Appendix~\ref{sec:cor_l_p_upperbound_proof} and~\ref{sec:cor_gamma_proof}, respectively.
\begin{corollary} \label{cor:l_p_upperbound}
Let $\tff_\star$ be the solution obtained by solving (\ref{eq:regularized_newton}) and $\tff$ be its restriction to $G$.
For our specific choice of regularization parameters
we have
\[
\|\tffstar\|_p \leq \frac{1}{ 10^6\cdot\delta^2 \|\ww\|_1 \cdot {\log\|\ww\|_1}}\,.
\]
\end{corollary}
\begin{corollary} \label{cor:gamma}
For the choice of $\hh = \delta\left(\frac{\ww^+}{\ss^+} - \frac{\ww^-}{\ss^-}\right)$,
and as long as 
$\delta \leq \|\ww\|_1^{-(1/4+o(1))}$
the $\hat{\gamma}$ in Lemma~\ref{lem:prec_lemma} can be upper bounded by 
\begin{align*}
\gamma = \delta^2 \|\ww\|_1 \cdot 32\sqrt{6} \cdot {\log\|\ww\|_1}\,.
\end{align*}
\end{corollary}

 \subsection{Weight Invariants} \label{sec:weight_invariants}
\paragraph{Total weight invariant.}
Our interior point method will inherently increase edge weights. However, the total increase has to remain bounded by $O(m)$
in order to be able to guarantee an upper bound on the energy $\energy{\ww,\ss}{\hh}$. We state the following
invariant that we intend to always enforce, in order to ensure that this is the case:
\begin{invariant}
The sum of weights is bounded: $\left\Vert \ww\right\Vert_1 \leq 3m$\,.
\label{inv:sum_weights}
\end{invariant}

\paragraph{Weight balancing.}
Before proceeding with the description of the method, we define a notion that will be used by the algorithm to deal with severe 
edge weight imbalances. 
Such imbalances can limit the usability of Lemma~\ref{lem:prec_lemma} and lead to the augmenting flow $\tff$ being infeasible
or expensive to correct.

\begin{definition}[Balanced edges]
\label{def:balanced}
An edge $e\in E$ is called \emph{balanced} 
if $\max\left\{w_e^+,w_e^-\right\} \leq \delta \left\Vert \ww\right\Vert_1$ or $\min\left\{w_e^+,w_e^-\right\} \geq 96 \cdot \delta^4 \left\Vert \ww\right\Vert_1^2$.
Otherwise it is called \emph{imbalanced}.
\label{def:imbalance}
\end{definition}
Even though imbalanced edges can generally emerge, we are able to balance them by manually reducing the disparity between $w_e^+$ and $w_e^-$, while breaking the flow demand
by a controllable amount.
We will maintain the following invariant right before solving the regularized objective:
\begin{invariant}
All edges are balanced.
\label{inv:imbalance}
\end{invariant}
In Section~\ref{sec:execute} we will see that there is a way to enforce 
Invariant~\ref{inv:imbalance} while only increasing the weight and breaking the demand by a small amount. 
The main motivation behind keeping edges balanced is that it implies that each edge either has a favorable stretch property,
or is not too congested:

\begin{lemma}
Let $\tff$ be the restriction of the regularized problem (\ref{eq:regularized_newton}) solution
to $G$,
with $\hh = \delta \left(\frac{\ww^+}{\ss^+}-\frac{\ww^-}{\ss^-}\right)$
and congestion $\vrho$, where $\delta \leq \|\ww\|_1^{-(1/4+o(1))}$. 
If $e\in E$ is balanced, then
$\max\left\{\left|\rho_e^+\right|,\left|\rho_e^-\right|\right\} < C_\infty$ or 
$\left|\frac{w_e^+\rho_e^+}{s_e^+}\right|+\left|\frac{w_e^-\rho_e^-}{s_e^-}\right| \leq 6\gamma$,
where $C_\infty = \frac{1}{2\delta\sqrt{2\|\ww\|_1}}$.
\label{lem:imb_cong_bound}
\end{lemma}
The proof of this lemma appears in Appendix~\ref{sec:lem_imb_cong_bound_proof}.
 \subsection{Executing the Interior Point Method}\label{sec:execute}
Having defined the regularized objective, we now show how to execute the interior point method using the solution returned by a high precision solver.
Since the solution to this objective does not exactly match the 
guarantee required in (\ref{eq:rhodef}-\ref{eq:linsys}), we will have to do some slight manual adjustments.

In the vanilla interior point method analysis that we saw earlier, we witnessed a very stringent requirement on the condition that we are able to correct a residual. Namely, we required that the energy required to route it decreases in every iteration of the correction step, which was guaranteed by the fact that after performing the first correction step the upper bound on energy $\sum w_i \rho_i^4$ is at most a small constant (i.e. $1/4$). 

This requirement is too strong since, as a matter of fact, the most important obstacle handled by interior point methods is preserving slack feasibility. In our specific context this means that we want to perform updates to the current flow without violating capacity constraints, which is guaranteed by a weaker $\ell_\infty$ bound, i.e. $\|\vrho\|_\infty \leq 1/2$. While this condition is sufficient to preserve feasibility, it is not clear that after performing the corresponding update to the flow, the energy required to route the residual will be small, so the resulting residual can be reduced to $\zerov$. Instead we can enforce this property by canceling the components of the gradient which cause this energy to be large. 

\begin{definition}[Perturbed residual correction]
\label{def:pert_res_cor}
Consider a flow $\ff$ with the corresponding slack vector $\ss > \zerov$,  weights $\ww$ and parameter $\mu > 0$, with a corresponding residual $\nabla F_{\mu}^{\ww}(\xx) = -\CC^\top \hh$
where $\hh = \delta \left(\frac{\ww^+}{\ss^+} - \frac{\ww^-}{\ss^-}\right)$ and $\delta\leq \|\ww\|_1^{-1/4}/2$.
The perturbed residual correction step is defined as an update to $\ff$ via:
\begin{align}
\ff' &= \ff + \tff\,, \\
(\ss^-)' &= \ss^- + \tff\,, \label{eq:s1_update}\\
(\ss^+)' &= \ss^+ - \tff\,, \label{eq:s2_update}
\end{align}
where $\tff$
 is the solution to the linear system
\begin{align}
\vrho^+ &= \frac{\tff}{\ss^+} \,,\\
\vrho^- &= \frac{-\tff}{\ss^-}\,, \\
\CC^\top \left( \frac{\ww^+ \vrho^+}{\ss^+} - \frac{\ww^-\vrho^-}{\ss^-} \right) &= \CC^\top (\hh + \Delta \hh)\,,
\end{align}
 such that
\begin{align}
\left\|\vrho\right\|_\infty \leq \frac{1}{2}\,,
\end{align}
for some  perturbation $\Delta \hh$,
 followed by the updates to the $\ww$ vector via:
\begin{align}\label{eq:w1}
(w^+_e)' = \begin{cases}
w_e^+ +  \frac{(s_e^+)'}{(s_e^-)'}\cdot w_e^- (\rho_e^-)^2 &\textnormal{ if $\vert\rho^-_e\vert \geq \Cinf$}\,, \\
w_e^+ &\textnormal{ otherwise, }
\end{cases}
\end{align}
and
\begin{align}\label{eq:w2}
(w^-_e)' = \begin{cases}
w_e^- +  \frac{(s^-_e)' }{(s^+_e)'}\cdot w_e^+ (\rho_e^+)^2&\textnormal{ if $\vert\rho^+_e\vert \geq  \Cinf$}\,, \\
w_e^- &\textnormal{ otherwise, }
\end{cases}
\end{align}
where $$\Cinf = \frac{1}{2\delta \sqrt{2\left\Vert\ww\right\Vert_1}}\,.$$\end{definition}

We also give a short lemma, which upper bounds the amount by which $\|\ww\|_1$ gets increased when applying the perturbed residual correction.

\begin{lemma}
\label{lem:local_weight_changes_for_small_energy}
Let $\ww$ and $\ww'$ be the old and new weights, respectively, as described in Definition~\ref{def:pert_res_cor}. 
Then one has that
\[
\|\ww'-\ww\|_1 \leq 192\sqrt{2}\cdot \gamma \cdot 
\left(\delta^2 \|\ww\|_1\right)^{3/2}\,.
\]
\end{lemma}
\begin{proof}
Let $S^+ = \left\{e\in E\ :\ \left|\rho_e^+\right| \geq \Cinf\right\}$
and $S^- = \left\{e\in E\ :\ \left|\rho_e^-\right| \geq \Cinf\right\}$.
The total weight increase is equal to 
\[ \sum\limits_{e\in S^+} \left((w_e^-)'-w_e^-\right) + \sum\limits_{e\in S^-} \left((w_e^+)'-w_e^+\right)\,. \]
Let us upper bound the weight increase contributed by a single edge $e\in S^+$. The respective bound will follow for an edge in $S^-$ by symmetry.
We have
\begin{align*}
(w_e^-)' - w_e^- = (s_e^-)' \cdot \frac{w_e^+ (\rho_e^+)^2}{(s_e^+)'}
&\leq  4 \cdot s_e^- \cdot \frac{w_e^+ (\rho_e^+)^2}{s_e^+}
= 4 \cdot s_e^- \cdot\vert\rho_e^+ \vert \cdot \left|\frac{w_e^+\rho_e^+}{s_e^+}\right|\,.
\end{align*}
By Lemma~\ref{eq:precond_guarantee_balancing}, all the edges in $S^+\cup S^-$ have the low stretch property:
\begin{align*}
\left|\frac{w_e^+\rho_e^+}{s_e^+}\right| + \left|\frac{w_e^-\rho_e^-}{s_e^-}\right| \leq 6\gamma\,,
\end{align*}
and so
\[
(w_e^-)' - w_e^- \leq 4 \cdot s_e^-\cdot \vert\rho_e^+ \vert \cdot \left|\frac{w_e^+\rho_e^+}{s_e^+}\right| \leq 24\gamma \vert\rho_e^+ \vert\,.
\]
The total weight increase due to $S^+$ is thus:
\begin{align*}
& \sum\limits_{e\in S^+} \left((w_e^-)'-w_e^-\right) \leq 24\gamma \sum\limits_{e\in S^+} \left|\rho_e^+\right|\,.
\end{align*}
Symmetrically for $S^-$ we get
\begin{align*}
& \sum\limits_{e\in S^-} \left((w_e^+)'-w_e^+\right) \leq 24\gamma \sum\limits_{e\in S^-} \left|\rho_e^-\right|\,,
\end{align*}
and so 
\begin{align*}
\|\ww'-\ww\|_1 
\leq 24\gamma \cdot\left\Vert\vrho_{S^+\cup S^-}\right\Vert_1\,.
\end{align*}
Now, since by Lemma~\ref{lem:opt_non_aug} the energy to route the perturbed residual is bounded
by the energy to route the original residual,
\begin{align*}
\left\Vert\rho\right\Vert_2^2
\leq \sum\limits_{e\in E} \left(w_e^+(\rho_e^+)^2+w_e^-(\rho_e^-)^2\right)
= \sum\limits_{e\in E} (\tf_e)^2 \left(\frac{w_e^+}{(s_e^+)^2} + \frac{w_e^-}{(s_e^-)^2}\right) \leq 
8\cdot\energymax(\hh,\ww,\ss)\,,
\end{align*}
we have that 
\begin{align*}
\left\Vert \vrho_{S^+\cup S^-}\right\Vert_1
& \leq 
\left\Vert\vrho\right\Vert_2^2\cdot\max\left\{\left\Vert 1/\vrho_{S^+}^+\right\Vert_\infty,\left\Vert 1/\vrho_{S^-}^-\right\Vert_\infty\right\} \\
	&\leq 
8\cdot \energymax(\hh,\ww,\ss)
\cdot C_\infty^{-1}\\
& \leq 4\cdot \delta^2 \|\ww\|_1 \cdot 2\cdot \delta\sqrt{\|\ww\|_1}\\
& = 8 \sqrt{2} \cdot \left(\delta^2 \|\ww\|_1\right)^{3/2} \,.
\end{align*}
Thus we conclude that 
\begin{align*}
\|\ww' - \ww\|_1 \leq
192\sqrt{2}\cdot \gamma \cdot 
\left(\delta^2 \|\ww\|_1\right)^{3/2}\,.
\end{align*}
\end{proof}

The effect of the weight updates is to zero out the coordinates in the new residual that contribute a lot to energy. We formalize this intuition in the lemma below, whose proof appears in Appendix~\ref{sec:lem_improved_correction_proof}.
\begin{lemma}
\label{lem:improved_correction}
Suppose that we perform a perturbed residual correction step 
as described in Definition~\ref{def:pert_res_cor}
and obtain a solution $\xx'$ with weights $\ww'$, residual $\nabla F_\mu^{\ww'}(\xx') = -\CC^\top \gg$, and let $\Delta\hh$ be the perturbation of the residual.
Then we have that
\begin{align}
\energy{\ww',\ss'}{\gg + \Delta \hh} \leq 1/4\,.
\end{align}
\end{lemma}

The next short lemma is a straightforward application of the vanilla correction and perfect correction routines, which shows that while only very slightly perturbing weights, we can obtain an instance where $\nabla F_\mu^{\ww''}(\xx'') - \CC^\top \Delta \hh = 0$.

\begin{lemma}
\label{lem:restore_centrality_from_low_energy}
Suppose that $\nabla F_\mu^{\ww'}(\xx') = -\CC^\top \gg$ for some vector $\gg$, and
\[
\energy{\ww',\ss'}{\gg + \Delta \hh} \leq 1/4\,.
\]
Then using $O(\log \log m)$ iterations of a vanilla residual correction step, we can obtain a new instance with weights $\ww''$ 
and $\mu' \leq \mu(1+\frac{1}{2}\|\ww\|_1^{-11})$ 
such that
\[
\nabla F_{\mu'}^{\ww''}(\xx'') = \CC^\top \Delta \hh\,.
\]
and $\|\ww''-\ww'\|_1 \leq \|\ww'\|_1^{-10}$.
\end{lemma}
\begin{proof}
We apply Lemma~\ref{lem:correction} for 
the perturbed function $\hat{F}_{\mu}^{\ww'}(\xx) = F_{\mu}^{\ww'}(\xx) - \langle \Delta\hh, \CC \xx\rangle$
to obtain a new solution $\ff'' = \ff_0 + \CC \xx''$ 
and a new set of weights $\ww''\geq \ww'$
such that $\ww'' \leq \ww'(1-m^{-10})$
and $\nabla F_{\mu'}^{\ww''}(\xx'') - \CC^\top \Delta \hh = \nabla \hat{F}_{\mu'}^{\ww''}(\xx'') = \zerov$.
\end{proof}

\begin{lemma}
\label{lem:perturbation_correction}
Suppose we have an instance where $\nabla F_\mu^{\ww}(\xx) = \CC^\top \Delta\hh$. Then there exists a set of weights $\ww' \geq \ww$ such that $\nabla F_\mu^{\ww'}(\xx) = \zerov$ and $\|\ww'-\ww\|_1 \leq \|\Delta\hh\|_1$.
\end{lemma}
\begin{proof}
By definition we have
\begin{align}
\frac{\CC^\top\cc}{\mu} + \CC^\top\left( \frac{\ww^+}{\ss^+} - \frac{\ww^-}{\ss^-} \right) = \CC^\top \Delta\hh\,.
\end{align}
We can restore exact centrality by perturbing the weights via the simple update:
\begin{align}
(\ww^+)' &= \ww^+ - \ss^+ \cdot (\Delta\hh)_{\leq 0}\,, \\
(\ww^-)' &= \ww^- + \ss^- \cdot (\Delta\hh)_{\geq 0}\,.
\end{align}
Plugging into the above identity immediately yields the desired condition: 
\begin{align*}
\nabla F_\mu^{\ww'}(\xx)& = \frac{\CC^\top\cc}{\mu} + \CC^\top\left( \frac{(\ww^+)'}{\ss^+} - \frac{(\ww^-)'}{\ss^-} \right) \\
& = 
\frac{\CC^\top\cc}{\mu} + \CC^\top\left( \frac{\ww^+}{\ss^+} - \frac{\ww^-}{\ss^-} \right) +
\CC^\top\left( \frac{(\ww^+)'-\ww^+}{\ss^+} - \frac{(\ww^-)'-\ww^-}{\ss^-} \right) \\
& = 
\CC^\top \Delta\hh-
\CC^\top \Delta\hh\\
& = \zerov\,.
\end{align*}
Finally, we bound 
$$\|\ww'-\ww\|_1 = -\langle \ss^+, (\Delta\hh)_{\leq 0}\rangle + \langle \ss^-, (\Delta\hh)_{\geq 0}\rangle \leq \|\ss\|_\infty \cdot \|\Delta\hh\|_1 \leq \|\Delta\hh\|_1\,. $$
For the final inequality we crucially used the fact that all slacks are at most $1$.
\end{proof}

Combining Lemmas~\ref{lem:improved_correction} and~\ref{lem:perturbation_correction} 
together with the vanilla correction step (Corollary~\ref{cor:energy_contract} and Lemma~\ref{lem:fine_correction}) we can derive an improved correction step
based on the solution to the regularized objective.
First we show that indeed this is possible, i.e. for a particular choice of regularization parameters we obtain a step with $\|\vrho\|_\infty \leq 1/2$.
\begin{lemma}[Feasibility lemma]
\label{lem:advance_feasibility}
Suppose we have an instance with weights $\ww$ and slacks $\ss$, residual $-\CC^\top \hh$ where $\hh = \delta\left(\frac{\ww^+}{\ss^+} - \frac{\ww^-}{\ss^-}\right)$, and 
$\|\ww\|_1^{-1/2} \leq \delta \leq \|\ww\|_1^{-(1/4+o(1))}$.
Then
by solving the regularized objective 
we obtain a 
flow $\tff$ satisfying 
\begin{align*}
\vrho^+ &= \frac{\tff}{\ss^+}\,, \\
\vrho^- &= \frac{-\tff}{\ss^-}\,, \\
\CC^\top \left( \frac{\ww^+ \vrho^+}{\ss^+} - \frac{\ww^-\vrho^-}{\ss^-} \right) &= \CC^\top (\hh + \Delta \hh)\,.
\end{align*}
such that:
\begin{enumerate}
\item the congestion satisfies 
$$\|\vrho\|_\infty \leq 1/2\,,$$
\item the perturbation $\Delta \hh$ is bounded
$$\|\Delta \hh\|_1 \leq p \cdot \|\ww\|_1^{1/p} \cdot (10^6\cdot \delta^2 \|\ww\|_1 \log\|\ww\|_1)^2\,,$$ 
\item the flow $\tff$ routes a demand $\tdd$ such that
$$\|\tdd\|_1 \leq 1\,.$$
\end{enumerate}
\end{lemma}
\begin{proof}
~\paragraph{1.} 
Let us first verify the crucial feature of $\tff$, namely that it ensures a low congestion $\|\vrho\|_\infty \leq 1/2$. For our specific choice of $\hh$ we can upper bound, by 
applying Lemma~\ref{lem:emax_upperbound}:
\[
\energymax(\hh,\ww,\ss) 
\leq \frac{1}{2} \delta^2 \|\ww\|_1\,.
\]
First we note that for this choice of parameters we can upper bound, using Corollary~\ref{cor:l_p_upperbound}:
\begin{align}
\label{eq:fpnorm}
\|\tff\|_\infty
\leq \|\tff\|_p
\leq \frac{1}{ 10^6\cdot \delta^2 \|\ww\|_1\cdot{\log\|\ww\|_1}}
\,. 
\end{align}
By Lemma~\ref{lem:imb_cong_bound}, for each edge $e\in E$
we either have 
\begin{align*}
\max\left\{\left|\rho_e^+\right|,\left|\rho_e^-\right|\right\} < C_\infty = \frac{1}{2\delta\sqrt{2\left\Vert\ww\right\Vert_1}} \leq 1/2
\end{align*}
by our assumption on $\delta$,
in which case we are done,
or the low stretch property
\begin{align*}
\left|\frac{w_e^+\rho_e^+}{s_e^+}\right| + \left|\frac{w_e^-\rho_e^-}{s_e^-}\right| \leq 6\gamma
\end{align*}
holds, where by Corollary~\ref{cor:gamma}
\begin{align*}
\gamma 
& = 
\delta^2 \|\ww\|_1 \cdot 32\sqrt{6}\cdot {\log \|\ww\|_1}\,.
\end{align*}
This implies that
\begin{align*}
\left|\rho_e^+\right| =
\left|\rho_e^+s_e^+\cdot \frac{\rho_e^+}{s_e^+}\right|^{1/2} =
\left(\left|\tf_e\right|\cdot \left|\frac{\rho_e^+}{s_e^+}\right|\right)^{1/2} \leq
\left(\frac{1}{10^6 \cdot \delta^2 \left\Vert\ww\right\Vert_1 \cdot{\log\left\Vert\ww\right\Vert_1}}\cdot 6\gamma\right)^{1/2}
< \frac{1}{2}\,.
\end{align*}

The argument for $\rho_e^-$ is symmetric. Therefore we have that $\|\vrho\|_\infty \leq 1/2$.

\paragraph{2.} 
Now let us upper bound 
the $\ell_1$ norm of the perturbation $\Delta \hh$.
Per Lemma~\ref{lem:opt_non_aug}, our step produced by solving the regularized objective yields
\[
\Delta \hh = -\Rp (\tff)^{p-1}\,.
\]
Therefore we can control
\[
\Vert \Delta \hh \Vert_1 = \Rp \|\tff\|_{p-1}^{p-1}
\leq \Rp \left( \|\ww\|_1^{\frac{1}{p-1}-\frac{1}{p}} \cdot \|\tff\|_p \right)^{p-1}
= \Rp \|\ww\|_1^{1/p} \cdot \|\tff\|_p^{p-1}\,.
\]
Using (\ref{eq:fpnorm}) we further upper bound this by
\begin{align*}
\|\Delta\hh\|_1 
& \leq 
{p \cdot (10^6\cdot\delta^2 \|\ww\|_1\cdot{\log\|\ww\|_1})^{p+1}} 
\cdot 
\|\ww\|_1^{1/p} 
\cdot
\left( \frac{1}{10^6\cdot \delta^2 \|\ww\|_1\cdot{\log\|\ww\|_1}} \right)^{p-1}\\
& \leq 
p \cdot \|\ww\|_1^{1/p} \cdot \left(10^6\cdot\delta^2 \|\ww\|_1\cdot{\log\|\ww\|_1}\right)^2\,.
\end{align*}

\paragraph{3.} 
Finally, let us upper bound the demand routed by $\tff$.
For the demand perturbation, 
by Lemma~\ref{lem:opt_non_aug} one has that after optimizing the regularized objective the resulting flow $\tff$ routes
a demand $\tdd$ such that 
\[
\|\tdd\|_1 \leq  \left( \frac{6 \|\ww\|_1 \cdot \energymax(\hh,\ww,\ss) }{\Rstar} \right)^{1/2} 
\leq \left(\frac{3 \delta^2 \|\ww\|_1^2  }{\Rstar}\right)^{1/2}
= 1\,,
\]
which gives us what we needed.

\end{proof}

\begin{definition}[Weight balancing procedure]
\label{def:balancing_proc}
Given a flow $\ff$ that is $\mu$-central with respect to weights $\ww$ and with slacks $\ss$, let $S\subseteq E$ be the set
of edges that are not balanced (Invariant~\ref{inv:imbalance}). The weight balancing procedure consists of computing new weights $\ww'$
such that 
\begin{itemize}
\item{ For each $e\in S$: 
If $w_e^+\leq w_e^-$ then 
$w_e^{'+}=96\cdot \delta^4 \left\Vert \ww\right\Vert_1^2$,
$w_e^{'-}=w_e^-$,
while if $w_e^+ > w_e^-$ then
$w_e^{'+}=w_e^+$,
$w_e^{'-}=96\cdot \delta^4 \left\Vert \ww\right\Vert_1^2$.}
\item{For each $e\notin S$ we set $w_e^{'+}=w_e^+$, $w_e^{'-}=w_e^-$.}
\end{itemize}
Additionally, we compute a flow $\ff'$ with slacks $\ss' > \zerov$ such that 
\[ \frac{\ww^{'+}}{\ss^{'+}} - \frac{\ww^{'-}}{\ss^{'-}} = \frac{\ww^{+}}{\ss^{+}} - \frac{\ww^{-}}{\ss^{-}}\,. \]
\label{def:weight_balancing}
\end{definition}
\begin{lemma}[Weight balancing lemma]
\label{lem:weight_bal}
Given $\ff, \ww, \ss$, $S$ and $\mu$ as in Definition~\ref{def:weight_balancing} after applying the weight
balancing procedure we get a $\mu$-central flow $\ff'$ with respect to weights $\ww'\geq \ww$ and with slacks $\ss'$
such that all edges satisfy Invariant~\ref{inv:imbalance} with respect to $\ww'$ and $\ss'$. Additionally
\[ \left\Vert \ww' - \ww\right\Vert_1 \leq 96\cdot |S| \cdot \delta^4 \left\Vert\ww\right\Vert_1^2\,,\]
and
\[ \left\Vert \dd' - \dd\right\Vert_1 \leq |S|\,, \]
and $\left\Vert d'-d\right\Vert_1 \leq |S|$
where $\dd$ is the demand routed by $\ff$ and $\dd'$ the demand routed by $\ff'$.
\label{lemma:weight_balancing}
\end{lemma}
\begin{proof}
Let $e\in S$ and without loss of generality $w_e^+ < 96\cdot \delta^4 \left\Vert\ww\right\Vert_1^2$ and $w_e^- > \delta \left\Vert\ww\right\Vert_1$. 
First of all, we have $w_e^{'+} = 96\cdot \delta^4 \left\Vert \ww\right\Vert_1^2 > w_e^+$ and the weight increase $w_e^{'+}-w_e^+$ is at most $96\cdot \delta^4 \left\Vert\ww\right\Vert_1^2$.
Furthermore, let us look at the function $g(t) = \frac{w_e^{'+}}{1-t} - \frac{w_e^{'-}}{t}$ for $t\in (0,1)$.
This is a continuous increasing function with $\underset{t\rightarrow 0}{\lim}\, g(t) = -\infty$ and $\underset{t\rightarrow 1}{\lim}\, g(t)= \infty$. Therefore
there exists a unique $f_e'\in(0,1)$ such that
\begin{align*}
\frac{w_e^{'+}}{1-f_e'} - \frac{w_e^{'-}}{f_e'}= \frac{w_e^{+}}{1-f_e} - \frac{w_e^{-}}{f_e}\,.
\end{align*}
As $f_e,f_e'\in(0,1)$, the demand perturbation $\left|f_e'-f_e\right|$ on this edge is at most $1$.
Putting everything together, we get that $\left\Vert \ww'-\ww\right\Vert_1 \leq 96\cdot |S| \cdot \delta^4 \left\Vert\ww\right\Vert_1^2$
and $\left\Vert \dd'-\dd\right\Vert_1 \leq |S|$.
\end{proof}

\begin{lemma}[Progress lemma]
\label{lem:advance}
Given a central instance with parameter $\mu$ and weights $\ww$, we can obtain a new central instance with parameter $\mu/(1+\delta)$ and weights $\ww'''+\Delta\ww\geq \ww$ 
with 
\[
\delta \geq m^{-1/4}/10\,,
\]
such that
\begin{align*}
\|\ww'''-\ww\|_1 & \leq \|\Delta\ww\|_1 \\
& + \left(\delta^2 \|\ww+\Delta\ww\|_1\right)^{5/2} \cdot 6 \cdot 10^4 \cdot {\log \|\ww+\Delta\ww\|_1}\\
& + m^{-10}\\
& + \left(10^6 \cdot \delta^2 \|\ww+\Delta\ww\|_1\cdot \log \|\ww+\Delta\ww\|_1\right)^2 \cdot p \cdot\|\ww+\Delta\ww\|_1^{1/p}\,.
\end{align*}
where $\Delta\ww\geq \zerov$ is the weight increase caused by applying the procedure described in Definition~\ref{def:weight_balancing} on weights $\ww$.
Furthermore, the demand perturbation is $\tdd + \Delta\tdd$, where
\[
\|\tdd\|_1 \leq 1\,.
\]
and $\Delta\tdd$ is the demand perturbation caused by applying the procedure described in Definition~\ref{def:weight_balancing} on weights $\ww$.
\end{lemma}
\begin{proof}
We first apply the weight balancing procedure as described in Definition~\ref{def:weight_balancing} to get new weights $\ww+\Delta\ww \geq \ww$.

We will apply the residual correction step with the flow $\tff$ guaranteed by Lemma~\ref{lem:advance_feasibility}, which guarantees that the corresponding flow
yields a congestion $\|\vrho\|_\infty \leq 1/2$.
As this step is obtained for a perturbed residual $\nabla F_{\mu}^{\ww+\Delta\ww}(\xx) - \CC^\top \hh = -\CC^\top \left(\hh + \Delta \hh\right)$,
 it means that after executing it
we obtain an instance with weights $\ww'$ and slacks $\ss'$ such that
\[
\energy{\ww',\ss'}{\hh' + \Delta \hh} \leq 1/4\,,
\]
where $\nabla F_{\mu}^{\ww'}(\xx) = -\CC^\top \hh'$.
This follows from applying Lemma~\ref{lem:improved_correction}.
Therefore, using Lemma~\ref{lem:restore_centrality_from_low_energy} we
can obtain a new point where
\[
\nabla F_\mu^{\ww''}(\xx'') = \CC^\top \Delta \hh\,,
\]
for a slight change in weights from $\ww'$ to $\ww''$.
Finally, applying Lemma~\ref{lem:perturbation_correction} we can establish exact centrality
\[
\nabla F_\mu^{\ww'''}(\xx'') = \zerov\,,
\]
while slightly increasing weights from $\ww''$ to $\ww'''$.

Having described the method, let us first bound the weight change it causes.
By Lemma~\ref{lem:local_weight_changes_for_small_energy} we have that 
the weight increase $\|\ww' - \left(\ww+\Delta\ww\right)\|_1$ caused by performing the perturbations 
described in Definition~\ref{def:pert_res_cor} is upper bounded by 
\[
192\sqrt{2} \cdot \gamma  \cdot (\delta^2 \|\ww+\Delta\ww\|_1)^{3/2}\,,
\]
where 
\[
\gamma =
\delta^2 \|\ww+\Delta\ww\|_1\cdot 32 \sqrt{6} \cdot {\log\|\ww+\Delta\ww\|_1}\,.
\]
Furthermore, restoring the condition that $\nabla F_\mu^{\ww''}(\xx'') = \CC^\top \Delta \hh$ is done while increasing the $\ell_1$ norm of the weight vector by $\|\ww''-\ww'\|_1\leq m^{-10}$
and restoring exact centrality via Lemma~\ref{lem:opt_non_aug} 
costs us a further weight increase $\|\ww'''-\ww''\|_1$ that is upper bounded by
$$\|\Delta \hh\|_1 \leq p \cdot \|\ww+\Delta\ww\|_1^{1/p} 
\cdot (10^6\cdot \delta^2 \|\ww+\Delta\ww\|_1 \log \|\ww+\Delta\ww\|_1)^2\,,$$
where we used the guarantee from Lemma~\ref{lem:advance_feasibility}.
Therefore we conclude that
\begin{align*}
\|\ww''' - \ww\|_1 
& =  \|\Delta\ww\|_1  + \|\ww' - (\ww+\Delta\ww)\|_1  + \|\ww'' - \ww'\|_1 + \|\ww''' - \ww''\|_1\\
& \leq \|\Delta\ww\|_1 \\
& + \left(\delta^2 \|\ww+\Delta\ww\|_1\right)^{5/2} \cdot 6\cdot 10^4 \cdot {\log \|\ww+\Delta\ww\|_1}\\
& + m^{-10}\\
& + \left( 10^6\cdot \delta^2 \|\ww+\Delta\ww\|_1 \cdot \log \|\ww\|_1\right)^2 \cdot p \cdot\|\ww+\Delta\ww\|_1^{1/p}\,.
\end{align*}

Finally, the two steps in which the flow demand changes is the weight balancing procedure, in which the change in demand is $\Delta\tdd$, and the 
step when we restrict the regularized problem solution to the non-augmented graph. In the latter case, the change in demand $\tdd$
can be upper bounded by Lemma~\ref{lem:advance_feasibility} by at most $1$ in $\ell_1$ norm. 
Thus the demand will be perturbed by $\|\dd'-\dd\|_1 = \|\tdd+\Delta\tdd\|_1 \leq 1 + 
\|\Delta\tdd\|_1$.
This concludes the proof.
\end{proof}

Lemma~\ref{lem:advance} is the main workhorse of the improved algorithm. It shows that we can make large progress within the interior point method, while paying for some demand perturbation and for some slight increase in $\|\ww\|_1$. In order to guarantee sufficient progress,  all we are left to do is to ensure that we can set an appropriate $\delta$ such that the sum of weights never increases beyond $O(m)$. This is a mere consequence of the result given above.

\begin{lemma}\label{lem:near_opt_flow_perturbed_demand}
Suppose we have a $\mu$-central instance with weights 
$\ww \geq \onev$, where 
$\|\ww\|_1 \leq 2m + 1$
and
$\mu = m^{O(1)}$. Let $\epsilon = m^{-O(1)}$, and let $\delta = m^{-(3/8+o(1))}$. In $\widetilde{O}(\delta^{-1})$ iterations of the procedure described in Lemma~\ref{lem:advance} we obtain an instance
with duality gap at most $\epsilon$ with a total demand perturbation of 
$\widetilde{O}(\delta^{-1})$.
\end{lemma}
\begin{proof}
We perform a sequence of iterations as described in  Lemma~\ref{lem:advance}.

We will argue that within $T = \widetilde{O}(\delta^{-1} \log (m \mu \epsilon^{-1}))$ iterations of the procedure described
in Lemma~\ref{lem:advance} the barrier weights will always satisfy $\|\ww\|_1 \leq 3m$.
We need to take into account the total increase guaranteed by Lemma~\ref{lem:advance}, together with the possible weight increases caused by the weight balancing procedure.

First let us bound the total number of balancing operations, since each of these can increase a single weight by a significant amount.
First, from Definition~\ref{def:balancing_proc} we see that once an edge gets balanced it will never become unbalanced again, as weights are monotonic.

Furthermore, we see that such an operation can only occur when the largest of the two paired weights is at least $\delta \|\ww\|_1 \geq 2m \delta$, since we maintain $\ww \geq \onev$. Under the invariant that throughout the entire algorithm $\|\ww\|_1 \leq 3m$, we therefore see that this can only happen at most $(3m)/(2m\delta) = 3/(2\delta)$ times. Invoking Lemma~\ref{lem:weight_bal} we therefore get that the total weight change caused by these operations is at most 
\[
\frac{3}{2\delta} \cdot 96 \cdot \delta^4 \cdot (3m)^2 = 1296 \cdot \delta^3 m^2\,.
\]
In addition, we incur weight increases due to the progress steps; per Lemma~\ref{lem:advance}, within each of the $T$ iterations, $\|\ww\|_1$ increases by at most
\begin{align*}
&
(\delta^2 (3m))^{5/2} \cdot 6 \cdot 10^4 \cdot { \log 3m} 
+ m^{-10}
+ 10^{12}\cdot p \delta^4 (3m)^{2+1/p} \cdot \log^2 (3m)
\,.
\end{align*}
Therefore the total weight increase over $T$ iterations is at most
\[
\widetilde{O}\left( 
\left(
\delta^4 m^{5/2}
+ p \delta^3 m^{2+1/p} 
\right) \log(m\mu\epsilon^{-1}) 
\right)
\]
which is $o(m)$ as long as 
\[
\delta \leq 1/\max\left( 
		m^{3/8 + o(1)} 
		,\, 
		p^{1/3} m^{ (1+1/p)/3 + o(1) }
		\right)
= 1 / m^{3/8 + o(1)}\,.
\] 
Thus this specific choice of $\delta$ insures that the invariant $\|\ww\|_1 \leq 3m$ is satisfied throughout the entire run of the algorithm.

Finally, we bound the total perturbation in demand suffered by the flow we maintain. From Lemma~\ref{lem:improved_correction}  each progress step perturbs the demand by at most $1$ in $\ell_1$ norm. Furthermore, the weight balancing operations may perturb it by an additional $\frac{3}{2\delta}$ overall.
Summing up we obtain the desired claim.
\end{proof}

Combining with the repairing procedure elaborated in~\cite{cmsv17}, which we show how to adapt to our present setting in Section~\ref{sec:repair}, we obtain the main theorem.

\begin{theorem}\label{th:main_theorem}
Given a directed graph $G(V,E,\cc)$ with $m$ arcs and $n$ vertices, such that $\|\cc\|_\infty \leq W$, and a demand vector $\dd \in \bbZ^{n}$,  in $m^{11/8+o(1)} \log W$ time we can obtain a flow $\ff$ which routes $\dd$ in $G$ while satisfying the capacity constraints $\zerov \leq \ff \leq \onev$ and minimizing the cost $\sum_{e\in E} c_e f_e$, or certifies that no such flow exists.
\end{theorem}
\begin{proof}
Before proceeding, we note that Lemma~\ref{lem:near_opt_flow_perturbed_demand} assumes that the initial centrality parameter $\mu$ is bounded by a polynomial in $m$. This may not be exactly true as the initial centering (Lemma~\ref{lem:initialization}) is done while setting a parameter $\mu$ which is upper bounded by $m^{O(1)} W$. As Lemma~\ref{lem:near_opt_flow_perturbed_demand} which bounds the number of iterations of our interior point method assumes that initially $\mu = m^{O(1)}$, we need to enforce the property that the method in 
the lemma is only called on instances where $W = m^{O(1)}$.

Similarly to~\cite{cmsv17}, we enforce this property by employing the scaling technique of~\cite{gabow1983scaling}. Thus, we reduce the problem to solving $O(\log W)$ instances of the problem, where costs are polynomially bounded. Without doing so, the running time of our algorithm would have depended super-linearly in $\log W$.

For each of the $O(\log W)$ instances, we proceed as follows. First, we apply Lemma~\ref{lem:initialization} to obtain a $\mu$-central solution for $\mu = m^{O(1)}$. Then, we apply Lemma~\ref{lem:near_opt_flow_perturbed_demand} with $\epsilon = m^{-3}$ to obtain a flow $\ff$ which routes a perturbed demand $\dd'$ where $\|\dd'-\dd\|_1 \leq m^{3/8+o(1)}$, and is close to optimal, in the sense that it is accompanied by dual variables which certify a duality gap of $O(m^{-3})$. Note that each of the $m^{3/8+o(1)}$ iterations of the algorithm from Lemma~\ref{lem:near_opt_flow_perturbed_demand} requires solving the regularized objective (\ref{eq:regularized_newton}) to high precision, which due to our choice of parameter $p$ (Corollary~\ref{cor:step_running_time}) can be done in $m^{1+o(1)}$ time.
Finally, we apply Lemma~\ref{lem:repair_final} to repair this flow in $m^{11/8+o(1)}$ time. Hence the total running time is $m^{11/8+o(1)}$.

We can certify infeasibility (the case where the demand can not be routed while satisfying the capacity constraints), by looking at the arcs used by the returned optimal solution. The initialization procedure assumes that a demand satisfying flow can be routed in the graph. This is done by adding to the graph $O(m)$ arcs with high costs (see Lemma~\ref{lem:augmented_graph_initialized}). Using Lemma~\ref{lem:augmented_graph_initialized} we see that simply inspecting the optimal solution returned for this graph we can decide if the routing is infeasible in $G$.
\end{proof}

\begin{figure}
\frame{
\begin{minipage}{\textwidth}
\vspace{10pt}
\hspace{5pt}
$\textsc{BalanceWeights}(G; \ww, \ff)$
\begin{enumerate}

\item {
For each $e\in E$: }

\item{\ \ \ If $w_e^+ < 96\cdot \delta^4 \left\Vert\ww\right\Vert_1^2$ and $w_e^- > \delta\left\Vert\ww\right\Vert_1$ then 
	$w_e^{'+} \leftarrow 96\cdot \delta^4 \left\Vert\ww\right\Vert_1^2$, $w_e^{'-} \leftarrow w_e^-$.}
\item{\ \ \ If $w_e^+ > \delta\left\Vert\ww\right\Vert_1$ and $w_e^- < 96\cdot \delta^4 \left\Vert\ww\right\Vert_1^2$ then 
		$w_e^{'+}\leftarrow w_e^+$,
	$w_e^{'-} \leftarrow 96\cdot \delta^4 \left\Vert\ww\right\Vert_1^2$.}
\item{Set $\zerov < \ff' < \onev$ such that $\frac{\ww^{'+}}{\onev-\ff'} - \frac{\ww^{'-}}{\ff'} = \frac{\ww^+}{\onev-\ff} - \frac{\ww^-}{\ff}$.}
\item{Return $\ww',\ff'$.}
\end{enumerate}
\vspace{5pt}
\end{minipage}
}
\caption{Weight balancing procedure}
\label{fig:balance_weights}
\end{figure}

\section{Repairing the Flow}
\label{sec:repair}
In this section we show that the flow produced in Section~\ref{sec:faster_mcf} can be repaired in such a way that we obtain an optimal flow for the original problem. This can be done by applying the combinatorial fixing techniques used in~\cite{cmsv17}. Let us first recall a few useful definitions concerning the bipartite perfect $\bb$-matching problem.

\paragraph{Bipartite Perfect $\bb$-matching.}
For a given weighted bipartite graph $G=(V,E)$ with $V=V_1 \cup V_2$ where $V_1$ and $V_2$ are the two sets of bipartition and a \emph{demand vector} $\bb\in\bbR_{+}^{V}$,
a \emph{perfect $\bb$-matching} is a vector $\xx\in\bbR_{+}^{E}$
such that $\sum_{e\in E(v)}x_{e}=b_{v}$ for all $v\in V$. A
perfect $\bb$-matching is a generalization of perfect matching; in
the particular case where all $\bb$'s equal $1$, integer $\bb$-matchings
(or $\onev$-matchings) are exactly perfect matchings.
We require that $\bb(V_1)=\bb(V_2)$ as otherwise they trivially do not exist.

The \emph{weighted bipartite perfect $\bb$-matching problem} is defined as follows: given a weighted bipartite graph $G=(V,E,\cc)$,  return a perfect $\bb$-matching in $G$ that has minimum weight, or conclude that there is no perfect $\bb$-matching in $G$. The dual problem to the weighted
perfect bipartite $\bb$-matching is a \emph{$\bb$-vertex packing problem}
where we want to find a vector $\yy\in\mathbb{R}^{V}$ satisfying the
following LP 
\begin{eqnarray*}
\max &  & \sum_{v\in V}y_{v}b_{v},\\
 &  & y_{u}+y_{v}\le w_{uv}\quad\forall u, v\in E.
\end{eqnarray*}

We employ the following result which follows from a direct application of Lemmas 35, 37, and Theorem 40 in~\cite{cmsv17}.
\begin{theorem}
\label{thm:fixmatching}
Consider an instance $G = (V, E, \cc)$ of the weighted perfect bipartite $\bb$-matching problem, where $\|\bb-\hbb\|_1 \leq P$, and $\|\bb\|_1 = O(m)$.
Given a feasible primal-dual pair of variables $(\xx, \yy)$ with duality gap at most $m^{-2}$, in $O(Pm + m\log n)$ time we can compute
an optimal primal-dual pair $(\hxx, \hyy)$ to the perfect $\hbb$-matching problem, or conclude that no such matching exists.
\end{theorem}

In order to apply this theorem we need to convert our instance into a $\bb$-matching instance with small duality gap. 
The output we receive from Lemma~\ref{lem:near_opt_flow_perturbed_demand} is a flow $\ff$ which routes a slightly different demand  $\dd$ from the one we originally intended. Furthermore $\ff$ satisfies the centrality condition
$\CC^\top \left( \frac{\ww^-}{\ff}-\frac{\ww^+}{\onev-\ff}\right) = \CC^\top \frac{\cc}{\mu}$ for a small value of $\mu$.

\paragraph{From Flow to Bipartite $b$-Matching.}

We convert this solution into a corresponding bipartite $b$-matching problem in a new bipartite graph $G'(V_1 \cup V_2, E', \cc')$, defined as follows:
\begin{enumerate}
\item for each vertex $v \in V$, create the corresponding vertex in $V_1$;
\item for each arc $e \in E$, create a vertex $v_e \in V_2$;
\item for each arc $e = (u,v) \in E$ create edges $(u, v_e)$ with cost $c_e$ and $(v,v_e)$ with cost $0$.
\end{enumerate}

Let $\dd$ be the demand routed by $\ff$. For the bipartite graph $G'$, we define the $\bb$ vector as follows:
\begin{enumerate}
\item for each $v \in V$ set for the corresponding vertex $v$ in $V_1$:
$b_{v} = \vert E^+(v) \vert + d_v$.
\item for each $v_2 \in V_2$ set $b_{v_2} = 1$.
\end{enumerate}

We can easily verify that an optimal primal-dual solution to the $\bb$-matching problem in $G'$ maps to an optimal primal-dual solution to the minimum cost flow problem in $G$.

\begin{lemma}\label{lem:matchingtoflow}
Let $(\xx,\yy)$ be a pair of feasible primal-dual variables for the $\bb$-matching problem in the graph $G'$, constructed according to the rules defined above. These can be mapped to a pair $(\ff, \zz)$ of optimal feasible primal-dual variables for the minimum cost flow problem in $G$.
Furthermore, this mapping can be constructed in linear time.
\end{lemma}
\begin{proof}
Let $\xx$ be the optimal $\bb$-matching, and let $\yy$ be its dual vector certifying optimality.

We construct $\ff$ as follows: for each arc $e = (u,v)\in G$ we look at the corresponding gadget in $G'$ consisting of vertices $u,v\in V_1$ and $v_e \in V_2$
and set $f_e = x_{u,v_e}$.

First we verify that $\ff$ is feasible. For each arc $e = (u,v) \in G$, the corresponding vertex $v_e \in V_2$ has $b_e = 1$, and therefore $x_{u,v_e} + x_{v,v_e} = 1$. As both are non-negative it means that in $G$, $0 \leq f_e \leq 1$.

Next we verify that $\ff$ indeed routes the demand it is supposed to route, i.e. $\dd$. By definition we have that for each $v \in V_1$, 
\[
\sum_{(v,v_e)\in E'} x_{v,v_e}= b_v = \vert E^+(v) \vert + d_v\,.
\]
Therefore the demand at vertex $v$ can be written as 
\begin{align*}
\sum_{e = (u,v) \in E} f_e - \sum_{e' = (v,u) \in E} f_{e'} &= \sum_{e = (u,v) \in E} (1-f_e) - \vert E^+(v) \vert + \sum_{e' = (v,u) \in E} f_{e'} 
\\
&= \sum_{(v,v_e)\in E'} x_{v,v_e} - \vert E^+(v) \vert \\
&= d_v\,.
\end{align*}
Finally we certify optimality for $\ff$ by exhibiting feasible dual variables which satisfy complementary slackness.
For the LP formulation defined in Section~\ref{sec:lpform} we require for each arc $e \in E$ two non-negative dual variables $z_e^-$ and $z_e^+$ such that summing up along each cycle
\[
z_e^+ - z_e^- + c_e
\]
with the appropriate sign depending on the orientation of the encountered arcs, we obtain exactly $0$. For each $e = (u,v)\in E$, we define them as
\begin{align*}
z_e^- &= c_e -y_u - y_{v_e}\,,\\
z_e^+ &= -y_v - y_{v_e}\,.
\end{align*}
Feasibility is obvious since $z_e^+, z_e^- \geq 0$ because $\yy$ is a feasible vector and hence for each edge the sum of the $y$'s of its vertices is upper bounded by the cost, and summing up $z_e^- - z_e^+$ along any cycle we obtain exactly the sum of the costs $c_e$ along that cycle, summed up with the appropriate sign.

We finally certify optimality for $\ff$ by verifying that $(\ff, (\zz^-; \zz^+))$ satisfy complementary slackness.
We have that for each arc $(u,v) \in E$, 

\begin{align*}
f_e z_e^- &= x_{u v_e} (c_e -y_u - y_{v_e} ) = 0 \\
(1-f_e) z_e^+ &= x_{v v_e} (-y_v - y_{v_e}) = 0\,,
\end{align*}
where we used the fact that $(\xx,\yy)$ satisfy complementary slackness. Hence $\ff$ is optimal in $G$.
\end{proof}

Our goal will be to obtain an optimal primal-dual solution for a $\hbb$-matching problem which carries over to an optimal solution to the minimum cost flow problem in $G$ with the original demand $\hdd$. 

In order to do so, we first convert our flow instance with small duality gap to a $\bb$-matching instance with the same duality gap. 
Afterwards, we fix the resulting $\bb$-matching to optimality and we convert it to an optimal $\hbb$-matching which maps to a flow $\hff$ routing the original demand $\hdd$ in $G$, via Theorem~\ref{thm:fixmatching}.

In order to do so, we use the $\mu$-central solution produced by the interior point method in Section~\ref{sec:faster_mcf}, with $\mu = O(m^{-10})$. Recalling that our current solution satisfies
\[
\CC^\top \left( \frac{\mu \ww^-}{\ff} - \frac{\mu \ww^+}{\onev-\ff} \right) = \CC^\top \cc\,,
\]
we set primal-dual variables $(\xx,\yy)$ for $G'$ as follows:
\begin{enumerate}
\item for each arc $e = (u,v) \in E$, set $x_{u, v_e} = f_e$
and $x_{v,v_e} = 1-f_e$.
\item pick an arbitrary vertex $v_0 \in V$ and set its corresponding $y_{v_0} = 0$;
then perform a breadth-first search in $G$ (ignoring orientations) starting from $v$, and for each vertex visited for the first time when traversing an arc $e=(u,v)$ (in any direction) such that
$y_u - y_v = \frac{\mu w^+_e}{1-f_e} - \frac{\mu w^-_e}{f_e} - c_e$;
\item for all $v_e \in V_2$ where $e = (u,v) \in E$, set
$y_{v_e} = -y_u  - \frac{\mu w_e^-}{f_e} + c_e = -y_v - \frac{\mu w_e^-}{1-f_e}$.
\end{enumerate}

Having defined $G'$ together with its demand vector $\bb$ and corresponding primal-dual solution $(\xx,\yy)$, let us show that indeed this is a feasible solution and it has low duality gap.

\begin{lemma}\label{lem:flowtomatching}
The primal-dual solution $(\xx,\yy)$ constructed above is feasible and has duality gap $\mu \|\ww\|_1$.
\end{lemma}
\begin{proof}
First we check feasibility.
Primal feasibility holds trivially by construction. For dual feasibility,
consider for each arc $e = (u,v) \in E$  the two edges appearing in its corresponding gadget in $G'$, $(u, v_e)$ and $(v,v_e)$.
For the former we have
\[
y_u + y_{v_e} = -\frac{\mu w_e^-}{f_e} + c_e \leq c_{u,v_e}
\]
and for the latter
\[
y_v + y_{v_e} = -\frac{\mu w_e^+}{1-f_e} \leq  c_{v,v_e} = 0\,.
\]
Therefore the dual constraints of the $\bb$-matching problem are also satisfied. 
Finally we write the duality gap via
\begin{align*}
\sum_{e = (u,v) \in E'} x_e (c_e - y_u - y_v)
&=
\sum_{e = (u,v) \in E} \left(f_e (c_e - y_u - y_{v_e}) + (1-f_e)(0 - y_v - y_{v_e})  \right) \\
&= \sum_{e = (u,v) \in E} \left( f_e \cdot \frac{\mu w_e^-}{f_e} + (1-f_e) \cdot \frac{\mu w_e^+}{1-f_e} \right)
\\
&=
\sum_{e = (u,v)\in E} \mu (w_e^- + w_e^+) \\
&= \mu \|\ww\|_1\,.
\end{align*}
\end{proof}
Now we are ready to state the main result of this section.
\begin{lemma}\label{lem:repair_final}
Given
a target integral demand $\hdd$, 
 a flow $\ff$ satisfying $\mu$-centrality with weights $\ww$ such that $\|\ww\|_1 = O(m)$, $\mu = O(m^{-10})$, and $\ff$ routes a demand $\dd$ we can obtain in time $\widetilde{O}(\|\dd-\hdd\|_1 \cdot m)$ a flow $\hff$ which is optimal among all flows routing  $\hdd$.
\end{lemma}
\begin{proof}
Using Lemma~\ref{lem:flowtomatching} we can convert the flow instance into a $\bb$-matching instance with duality gap $O(m^{-2})$. Now we can invoke Theorem~\ref{thm:fixmatching} to convert the matching into an optimal $\hbb$-matching where $\hbb$ corresponds to the demand $\hdd$ that the flow is supposed to route in $G$. We see that by construction $\|\bb-\hbb\|_1 = \|\dd-\hdd\|_1$.
Therefore obtaining the corresponding optimal $\hbb$-matching is done in time 
$O((\|\dd - \hdd\|_1 + \log n) m)$. Finally applying Lemma~\ref{lem:matchingtoflow} we obtain the optimal flow $\hff$ in $G$.
\end{proof}

 \section{Improving the Running Time to $m^{4/3+o(1)}\log W$}
\label{sec:improved_running_time}

The running time of the algorithm we presented above hits a barrier of $m^{11/8+o(1)}\log W$. The key bottleneck there is the post-processing we do on the residual error of the (intermediate) solutions we obtain after performing each progress step. Indeed, while the length of our steps is dictated by the $\ell_\infty$-norm of our step size $\|\vrho\|_\infty$ (which is, in a sense optimal), it is unclear how to ensure that the energy required to route a flow that fixes the corresponding residual error is sufficiently small, without overly increasing the weights of the constraint barriers. In fact, the extent of weight perturbations necessary for this post-processing step are exactly what determines the $m^{11/8+o(1)}\log W$ running time. All the other weight perturbations, which are caused by the regularization terms, are much milder and would lead to the desired $m^{4/3+o(1)}\log W$ running time. 

After the first version of this paper was posted \cite{axiotis2020circulation}, Liu and Sidford~\cite{liu2020faster} published a preprint that obtains an improved running time for the unit-capacity maximum flow. The main technique introduced in that paper boils down exactly to avoiding the aforementioned bottleneck. Roughly, instead of advancing from a central point to the next one via a progress step
followed by a sequence of residual correction steps, they instead directly solve the optimization problem which lands them at the next central point. 

To do so, one must guarantee that this optimization problem is well-conditioned at all times, in the sense that the objective has a Hessian which always stays within a constant factor from the one at the origin. While such a Hessian stability condition is not true in general, Liu and Sidford~\cite{liu2020faster} modify the logarithmic barriers they use by extending them with quadratics outside the region where they would be naturally well-behaved. The resulting new objective function can be efficiently optimized by slightly extending the mixed $\ell_2$-$\ell_p$ solver of Kyng et al~\cite{kyng2019flows} (see Appendix~\ref{sec:strong_solver}).

Incorporating this idea in our framework yields a the desired running time improvement of our unit-capacity minimum cost flow algorithm as well. Most of the details, particularly those involving preconditioning and weight perturbations carry over from the previous sections. In fact, the only change to the algorithm needed is to replace the regularized Newton step (cf. line 3 in Figure~\ref{fig:modified_augment}) with solving a regularized problem which directly involves the logarithmic barrier. (Furthermore, lines 7 and 8 are no longer required.)

\paragraph{Method Overview.}
Let us specify the ideal optimization problem which we would solve in order to advance along the central path. Suppose we have a $\mu$-central instance i.e. we have a flow $\ff = \ff_0 + \CC\xx$ which satisfies 
\begin{align}
\label{eq:central_1}
\CC^\top \left( \frac{\ww^+}{\onev-\ff_0 - \CC\xx} - \frac{\ww^-}{\ff_0 + \CC\xx} \right) = -\frac{\CC^\top \cc}{\mu}\,.
\end{align}
as it optimizes the convex objective
\[
\min_{\xx} F_{\mu}^{\ww}(\xx) = 
\min_{\xx}
\frac{1}{\mu}\langle \cc, \CC\xx \rangle
-\sum_{e\in E}
\left(
w_e^+ \log(\onev - \ff_0 - \CC\xx)_e + w_e^- \log(\ff_0 + \CC\xx)_e
\right)
\,.
\]
Our goal is to design an optimization procedure which enables us to augment $\ff$ with a circulation $\CC\txx$ in order to obtain an optimizer for 
\begin{align*}
\min_{\xx'} F_{\mu/(1+\delta)}^{\ww}(\xx') &= \min_{\xx'}
\frac{1+\delta}{\mu}\langle \cc, \CC\xx' \rangle
-\sum_{e\in E}
\left(
w_e^+ \log(\onev - \ff_0 - \CC\xx')_e + w_e^- \log(\ff_0 + \CC\xx')_e
\right) \\
&= \min_{\txx}
\frac{1+\delta}{\mu}\langle \cc, \ff + \CC\txx \rangle
-\sum_{e\in E}
\left(
w_e^+ \log(\onev - \ff - \CC\txx)_e + w_e^- \log(\ff + \CC\txx)_e
\right) 
\,.
\end{align*}
We can equivalently rewrite this optimization procedure, after adding a constant term to it, as
\[
\min_{\txx} \Psi_{\mu/(1+\delta)}^{\ww,\ff}(\txx)\,,
\]
where
\begin{align*}
\Psi_{\mu/(1+\delta)}^{\ww,\ff}(\txx)  :&= F_{\mu/(1+\delta)}^{\ww}(\xx+\txx)
- \frac{1+\delta}{\mu} \langle \cc, \ff \rangle
+ \sum_{e\in E} \left( w_e^+ \log (1-f_e) + w_e^- \log f_e \right) \\
 &= \frac{1+\delta}{\mu}\langle \cc, \CC\txx\rangle
- \sum_{e\in E}\left(
 w_e^+ \log\left(\onev- \frac{\CC\txx}{\onev-\ff} \right)_e
 +
 w_e^- \log\left(\onev+ \frac{\CC\txx}{\ff} \right)_e
  \right)\,.
\end{align*}
Optimizing this function is hard to do in general. However, assuming its optimal augmenting circulation $\CC\txx$ satisfies 
\begin{equation}
\label{eq:no_congest_condition}
\inorm{\frac{\CC\txx}{\min\{\ff, \onev-\ff\}}} \leq \frac{1}{10}\,,
\end{equation}
which in other words, says that augmenting $\ff$ with $\CC\txx$ will not 
come close to breaking
the feasibility constraints, we can instead solve a well-conditioned objective obtained by replacing the $\log$'s with a better behaved function 
$\flog$ satisfying $\log(1+t) = \flog(1+t)$ whenever $\vert t\vert\leq 1/10$. 

More precisely, we use the definition from~\cite{liu2020faster} which we reproduce below:
\[
\flog(1+t) =
\begin{cases}
\log(1+t), &\textnormal{ if }t\in[-\theta, \theta]\,,\\
\log(1+\theta) + (t-\theta) \cdot \log'(1+\theta) + (t-\theta)^2 \cdot \frac{1}{2}\log''(1+\theta), &\textnormal{ if }t > \theta\,,\\
\log(1-\theta) + (t+\theta) \cdot \log'(1-\theta) + (t+\theta)^2 \cdot \frac{1}{2}\log''(1-\theta), &\textnormal{ if }t < -\theta\,,\\
\end{cases}
\]
where we set $\theta = 1/10$. We can also easily verify that on the boundary of the interval $[-\theta,\theta]$ the first and second order derivatives exactly match those of $\log(1+t)$. The essential feature is that outside this range, the second derivative stays constant, whereas in the case of $\log(1+t)$ it changes very fast.

Using this, we can define the minimization problem
\[
\min_{\txx} \widetilde{\Psi}_{\mu/(1+\delta)}^{\ww,\ff}(\txx)\,,
\]
where
\begin{equation}
\label{eq:psitil_def}
\widetilde{\Psi}_{\mu/(1+\delta)}^{\ww,\ff}(\txx)
:=
\frac{1+\delta}{\mu}\langle \cc, \CC\txx\rangle
- \sum_{e\in E}\left(
 w_e^+ \flog\left(\onev- \frac{\CC\txx}{\onev-\ff} \right)_e
 +
 w_e^- \flog\left(\onev+ \frac{\CC\txx}{\ff} \right)_e
  \right)\,.
\end{equation}

\begin{observation}
If the minimizer $\txx$ of $\widetilde{\Psi}_{\mu/(1+\delta)}^{\ww,\ff}(\txx)$ satisfies the low-congestion condition from (\ref{eq:no_congest_condition}), then 
it also minimizes ${\Psi}_{\mu/(1+\delta)}^{\ww,\ff}(\txx)$.
\end{observation}

Due to the fact that the second order derivatives of $\flog$ are bounded, $\widetilde{\Psi}_{\mu/(1+\delta)}^{\ww,\ff}(\txx)$ is well-conditioned and, as a matter of fact can easily be minimized by using a small number of calls to a routine which minimizes quadratics over the set of circulations. As specified in~\cite{liu2020faster} this can be easily done by using fast Laplacian system solvers. However, the main difficulty that arises is ensuring that the minimizer $\txx$ satisfies the confition from (\ref{eq:no_congest_condition}).

Enforcing this property is non-trivial, and requires regularizing the function $\widetilde{\Psi}_{\mu/(1+\delta)}^{\ww,\ff}$ in an identical manner to the way we did it in the analysis from Section~\ref{sec:faster_mcf}. Most of the results we used there carry over, after performing some minor modifications in the analysis.

\paragraph{Regularizing the Objective.}
In order to enforce the required property, we add two regularization terms to our optimization step. Just like in Section~\ref{sec:regularized_newton} we augment the graph $G$ to $\Gstar$, which has a cycle basis $\CCstar$, and in this graph we write down the equivalent concave maximization problem for (\ref{eq:psitil_def}), which we regularize with two extra terms. We slightly abuse notation 
by  making 
$\widetilde{\Psi}_{\mu/(1+\delta)}^{\ww,\ff}$ act on
an element in the circulation space of $\Gstar$, with the
meaning that the linear and logarithmic terms
only act on edges in $E$:
\begin{align}
\label{eq:strong_reg_obj}
&\max_{\tffstar = \CCstar \txx} -\widetilde{\Psi}_{\mu/(1+\delta)}^{\ww,\ff}(\txx) - \frac{\Rstar}{2} \sum_{e\in E'} (\tf_\star)_e^2 - \frac{\Rp}{p} \sum_{e\in E \cup E'} (\tf_\star)_e^p \,.
\end{align}

Writing $\tffstar = \tff + \tff'$ where $\tff$ is the restriction of $\tffstar$ to the edges $E$ of $G$, and $\tff'$ is the restriction to the augmenting edges $E'$, and 
using the centrality condition (\ref{eq:central_1}),
 we can further write this as
\begin{align}\notag
\max_{\tffstar = \CCstar \txx}
(1+\delta)
\left\langle \frac{\ww^+}{\onev-\ff} - \frac{\ww^-}{\ff}, \tff \right\rangle
&+
\sum_{e\in E} 
\left(
 w_e^+ \flog\left(1 - \frac{\tf_e}{1 - f_e}\right)
+
 w_e^- \flog\left( 1 + \frac{\tf_e}{f_e}\right)
\right) 
\\
&- \frac{\Rstar}{2} \sum_{e \in E'} (\tf_e')^2 
- \frac{\Rp}{p} \sum_{e \in E \cup E'} (\tfstar)_e^p\,.
\label{eq:strong_reg_obj_max}
\end{align}

Note that this optimization problem can be solved efficiently using Lemma~\ref{lem:strong_solver}. The key reason is that all the terms except for the one involving $p$ powers of the flow $\tf$ have a second order derivative which is either  constant, or which stays bounded within a small multiplicative factor from the one at $0$. Similarly to before, the solver from Lemma~\ref{lem:strong_solver} yields a high-accuracy, yet inexact solution. We can, however, assume we obtain an exact solution by performing minor perturbations to our problem, in a manner similar to the one discussed in Section~\ref{sec:solver_error}.

We can now write the first order optimality condition for the objective in (\ref{eq:strong_reg_obj}), thus providing an analogue of Lemma~\ref{lem:opt_non_aug}.
Before doing so, we give the following helper lemma, which will be the main driver of the results in this section.
It intuitively states that the optimal solution to (\ref{eq:strong_reg_obj_max}) can be thought of as the solution to a
 regularized Newton step as the one in (\ref{eq:regularized_newton}), but where the coefficients of the quadratic 
$\frac{1}{2} \sum\limits_{e\in E} \tf_e^2 \cdot \left(\frac{w_e^+}{(s_e^+)^2} + \frac{w_e^-}{(s_e^-)^2}\right)$
have been slightly perturbed by a small multiplicative constant.

\begin{lemma}[Optimality with average Hessian]\label{lem:optimality_avg_hessian}
Let $\tffstar = \CCstar \txxstar$ be the optimizer of (\ref{eq:strong_reg_obj_max}), and let $\tffstar = \tff + \tff'$, where the two components are supported on $E$ and $E'$, respectively. Then 
there exists a vector $\valpha = (\valpha^+; \valpha^-)$, $(1+\theta)^{-2} \cdot \onev \leq \valpha \leq (1-\theta)^{-2} \cdot \onev$, which can be explicitly computed, such that for any circulation $\gg$ in $\CCstar$, 
\begin{align*}
\left\langle \gg, \left[\begin{array}{c} 
\delta \left( \frac{\ww^+}{\onev-\ff} - \frac{\ww^-}{\ff} \right)  
- \tff \left(
\frac{\valpha^+ \ww^+}{(\onev-\ff)^2} 
+
\frac{\valpha^- \ww^-}{\ff^2} 
\right)
- \Rp \cdot (\tff)^{p-1}  \\ 
- \Rstar \cdot \tff' 
-\Rp  \cdot (\tff')^{p-1} 
\end{array}\right] \right\rangle 
= 0\,.
\end{align*}
\end{lemma}
\begin{proof}
We write the first order optimality condition for (\ref{eq:strong_reg_obj_max}). We have that for any circulation $\gg$ in $\Gstar$:
\begin{align*}
\left\langle \gg, \left[\begin{array}{c} 
(1+\delta) \left( \frac{\ww^+}{\onev-\ff} - \frac{\ww^-}{\ff} \right)  
+ \ww^+ \cdot \nabla_{\tff}\flog\left(\onev-\frac{\tff}{\onev-\ff}\right)
+ \ww^- \cdot \nabla_{\tff}\flog\left(\onev + \frac{\tff}{\ff}\right)
- \Rp \cdot (\tff)^{p-1}  \\ 
- \Rstar \cdot \tff' 
-\Rp  \cdot (\tff')^{p-1} 
\end{array}\right] \right\rangle = 0\,.
\end{align*}
Next we use the expansion
\[
\frac{d}{dx}  \flog\left(1+\frac{x}{a} \right) = \frac{1}{a} \cdot \flog'\left(1+\frac{x}{a}\right)
= \frac{1}{a} \cdot   \left(
1 + \frac{x}{a} \int_0^1  \flog''\left(1 + t \cdot \frac{x}{a}\right) dt 
\right)
\,.
\]
From the definition of $\flog$ we know that its second order derivative is always stable which enables us to bound
\[
-\frac{1}{(1-\theta)^2}
\leq
\int_0^1  \flog''\left(1 + t \cdot \frac{x}{a}\right) dt
\leq
-\frac{1}{(1+\theta)^2}\,,
\]
hence for some $\alpha \in [(1+\theta)^{-2}, (1-\theta)^{-2}]$, one has
\[
\frac{d}{dx}  \flog\left(1+\frac{x}{a} \right) = \frac{1}{a} \cdot \flog'\left(1+\frac{x}{a}\right)
= \frac{1}{a} - \alpha \cdot \frac{x}{a^2} 
\,.
\]
Therefore there exists a vector $\valpha = (\valpha^+; \valpha^-)$ satisfying
$ (1+\theta)^{-2} \cdot \onev \leq \valpha \leq (1-\theta)^{-2} \cdot \onev $ 
such that
\begin{align*}
&\left\langle \gg, \left[\begin{array}{c} 
(1+\delta) \left( \frac{\ww^+}{\onev-\ff} - \frac{\ww^-}{\ff} \right)  
+ \ww^+ \cdot \nabla_{\tff}\flog\left(\onev-\frac{\tff}{\onev-\ff}\right)
+ \ww^- \cdot \nabla_{\tff}\flog\left(\onev + \frac{\tff}{\ff}\right)
- \Rp \cdot (\tff)^{p-1}  \\ 
- \Rstar \cdot \tff' 
-\Rp  \cdot (\tff')^{p-1} 
\end{array}\right] \right\rangle 
\\
=& 
\left\langle \gg, \left[\begin{array}{c} 
\delta \left( \frac{\ww^+}{\onev-\ff} - \frac{\ww^-}{\ff} \right)  
- \tff \left(
\frac{\valpha^+ \ww^+}{(\onev-\ff)^2} 
+
\frac{\valpha^- \ww^-}{\ff^2} 
\right)
- \Rp \cdot (\tff)^{p-1}  \\ 
- \Rstar \cdot \tff' 
-\Rp  \cdot (\tff')^{p-1} 
\end{array}\right] \right\rangle 
= 0
\,,
\end{align*}
which is what we needed.

\end{proof}

Lemma~\ref{lem:optimality_avg_hessian} enables us to use an optimality condition very similar to the one we had before in Section~\ref{sec:faster_mcf}. As a matter of fact, all the remaining statements are nothing but "robust" versions of those we previously used. Essential here are new versions of Lemma~\ref{lem:opt_non_aug} and Lemma~\ref{lem:prec_lemma} which accommodate the extra multiplicative factors on  resistances. Roughly, our goal is to provide upper bounds on $\|\tff\|_\infty$
and $\|{\tff/\min\{\onev - \ff, \ff\}^2}\|_\infty$, which together will imply that the condition from (\ref{eq:no_congest_condition}) is satisfied.

After proving that this is the case, we will show how to advance to the next point on the central path -- the regularization terms on (\ref{eq:strong_reg_obj_max}) will require us to increase the weights $\ww$ in order to obtain optimality for the non-regularized objective i.e. $\nabla \widetilde{\Psi}_{\mu/(1+\delta)}^{\ww,\ff}(\txx) = 0$. Nevertheless, this procedure is essentially identical to the one we previously used in Lemma~\ref{lem:perturbation_correction}.

\subsection{Bounding Congestion}
Here we show that the congestion condition from (\ref{eq:no_congest_condition}) is satisfied, and hence the minimizer of  (\ref{eq:strong_reg_obj_max}) also minimizes the expression after replacing $\flog$ with $\log$.
The proofs are almost identical to those from Section~\ref{sec:faster_mcf}. For consistency we will use the slack notation
\begin{align*}
\ss^- = \ff\,, \quad
\ss^+ = \onev - \ff\,,
\end{align*}
and use the shorthand notation for the residual
\begin{align*}
\hh &= \delta\left( \frac{\ww^+}{\ss^+} - \frac{\ww^-}{\ss^-} \right)\,.
\end{align*}

We first give a short lemma providing an optimality condition for the restriction of the flow $\tffstar$ computed by 
(\ref{eq:strong_reg_obj_max}) to the edges $E$ of the original graph $G$.

\begin{lemma}[Optimality in the non-augmented graph]
\label{lem:strong_opt_non_aug}
Let $\tffstar = \CCstar \txxstar$ be the optimizer of (\ref{eq:strong_reg_obj}) and let $\tff$ be its restriction to the edges of $G$. Let $\tdd$ be the demand routed by $\tff$ in $G$. Then there exists a vector $\valpha =(\valpha^+;\valpha^-)\in \bbR^{2m}$, $\frac{1}{(1+\theta)^2} \cdot \onev \leq \valpha \leq \frac{1}{(1-\theta)^2} \cdot \onev$, which can be explicitly computed, such that
\begin{align}
\label{eq:first_order_opt_mixed_obj_barrier}
\CC^\top \cdot \left( \frac{\valpha^+ \ww^+}{(\ss^+)^2}  + \frac{\valpha^-\ww^-}{(\ss^-)^2} \right) \cdot \tff = \CC^\top (\hh + \Delta\hh)
\end{align}
where $\Delta\hh = -\Rp (\tff)^{p-1}\,,$
and
\begin{align*}
&\|\tdd\|_1 
\leq 
 3 \left( \frac{\|\ww\|_1 \cdot \energymax(\hh,\ww,\ss) }{\Rstar} \right)^{1/2}\,, \\
&\|\tffstar\|_p 
\leq 
\left(\frac{p\cdot \frac{3}{2}\energymax(\hh,\ww,\ss)}{\Rp}\right)^{1/p}\,.
\end{align*}
\end{lemma}
The proof can be found in Appendix~\ref{sec:lem_strong_opt_non_aug}.

Next we provide a guarantee enforced by the component of the regularizer involving augmenting edges.
\begin{lemma}
\label{lem:strong_prec_lemma}
Let $\tffstar$ be the solution of the regularized objective and $\tff$ its restriction on $G$, and suppose that $\|\ww\|_1 \geq 3$. Then
there exists a vector $\valpha = (\valpha^+; \valpha^-) \in \bbR^{2m}$ such that for all edges $e\in E$:
\begin{equation}
\begin{aligned}
\left\vert \left( 
\frac{\alpha^+_e w^+_e}{(s^+_e)^2}+\frac{\alpha^-_e w^-_e}{(s^-_e)^2}  
+
\Rp \cdot \tf_e^{p-2}
\right) \tf_e - h_e \right\vert
&\leq \hat{\gamma} \,,
\end{aligned}
\end{equation}
where 
\begin{align*}
\hat{\gamma} = \left(\Rstar
+\Rp \cdot \|\tffstar\|_\infty^{p-2}\right)^{1/2} 
\cdot \left\| \frac{ \hh }{  \sqrt{(\ww^+ + \ww^-)\left( \frac{\valpha^+ \ww^+}{(\ss^+)^2}+\frac{\valpha^- \ww^-}{(\ss^-)^2} \right) }} \right\|_\infty \cdot {32 \log \|\ww\|_1} \,.
\end{align*}
Furthermore, this implies that
\begin{align}
\left(\frac{\alpha^+_e w_e^+}{(s_e^+)^2} + \frac{\alpha^-_e w_e^-}{(s_e^-)^2}\right) \cdot \left\vert\tf_e  \right\vert
&\leq 
\left\vert h_e \right \vert + \hat{\gamma}\,.
\end{align}
\end{lemma}
The proof is essentially identical to that of Lemma~\ref{lem:prec_lemma}. 
We discuss it at the end of Appendix~\ref{sec:prec_proof}.

Combining Lemmas~\ref{lem:strong_opt_non_aug} and~\ref{lem:strong_prec_lemma} we can finally prove the main statement of this section. This will be true for as long as the condition we specified in Definition~\ref{def:balanced} holds. This condition will be later enforced by performing the balancing operation we also specified in Definition~\ref{def:balancing_proc}.

In order to obtain the desired congestion bound, we set our regularization parameters identically to 
Definition~\ref{def:reg_parameters}: \begin{align*}
p &= \min\left\{ k \in 2\bbZ : k \geq (\log m)^{1/3} \right\} \,,\\
R_p &= p\cdot \left(10^6\cdot \delta^2 \|\ww\|_1
 \cdot {\log\|\ww\|_1}\right)^{p+1}\,,\\
R_\star &= 3\cdot \delta^2\|\ww\|_1^2\,.
\end{align*}

\begin{lemma}[Congestion bound]
\label{lem:strong_congestion_bound}
Suppose that all edges $e\in E$ are \textit{balanced}, per Definition~\ref{def:balanced}, and $\delta > 10 \cdot \|\ww\|_1^{-1/2}$.
 Then the restriction $\tff$ of the flow $\tffstar$ computed via (\ref{eq:strong_reg_obj_max}) satisfies the low-congestion condition (\ref{eq:no_congest_condition}).
\end{lemma}
\begin{proof}
Lemma~\ref{lem:imb_cong_bound} proves that if all edges are balanced, for any flow $\hff$ in $G$ such that for all $e\in E$
\begin{align*}
\left(\frac{w_e^+}{(s_e^+)^2} + \frac{w_e^-}{(s_e^-)^2}\right) \cdot \left|\hf_e\right| \leq \left|h_e\right| + \hat{\gamma}\,,
\end{align*}
where $\hh = \delta \left(\frac{\ww^+}{\ss^+}-\frac{\ww^-}{\ss^-}\right)$,
then for each $e\in E$ it is true that 
$\max\left\{\left|\hrho_e^+\right|,\left|\hrho_e^-\right|\right\} < C_\infty$ or
$\left|\frac{w_e^+\hrho_e^+}{s_e^+}\right|+\left|\frac{w_e^-\hrho_e^-}{s_e^-}\right| \leq 6\gamma$,
where $\widehat{\rho}_e^+ = \frac{\widehat{f}_e}{s_e^+}$, $\widehat{\rho}_e^- = \frac{-\widehat{f}_e}{s_e^-}$, and
$C_\infty = \frac{1}{2\delta\sqrt{2\|\ww\|_1}}$.
We apply this statement for 
\[ \left|\hf_e\right| := \frac{\frac{\alpha_e^+ w_e^+}{(s_e^+)^2} + \frac{\alpha_e^- w_e^-}{(s_e^-)^2} }{\frac{w_e^+}{(s_e^+)^2} + \frac{w_e^-}{(s_e^-)^2}} \left|\tf_e\right|
\geq \frac{1}{(1+\theta)^2} \left|\tf_e\right| \,,\]
which also implies that $\left|\rho_e^+\right| = \left|\frac{\tf_e}{s_e^+}\right| \leq (1+\theta)^2 \left|\hrho_e^+\right|$
, $\left|\rho_e^-\right| = \left|\frac{\tf_e}{s_e^-}\right|\leq (1+\theta)^2 \left|\hrho_e^-\right|$.
Note that we now have 
\begin{align*}
\left(\frac{w_e^+}{(s_e^+)^2} + \frac{w_e^-}{(s_e^-)^2}\right) \left|\hf_e\right| =
\left(\frac{\alpha_e^+ w_e^+}{(s_e^+)^2} + \frac{\alpha_e^- w_e^-}{(s_e^-)^2}\right) 
\cdot
	\left|\tf_e\right| 
\leq \left|h_e\right| + \hat{\gamma}\,,
\end{align*}
and thus we get that for all $e\in E$ at least one of  
\begin{align*}
\max\left\{\left|\rho_e^+\right|,\left|\rho_e^-\right|\right\} \leq 
(1+\theta)^2 \cdot \max\left\{\left|\hrho_e^+\right|,\left|\hrho_e^-\right|\right\} < 
(1+\theta)^2 \cdot C_\infty
\end{align*}
and
\begin{align*}
\left|\frac{w_e^+\rho_e^+}{s_e^+}\right|+\left|\frac{w_e^-\rho_e^-}{s_e^-}\right| 
\leq (1+\theta)^2\cdot \left(\left|\frac{w_e^+\hrho_e^+}{s_e^+}\right|+\left|\frac{w_e^-\hrho_e^-}{s_e^-}\right| \right)
\leq 
6\cdot (1+\theta)^2 \cdot \gamma
\end{align*}
is true. Using the fact that $\delta > 10 \cdot \|\ww\|_1^{-1/2}$, the former
becomes
$\max\left\{\left|\rho_e^+\right|,\left|\rho_e^-\right|\right\} = 
\frac{(1+\theta)^2}{2\delta \sqrt{ 2\|\ww\|_1}} < \frac{1}{20}$. For all the remaining edges, we have
\[ \left|\frac{w_e^+\rho_e^+}{s_e^+}\right|+\left|\frac{w_e^-\rho_e^-}{s_e^-}\right| \leq 
6\cdot (1+\theta)^2 \cdot \gamma
= 6\cdot (1+\theta)^2 \cdot \delta^2 \|\ww\|_1 \cdot 32 \sqrt{6}\cdot \log\|\ww\|_1\,,
\]
which, combined with
\begin{align*}
\left|\tf_e\right| \leq \|\tff\|_p 
& \leq \left(\frac{p\cdot \frac{3}{2} \energymax(\hh,\ww,\ss)}{R_p}\right)^{1/p}\\
& \leq \left(\frac{p \cdot \frac{3}{4} \delta^2 \|\ww\|_1}{p \left(10^6 \cdot \delta^2 \|\ww\|_1 \log\|\ww\|_1\right)^{p+1}}\right)^{1/p}\\
& \leq \frac{1}{10^6\cdot \delta^2 \|\ww\|_1 \log\|\ww\|_1}
\end{align*}
gives
\begin{align*} 
\left|\rho_e^+\right| 
& = \left(\left|\tf_e\right|\cdot \left|\frac{\rho_e^+}{s_e^+}\right|\right)^{1/2} \\
& \leq \left(\frac{1}{10^6\cdot \delta^2 \|\ww\|_1 \log \|\ww\|_1} \cdot 6\cdot (1+\theta)^2 \cdot \delta^2\|\ww\|_1 \cdot 32\sqrt{6} \log \|\ww\|_1\right)^{1/2}\\
& < \frac{1}{20}\,.
\end{align*}
Symmetrically, we get that 
$\left|\rho_e^-\right| < \frac{1}{20}$ and so we are done.
\end{proof}

\subsection{Making Progress}
Here we show how to use the flow obtained from optimizing (\ref{eq:strong_reg_obj_max}) in order to achieve centrality for a new parameter $\mu/(1+\delta)$, at the expense of slightly increasing weights from $\ww$ to some $\ww'$.

\begin{lemma}[Almost-centrality after executing step]
\label{lem:strong_central_after_step}
Let $\tffstar$ be the optimizer of (\ref{eq:strong_reg_obj_max}) and let $\tff$ be its restriction to the edges of $G$.
Then 
\begin{align*}
\CC^\top \left(
(1+\delta) \left( \frac{\ww^+}{\onev-\ff} - \frac{\ww^-}{\ff} \right)  
- \left(\frac{\ww^+}{\onev-\ff-\tff}- \frac{\ww^-}{\ff+\tff} \right)
\right) = \CC^\top \cdot \Rp \cdot (\tff)^{p-1}\,.
\end{align*}
\end{lemma}
\begin{proof}
Writing the optimality condition for $(\ref{eq:strong_reg_obj_max})$, and using Lemma~\ref{lem:strong_congestion_bound}, we obtain
\begin{align*}
\CC^\top \left(
(1+\delta) \left( \frac{\ww^+}{\onev-\ff} - \frac{\ww^-}{\ff} \right)  
+ \ww^+ \cdot \nabla_{\tff}\log\left(\onev-\frac{\tff}{\onev-\ff}\right)
+ \ww^- \cdot \nabla_{\tff}\log\left(\onev + \frac{\tff}{\ff}\right)
- \Rp \cdot (\tff)^{p-1}  
\right) = 0\,,
\end{align*}
which we can rewrite equivalently as
\begin{align*}
\CC^\top \left(
(1+\delta) \left( \frac{\ww^+}{\onev-\ff} - \frac{\ww^-}{\ff} \right)  
- \left(\frac{\ww^+}{\onev-\ff-\tff}- \frac{\ww^-}{\ff+\tff} \right)
\right) = \CC^\top \cdot \Rp \cdot (\tff)^{p-1}\,,
\end{align*}
which is what we needed.

\end{proof}

As we can see, the regularization terms have two effects. One is that the update $\tff$ is not exactly a circulation, so this will account for some change in the routed demand. The other effect is that  after augmenting the current flow $\ff$ with $\tff$ we do not obtain a central solution, as shown in  Lemma~\ref{lem:strong_central_after_step}. We proceed to fix this manually by slightly increasing the weights $\ww$. This follows from applying 
Lemma~\ref{lem:perturbation_correction} to the residual $\Delta \hh = -\Rp (\tff)^{p-1}$, which produces a central solution with a new set of weights $\ww' \geq \ww$ such that 
$\|\ww'-\ww\|_1 \leq \Rp \cdot \sum\limits_{e\in E} \left|\tf_e\right|^{p-1}$.

We are now ready to characterize the amount of progress we make in a single iteration of the method previously described.

\begin{lemma}[Progress lemma]
\label{lem:strong_progress_lemma}
Given a $\mu$-central instance, i.e. a flow $\ff$ and balanced weights $\ww$ such that
\[
\CC^\top \left(\frac{\ww^+}{\onev-\ff}- \frac{\ww^-}{\ff}\right) = -\frac{\CC^{\top}\cc}{\mu}\,,
\] in the time require to solve (\ref{eq:strong_reg_obj_max}) we can obtain a $\mu/(1+\delta)$-central instance, i.e. a flow $\ff+\tff$ and weights $\ww' \geq \ww$, such that
\[
\CC^\top \left(\frac{\ww^+}{\onev-\ff-\tff}- \frac{\ww^-}{\ff+\tff}\right) = -(1+\delta)\frac{\CC^{\top}\cc}{\mu}\,,
\]
where
\[
\|\ww'-\ww\|_1 \leq p \cdot 10^{12} \cdot \delta^4 \|\ww\|_1^{2+1/p} \cdot \log^2 \|\ww\|_1 \,,
\]
and
$\tff$ routes a demand $\tdd$ such that
\[
\|\tdd\|_1 \leq 3/2\,.
\]
\end{lemma}
\begin{proof}
Using 
Lemma~\ref{lem:strong_opt_non_aug}
and 
Lemma~\ref{lem:strong_central_after_step}, after solving 
(\ref{eq:strong_reg_obj_max}) we obtain an augmenting flow $\tff$ routing a demand $\tdd$ with 
\[
\|\tdd\|_1 \leq  3 \left( \frac{\|\ww\|_1 \cdot \energymax(\hh,\ww,\ss) }{\Rstar} \right)^{1/2} \leq 
3 \left( \frac{\|\ww\|_1 \cdot \frac{1}{2} \delta^2 \|\ww\|_1}{3\delta^2 \|\ww\|_1^2} \right)^{1/2} < 3/2\,, 
\]  such that
\[
\CC^\top \left(
(1+\delta) \left( \frac{\ww^+}{\onev-\ff} - \frac{\ww^-}{\ff} \right)  
- \left(\frac{\ww^+}{\onev-\ff-\tff}- \frac{\ww^-}{\ff+\tff} \right)
\right) = \CC^\top \cdot \Rp \cdot (\tff)^{p-1}\,.
\]
We then increase the weights $\ww$ to $\ww'$ to make the right-hand side of this identity equal to $\zerov$. Per Lemma~\ref{lem:perturbation_correction}, this increases the weights to $\ww'$ such that 
\[
\|\ww'-\ww\|_1 = \Rp \cdot \sum_{e\in E} \left|\tf_e\right|^{p-1} \leq \Rp \left( m^{\frac{1}{p-1}-\frac{1}{p}} \cdot \|\tff\|_p \right)^{p-1} 
\leq \Rp \cdot\|\ww\|_1^{1/p} \cdot \|\tff\|_p^{p-1}\,,
\]
where we used the fact that $\|\ww\|_1 \geq m$.
From Lemma~\ref{lem:strong_opt_non_aug}, we have that for our specific choice of regularization parameters,
\[
\|\tff\|_p \leq \left(\frac{p\cdot \frac{3}{2}\energymax(\hh,\ww,\ss)}{\Rp}\right)^{1/p}
\leq
 \frac{1}{10^6\cdot \delta^2 \|\ww\|_1 \log\|\ww\|_1}\,.
\]
Therefore 
\begin{align*}
\|\ww'-\ww\|_1
&\leq 
\left(p\cdot \left(10^6\cdot \delta^2 \|\ww\|_1
 \cdot {\log\|\ww\|_1}\right)^{p+1}\right) \cdot\|\ww\|_1^{1/p} \cdot
 \left(
  \frac{1}{10^6\cdot \delta^2 \|\ww\|_1 \log\|\ww\|_1}
 \right)^{p-1}
 \\
 &= p \cdot \left(10^6\cdot \delta^2 \|\ww\|_1 \cdot \log \|\ww\|_1\right)^2 \cdot \|\ww\|_1^{1/p}
 = p \cdot 10^{12} \cdot \delta^4 \|\ww\|_1^{2+1/p} \cdot \log^2 \|\ww\|_1\,.
\end{align*}
\end{proof}

\subsection{Wrapping Up}
We can now give the main statement of this section, which follows from running the interior point method, based on the guarantee provided by Lemma~\ref{lem:strong_progress_lemma}.

\begin{lemma}\label{lem:strong_near_opt_flow_perturbed_demand}
Suppose we have a $\mu$-central instance with weights $\ww \geq \onev$, where 
$\|\ww\|_1 \leq 2m + 1$ and 
$\mu = m^{O(1)}$. Let $\epsilon = m^{-O(1)}$, and let $\delta = m^{-(1/3+o(1))}$. In time dominated by $\widetilde{O}(\delta^{-1})$ iterations of the procedure described in Lemma~\ref{lem:strong_progress_lemma} we obtain an instance
with duality gap at most $\epsilon$ with a total demand perturbation of 
$\widetilde{O}(\delta^{-1})$.
\end{lemma}
\begin{proof}

We perform a sequence of iterations as described in Lemma~\ref{lem:strong_progress_lemma}. These are 
interspersed with calls to the weight balancing procedure (Lemma~\ref{lem:weight_bal}), required in order to maintain the invariant needed by Lemma~\ref{lem:strong_congestion_bound}.

We will argue that within $T = \widetilde{O}(\delta^{-1} \log (m \mu \epsilon^{-1}))$ iterations the barrier weights will always satisfy $\|\ww\|_1 \leq 3m$.
We need to take into account the total increase guaranteed by Lemma~\ref{lem:strong_progress_lemma} together with the possible weight increases caused by the weight balancing procedure.

First let us bound the total number of balancing operations, since each of these can increase a single weight by a significant amount.
First, from Definition~\ref{def:balancing_proc} we see that once an edge gets balanced it will never become unbalanced again, as weights are monotonic.

Furthermore, we see that such an operation can only occur when the largest of the two paired weights is at least $\delta \|\ww\|_1 \geq 2m \delta$, since we maintain $\ww \geq \onev$. Under the invariant that throughout the entire algorithm $\|\ww\|_1 \leq 3m$, we therefore see that this can only happen at most $(3m)/(2m\delta) = 3/(2\delta)$ times. Invoking Lemma~\ref{lem:weight_bal} we therefore get that the total weight change caused by these operations is at most 
\[
\frac{3}{2\delta} \cdot 96 \cdot \delta^4 \cdot (3m)^2 = 1296 \cdot \delta^3 m^2\,.
\]
In addition, we incur weight increases due to the progress steps; per Lemma~\ref{lem:strong_progress_lemma}, within each of the $T$ iterations, $\|\ww\|_1$ increases by at most
\begin{align*}
p \cdot 10^{12} \cdot \delta^4 (3m)^{2+1/p} \cdot \log^2 (3m)\,.
\end{align*}
Therefore the total weight increase over $T$ iterations is at most
\[
\widetilde{O}\left( \left(p \delta^3 m^{2+1/p} \right) \log(m\mu\epsilon^{-1}) \right)
\]
which is $o(m)$ as long as 
\[
\delta \leq 1/\left( p^{1/3}m^{ (1+1/p)/3 + o(1) } \right)
= 1 / m^{1/3 + o(1)}\,.
\] 
Thus this specific choice of $\delta$ insures that the invariant $\|\ww\|_1 \leq 3m$ was satisfied throughout the entire run of the algorithm.

Finally, we bound the total perturbation in demand suffered by the flow we maintain. Each of the $\widetilde{O}(\delta^{-1})$ iterations perturbs the demand by at most $3/2$ in $\ell_1$ norm, and the weight balancing operations may perturb the demand by an additional $\frac{3}{2\delta}$ overall.
Summing up we obtain the desired claim.
\end{proof}

This enables us to obtain a running time of $m^{4/3+o(1)}$ for minimum cost flow in unit-capacity graphs. The proof is identical to that of Theorem~\ref{th:main_theorem}, we use scaling to obtain a logarithmic dependence in $W$, and resort to the fixing procedure from Section~\ref{sec:repair} to repair the demand perturbation. The time required to implement each iteration of the interior point method is dominated by the time required by one call to the solver in Theorem~\ref{lem:strong_solver}, which is $m^{1+o(1)}$ by our choice of parameters.
\begin{theorem}\label{th:strong_main_theorem}
Given a directed graph $G(V,E,\cc)$ with $m$ arcs and $n$ vertices, such that $\|\cc\|_\infty \leq W$, and a demand vector $\dd \in \bbZ^{n}$,  in $m^{4/3+o(1)} \log W$ time we can obtain a flow $\ff$ which routes $\dd$ in $G$ while satisfying the capacity constraints $\zerov \leq \ff \leq \onev$ and minimizing the cost $\sum_{e\in E} c_e f_e$, or certifies that no such flow exists.
\end{theorem}
 
\appendix

\section{Initializing the Interior Point Method} \label{sec:init_appendix}
We require finding an initial $\ff_0$ for which the solution can be easily brought to centrality. To do so, we modify the original graph by adding a set of $O(m)$ arcs with high cost $\cinf = (m+1)\|\cc\|_\infty$ such that the flow which pushes exactly $1/2$ units on every arc in the modified graph routes the original demand $\dd$.

The consequence of adding such these edges is that, while the solution to the original problem remains unchanged, since it will never be beneficial to use an arc with cost $\cinf$ in the optimal solution. Meanwhile initializing $\ff_0 = \onev/2 $ ensures that the flow on each arc is equally far from the upper and the lower barrier, and therefore their contributions to the gradient will cancel. This will make centering trivially easy.

First let us show that such an augmentation is indeed possible.
\begin{lemma}\label{lem:augmented_graph_initialized}
Let $G=(V,E,\cc)$ be a directed graph with $\vert E \vert = m$ arcs with unit capacity, $\vert V \vert = n$ vertices, costs on arcs $\cc \geq 0$, and let $\dd \in \bbZ^n$ be a demand vector $\sum_{i=1}^n d_i = 0$.

Then there exists a graph $G' = (V', E', \cc')$ with at most $2m$ unit-capacity arcs, and a demand vector $\dd'$ such that the flow $\ff' = \onev/2$ routes the demand $\dd'$ in $G'$. 

Furthermore if $(\ff')^*$ is a flow in $G'$
with $\zerov \leq (\ff')^* \leq \onev$ 
and which routes $\dd'$ 
such that the cost 
$\langle \cc', \ff' \rangle$ is minimized,
then one can convert it in $O(m)$ time into a flow 
$\ff^*$ which routes $\dd$ in $G$, such that $\zerov \leq \ff^* \leq \onev$ and the cost
 $\langle \cc, \ff \rangle$ is minimized,
 or certify that no such flow exists.
\end{lemma}
\begin{proof}
First we construct the graph $G'$.
Let $V' = V \cup \{v_0\}$, where $v_0$ is a new vertex.
Initialize $S = \emptyset$. 
For each vertex $v\in V$, let 
$\ell(v) = d_v - \frac{1}{2}(\lvert E^-(v)\lvert-\lvert E^+(v)\lvert)$,
representing the excess flow at vertex $v$ after routing $\ff = 1/2$ on each arc.
Let $\cinf = (m+1)\|\cc\|_\infty$.
For each $v$ where $\ell(v) > 0$, create $ 2 \ell(v) $ arcs $(v,v_0)$ with cost $\cinf$, and add them to $S$. Similarly, for each $v$ where $\ell(v) <0$
create $-2 \ell(v)$ arcs $(v_0, v)$ with cost $\cinf$ and add them to $S$.
Note that $2\ell(v)$ is an integer, since $\dd$ is integral.
Let $E' = E \cup S$ and $\cc'$ be the corresponding cost vector where arcs in $E$ preserve their original cost $\cc$, while those in $S$ have cost $\cinf$.
Let $d'_v = d_v$  for all $v\in V$ and $d'_{v_0} = 0$.

Now let $(\ff')^*$ be the minimum cost flow in $G'$ which satisfies capacity constraints. If $(\ff')^*$ is not supported on any arcs in $S$, then $(\ff')^*$ is also a feasible flow in $G$. Furthermore it must be the optimal flow in $G$, since otherwise $(\ff')^*$ would not have been optimal in $G'$.
If $(\ff')^*$ has nonzero flow on some arc in $S$, then we must conclude that it is impossible to route $\dd$ in $G$ while satisfying capacity constraints.

Suppose it were possible to do so using a flow $\ff^*$. Then, consider the circulation $\gg = \ff^* - (\ff')^*$.
Now convert $\gg$ into a \textit{minimal} circulation $\tgg$ which preserves the flows on $S$ as follows: while $\gg$ contains a cycle without any arcs in $S$, decrease the value of all flows along that cycle by the minimum value of the flow along it. This operation always zeroes out the flow on at least one edge, so the process must finish. Furthermore note that the cost of any such cycle must be non-positive, since otherwise we contradict the optimality of $\ff^*$.

At this point we are left with a circulation $\tgg$ which does not contain any cycles supported only in $E$, and whose cost is at at least $\langle \cc, \gg \rangle \geq 0$; the latter inequality follows from the optimality of $(\ff')^*$.

Using the fact that $\tgg$ does not contain cycles supported in $E$, we get that the restriction to the arcs in $E$, $\tgg_{|E}$ satisfies $\|\tgg_{|E}\|_1 \leq m \cdot \|\tgg_{|S}\|_1$, where $\tgg_{|S}$ is the corresponding restriction to $S$. Therefore $\lvert \langle \cc, \tgg_{|E} \rangle \rvert \leq \|\cc\|_\infty \cdot \|\tgg_{|E}\|_1 \leq m\cdot \|\cc\|_{\infty} \cdot \|\tgg_{|S}\|_1$.
Since by definition $\langle \cc, \tgg_{|S} \rangle = -\cinf \cdot \|\tgg_{|S}\|_1$, we conclude that $\langle \cc, \tgg \rangle =  \langle \cc, \tgg_{|E}\rangle + \langle \cc, \tgg_{|S} \rangle \leq m \cdot \|\cc\|_\infty \cdot \|\tgg_{|S}\|_1 - \cinf \cdot \|\tgg_{|S}\|_1 = \|\tgg_{|S}\|_1 \cdot (m\|\cc\|_\infty - \cinf) < 0$. This yields a contradiction, so a feasible $\ff^*$ can not possibly exist.
\end{proof}

At this point, using the reduction given above, we can assume without loss of generality that the flow $\ff = 1/2\cdot \onev$ routes the demand $\dd$.

\begin{lemma}\label{lem:initialization}
If $\ff = 1/2\cdot\onev$ routes the input demand $\dd$, then in the time dominated by $O(\log \log m)$ residual correction steps, we can produce a solution $\ff = \CC\xx$ and a set of weights $\ww \geq \onev$ such that $\|\ww\|_1 \leq 2m+1$ and $\nabla F_{\mu}^{\ww}(\xx) = \zerov$, for $\mu \leq 2\|\cc\|_2$.
\end{lemma}
\begin{proof}
For this flow $\ff = \CC\xx$ we have that $\ss^+ = \ss^- = \onev/2$ and the corresponding residual satisfies 
$\nabla F_{\mu}^{\onev}(\xx) = \frac{\CC^\top \cc}{\mu}$.
Per Definition~\ref{def:energy} we can certify an upper bound on 
$\energy{\onev, \ss}{\nabla F_{\mu}^{\onev}(\xx)}$ using $\yy = (\zerov; \cc/\mu)$ which shows that
\[
\energy{\onev, \ss}{\nabla F_{\mu}^{\onev}(\xx)} \leq \frac{1}{2} \cdot 
\left\langle  (\ss^+; \ss^-)^2, (\zerov; \cc/\mu)^2 \right\rangle
= \frac{1}{2} \cdot \frac{1}{4 \mu^2} \cdot \|\cc\|_2^2\,.
\]
Therefore setting $\mu = \|\cc\|_2$ we have that the energy is at most $1/8$, which fulfills the conditions required by Corollary~\ref{cor:energy_contract}
and Lemma~\ref{lem:fine_correction} to produce in the time dominated by $O(\log\log m)$ residual correction steps a solution $\xx$ and 
a weight vector $\ww \geq \onev$ such that $\|\ww - \onev\|_1 \leq 2 m^{-9}\leq 1$ for which 
$\nabla F_{\mu'}^{\ww}(\xx) = \zerov$, where $\mu' \leq \|\cc\|_2(1+m^{-10}) \leq 2\|\cc\|_2$.
\end{proof}

\section{Preconditioning Proof}
\label{sec:prec_proof}
In this section we provide the proof for Lemma~\ref{lem:prec_lemma}.

\begin{proof}[Proof of Lemma~\ref{lem:prec_lemma}]
Let $\tffstar = \tff + \tff'$.
Writing the optimality condition for (\ref{eq:regularized_newton}) we obtain
\begin{align}\label{eq:cycle_condition}
\CCstar^\top (\rr \tffstar - \hh) = \zerov\,,
\end{align}
where
\begin{align*}
\rr = \begin{cases}
\frac{\ww^+}{(\ss^+)^2} + \frac{\ww^-}{(\ss^-)^2} + \Rp \cdot (\tff)^{p-2} & \textnormal{for edges in $E$,} \\
\Rstar  + \Rp \cdot (\tff')^{p-2} & \textnormal{for edges in $E'$.}
\end{cases}
\end{align*}
Since $\CCstar$ is a cycle basis for $\Gstar$, the condition (\ref{eq:cycle_condition}) implies that along any cycle in $G_\star$ the (appropriately signed) sum of weights $\rr\tffstar-\hh$  is exactly $0$. This means that these weights are determined by an embedding of the vertices of $\Gstar$ onto the line, c.f. Lemma~\ref{lem:cycle_space_lemma}. Hence there exists a vector 
$\vphi\in\mathbb{R}^{\left|V\cup\{v_\star\}\right|}$ such that
for every edge $(u,v)\in E\cup E'$:
\begin{align}\label{eq:circulation_identity}
r_e (\tfstar)_e - h_e = \phi_u - \phi_v\,.
\end{align}
Also, to shorten notation, let us define
\begin{align}
\oww = \ww^+ + \ww^-\,.
\end{align}
Next we will prove that the coordinates of $\vphi$ must lie within a short interval. The intuition here relies on the fact that the preconditioning edges make the graph $\Gstar$
have good enough expansion; in turn, using an argument similar to~\cite{kelner2009electric} which was subsequently employed under various forms in~\cite{kelner2014almost,madry2013navigating,cmsv17} we argue that good expansion means that all the vertices are close to each other in the potential embedding.

Given a scalar $t$, let $S_t$ be the set of edges $(u,v)$ for which $t \in (\min\{\phi_u, \phi_v\}, \max\{\phi_u,\phi_v\})$.
Our proof proceeds by lower bounding for each $t$:
\begin{align}\label{eq:inv_pot_diff_lb}
\sum_{ (u,v) \in S_t} \frac{\ow_{uv}}{\vert \phi_u - \phi_v \vert}
\geq 
\frac{  \left(\sum_{(u,v) \in S_t} \frac{\sqrt{\ow_{uv}}}{\sqrt{r_{uv}}}  \right)^2 }{ \sum_{(u,v) \in S_t}  \frac{\vert \phi_u - \phi_v \vert }{r_{uv}} }\,,
\end{align}
which we obtained using Cauchy-Schwarz.

Next we observe that the quantity in the denominator is upper bounded by
\begin{align}\label{eq:cut_flows_ineq}
\sum_{(u,v) \in S_t}  \frac{\vert \phi_u - \phi_v \vert }{r_{uv}} 
&\leq \sum_{(u,v)\in S_t} \frac{\vert h_{uv}\vert}{r_{uv}}\,.
\end{align}
We can see why this is true by using the fact that $\tffstar$ is a circulation, therefore along any cut $S_t$ one has 
\begin{align*}
\sum_{\substack{(u,v)\in S_t\\ \phi_u > \phi_v}} (\tfstar)_{uv}
=
\sum_{\substack{(u,v)\in S_t\\ \phi_u \leq \phi_v}} (\tfstar)_{uv} 
\,,
\end{align*}
i.e. the sum of the values of flows going from left to right is equal to the sum of values of flows going from right to left in the embedding.

By substituting $(\tf_\star)_{uv}$ with the value determined from (\ref{eq:circulation_identity})
we equivalently obtain that
\begin{align*}
\sum_{\substack{(u,v)\in S_t \\ \phi_u > \phi_v}} \frac{h_{uv} +\vert\phi_u - \phi_v\vert}{r_{uv}} 
=
\sum_{\substack{(u,v)\in S_t \\ \phi_u \leq \phi_v}} \frac{h_{uv} -\vert \phi_u - \phi_v\vert}{r_{uv}} \,,
\end{align*}
and by rearranging
\begin{align*}
\sum_{(u,v)\in S_t} \frac{ \vert\phi_u-\phi_v\vert }{r_{uv}} = \sum_{\substack{(u,v)\in S_t \\ \phi_u > \phi_v } }\frac{-h_{uv}}{r_{uv}}
+
\sum_{\substack{(u,v)\in S_t \\ \phi_u \leq \phi_v } }\frac{h_{uv}}{r_{uv}}
\leq \sum_{(u,v) \in S_t} \frac{\vert h_{uv}\vert}{r_{uv}}\,.
\end{align*}
Therefore, plugging (\ref{eq:cut_flows_ineq}) into (\ref{eq:inv_pot_diff_lb}), we obtain
\begin{align}\label{eq:lower_bound_eaten_edges}
\sum_{ (u,v) \in S_t} \frac{\ow_{uv}}{\vert \phi_u - \phi_v \vert}
\geq 
\frac{  \left(\sum_{(u,v) \in S_t} \frac{\sqrt{\ow_{uv}}}{\sqrt{r_{uv}}}  \right)^2 }{ \sum_{(u,v) \in S_t}  \frac{\vert h_{ uv}\vert}{r_{uv}} } 
\geq 
\frac{
\sum_{(u,v) \in S_t} \frac{\sqrt{\ow_{uv}}}{\sqrt{r_{uv}}}  
}
 {  \max \frac{\vert h_{uv} \vert / r_{uv}}{\sqrt{\ow_{uv}}/\sqrt{r_{uv}}} }
=
 \left(
\sum_{(u,v) \in S_t} \frac{\sqrt{\ow_{uv}}}{\sqrt{r_{uv}}}  
\right)
\cdot \frac{1}{\|\hh \oww^{-1/2} \rr^{-1/2}\|_\infty}\,.
\end{align}
At this point we can prove that the all the values in $\vphi$ lie within a small interval. In order to do so we will crucially use the augmenting edges, 
which endow $\Gstar$ with better expansion properties.
Suppose w.l.o.g. that $\phi_{\vstar} = 0$. For every edge $e = (u,v) \in E \cup E'$, let $\lphi_e = \min \{\phi_u, \phi_v\}$, $\uphi_e = \max\{\phi_u, \phi_v\}$.
For all edges in $e = (u,\vstar)\in E'$ (of which multiple copies can occur) we define 
\[
\ow_{e} = 1\,.
\]
Hence
we may define for $t\geq 0$:
\begin{align}\label{def:F_slide}
F(t) = \sum_{ \substack{ (u,v) \in E \cup E' \\ \uphi_{uv} \geq t  }}
\ow_{uv}\cdot \frac{ \uphi_{uv} - \max\{\lphi_{uv},  t\} }{\vert\phi_u - \phi_v\vert}\,,
\end{align}
which represents the sum of weighted fractions of edges that are on the right side of the cut at position $t$ on the real line. 

Our goal is to show that since $F(t)$ decreases very fast, we do not need to increase $t$ very much before we run out of edges i.e. $F(t)$ becomes $0$. Indeed, (\ref{eq:lower_bound_eaten_edges}) offers a lower bound on the instantaneous decrease of $F(t)$, as $t$ increases. The reason is that all the augmenting edges $(\vstar, u)$ for $\phi_u > t$ appear in the cut.
This also means that $F(t) > 0$ if $S_t\neq \emptyset$ and $F(t) = 0$ otherwise.

Intuitively (\ref{eq:lower_bound_eaten_edges}) states that when slightly increasing $t$, $F(t)$ must decrease by a constant factor, scaled by the change in $t$. To formalize this we simply use the fact that 
\begin{align}\label{eq:pp2}
\sum_{ (u,v) \in S_t} \frac{\sqrt{\ow_{uv}}}{\sqrt{r_{uv}}}
= \sum_{ (u,v) \in S_t} \frac{\ow_{uv}}{\sqrt{\ow_{uv} r_{uv}}}
\geq
 \sum_{(u,\vstar) \in S_t} \frac{\ow_{u\vstar}}{\sqrt{\ow_{u\vstar} r_{u\vstar}}}
\geq
\frac{1}{ \max_{e\in E'} \sqrt{\ow_e r_e}} \cdot
\frac{F(t)}{2}\,,
\end{align}
which follows from the inequality 
\[
F(t) \leq \sum_{e\in S_t} \ow_e
\leq 2\cdot \sum_{e \in S_t \cap E'} \ow_e \,.
\]
The latter is ensured by the fact that
by definition
each vertex $v\in V$ is incident to at least $\sum_{e \in E : e \sim v} \ow_e$ augmenting edges in $E'$.
Furthermore, even if $F(t)$ is very small but $S_t$ is still nonempty, we can use the lower bound 
\begin{align}\label{eq:pp3}
\sum_{(u,v) \in S_t} \frac{\sqrt{\ow_{uv}}}{\sqrt{r_{uv}}}
=
\sum_{(u,v) \in S_t} \frac{{\ow_{uv}}}{\sqrt{\ow_{uv} r_{uv}}}
 \geq 
 \frac{\vert S_t \vert}{ \max_{e\in E'} \sqrt{\ow_e r_e} }
 \geq
 \frac{1}{ \max_{e\in E'} \sqrt{\ow_e r_e} }
 \,.
\end{align}
Combining (\ref{eq:lower_bound_eaten_edges}), (\ref{eq:pp2}), and (\ref{eq:pp3}) we obtain that
if $S_t$ is nonempty, or equivalently $F(t) > 0$, then
\begin{align}\label{eq:finalineqp}
\sum_{(u,v)\in S_t} \frac{\ow_{uv}}{\vert\phi_u - \phi_v\vert} \geq 
\frac{ 1 }{ \|\hh \oww^{-1/2} \rr^{-1/2}\|_\infty \cdot \max_{e \in E'}\sqrt{\ow_e r_e}} \cdot \max\{F(t)/2, 1\}\,.
\end{align}
As a matter of fact, this tells us that $F$ must decrease very fast, since from the definition (\ref{def:F_slide}) we have that $F$ is a continuous function, differentiable almost everywhere, such that at all points $t\geq 0$ where it is differentiable it satisfies
\[
\frac{d}{dt}F(t) = 
-\sum_{ \substack{ (u,v) \in E \cup E' \\ \uphi_{uv} \geq t  }}
\frac{\ow_{uv}}{\vert\phi_u - \phi_v\vert}
\leq
-\sum_{(u,v) \in S_t} \frac{\ow_{uv}}{\vert\phi_u -\phi_v\vert} \,.
\]
Thus using (\ref{eq:finalineqp}) and Lemma~\ref{lem:diff} we obtain that $F(T) = 0$ for
\begin{align}
T &= \|\hh \oww^{-1/2} \rr^{-1/2}\|_\infty \cdot \max_{e \in E'}\sqrt{\ow_e r_e} \cdot(1 + 2\log F(0)) \notag
\\
&\leq
\|\hh \oww^{-1/2} \rr^{-1/2}\|_\infty \cdot \max_{e \in E'}\sqrt{\ow_e r_e} \cdot(1 + 2\log (4\|\ww\|_1)) \label{eq:max_embedding}
\,,
\end{align}
where the last inequality follows from accounting the weights of the augmenting edges in $E'$.

Using the identical argument for vertices embedded to the left of $\vstar$, we conclude that (\ref{eq:max_embedding}) yields an upper bound on $\uphi_e - \lphi_e$, and therefore, for all 
edges $e \in E$
\begin{align*}
\vert r_e \tf_e - h_e \vert
&\leq 
{\max_{e\in E'} \sqrt{\ow_e r_e}}\cdot{\|\hh \oww^{-1/2} \rr^{-1/2}\|_\infty} \cdot 2(1+2\log(4 \|\ww\|_1))
\\
&\leq
{\max_{e\in E'} \sqrt{\ow_e r_e}}\cdot{\|\hh \oww^{-1/2} \rr^{-1/2}\|_\infty} \cdot {32} \log \|\ww\|_1
\,.
\end{align*}

Now we use Lemma~\ref{lem:opt_non_aug} to obtain that for all $(u,\vstar) = e \in E'$:
\[
\ow_e r_e = \Rstar + \Rp \cdot (\tf_e')^{p-2}
\leq \Rstar + \Rp \cdot \|\tff'\|_{\infty}^{p-2}\,,
\]
which ensures that 
\begin{align}
\left\vert \left( \frac{w_e^+}{(s_e^+)^2} + \frac{w_e^-}{(s_e^-)^2} + \Rp (\tf_e)^{p-2} \right) \tf_e -h_e \right\vert
&\leq 
\left(
\Rstar + \Rp \cdot \|\tffstar\|_\infty^{p-2}
\right)^{1/2} \cdot \left\Vert \frac{\hh}{\sqrt{\oww \rr}} \right\Vert_\infty \cdot {32 \log \|\ww\|_1}\,,
\end{align}
and therefore that 
\begin{align}
\left\vert \left( \frac{w_e^+}{(s_e^+)^2} + \frac{w_e^-}{(s_e^-)^2}  \right) \tf_e -h_e \right\vert
&\leq 
\left\vert \Rp \cdot (\tf_e)^{p-1} \right\vert 
+ 
\left(
\Rstar + \Rp \cdot \|\tffstar\|_\infty^{p-2}
\right)^{1/2} 
\\
&\cdot \left\Vert \frac{\hh}{\sqrt{(\ww^+ + \ww^-)\left( \frac{\ww^+}{(\ss^+)^2}+\frac{\ww^-}{(\ss^-)^2} \right)}} \right\Vert_\infty
\cdot {32 \log \|\ww\|_1} \,,
\end{align}
which is what we needed.

We can furthermore establish a similar upper bound on 
\begin{align*}
\left\vert  \frac{w_e^+}{(s_e^+)^2} \cdot \tf_e  \right\vert
&\leq
\left\vert \left( \frac{w_e^+}{(s_e^+)^2} + \frac{w_e^-}{(s_e^-)^2} + \Rp \cdot \left(\tf_e\right)^{p-2}\right) \tf_e \right\vert
\\
&\leq
\left\vert h_e  \right\vert
+ 
{\max_{e\in E' }\sqrt{\ow_e r_e}}\cdot{\|\hh \oww^{-1/2} \rr^{-1/2}\|_\infty} \cdot {32 \log \|\ww\|_1 } \\
&\leq 
\left\vert h_e \right \vert +
\left(
\Rstar + \Rp \cdot \|\tffstar\|_\infty^{p-2}
\right)^{1/2} 
\\
&\cdot \left\Vert \frac{\hh}{\sqrt{(\ww^+ + \ww^-)\left( \frac{\ww^+}{(\ss^+)^2}+\frac{\ww^-}{(\ss^-)^2} \right)}} \right\Vert_\infty
\cdot{32 \log \|\ww\|_1} \,,
\end{align*}
and an identical upper bound on $\left\vert  \frac{w_e^-}{(s_e^-)^2} \cdot \tf_e  \right\vert$.

\end{proof}

\begin{lemma}
\label{lem:cycle_space_lemma}
Let $G$ be a graph with $n$ vertices and $m$ edges, and 
$\CC$ be a matrix that encodes a cycle basis for $G$, in the sense that for any $\xx$, $\CC \xx$ is a circulation in $G$ and
for any circulation $\ff$ in $G$, there exists a vector $\xx$ such 
that $\ff = \CC \xx$. Suppose that $\yy \in \bbR^m$ is a vector such that $\CC^\top \yy = 0$. Then
there exists a vector $\vphi \in \bbR^n$ such that for all $(u,v) \in E$ one has that $y_{uv} = \phi_u - \phi_v$.
\end{lemma}
\begin{proof}
By definition, the image of $\CC$ is the space of circulations in $G$. Therefore the kernel of $\CC^{\top}$ is orthogonal to the space of circulations in $G$, and therefore so is the vector $\yy$. 
Now consider the incidence matrix $\BB \in \bbR^{m\times n}$, which is constructed as follows: for each edge $(u,v) \in E$, add a row in with two nonzero entries, $+1$ at position $u$, and $-1$ at position $v$. One can easily see that $\ker \BB^{\top}$ is exactly the space of circulations, as the $\BB^{\top}$ operator acts on flows by returning the vector of demands that they route. Hence $\yy$ lies in the image of $\BB$, i.e. $\yy = \BB \vphi$. By the definition of $\BB$, the conclusion follows.
\end{proof}

\begin{lemma}
\label{lem:diff}
Let $F : \bbR_{\geq 0} \rightarrow \bbR_{\geq 0}$ be a continuous function, differentiable almost everywhere, such that $F(0)\geq 0$. Suppose that 
for all $t \geq 0$ where $F(t) > 0$ and $dF(t)/dt$ exists, we have
\[
\frac{d}{dt} F(t) \leq -\frac{1}{\alpha} \max\left\{\frac{F(t)}{2}, 1 \right\}\,,
\]
for some $\alpha > 0$. Then
\[
F(\alpha(1+2\log F(0))) = 0\,.
\]
\end{lemma}
\begin{proof}
Let $T > 0$ be any point for which $F(T) > 0$.
From the hypothesis we know that the instantaneous decrease in $F$ at all points $t\in[0,T]$
is at least $dt/\alpha$. Hence we have that:\begin{align*}
F(T) \leq \min_{0 \leq t \leq T} F(t) - \frac{T-t}{\alpha}\,.
\end{align*}
Furthermore, using the fact that the instantaneous decrease in $F(t)$ is at least
$F(t) / (2\alpha)$, we solve the corresponding ODE to obtain that for all $t$,
\begin{align*}
F(t) \leq F(0) \exp\left(-\frac{t}{2\alpha}\right)\,.
\end{align*}
Combining the two inequalities, and setting $t = 2\alpha \log F(0)$, we get that
\[
F(T) \leq 1 - \frac{T - 2\alpha \log F(0)}{\alpha} = (1+2\log F(0)) - \frac{T}{\alpha}\,.
\]
which implies that $T \leq \alpha (1+2\log F(0))$, since $F(T) \geq 0$.
\end{proof}

Finally, we discuss how to adapt the proof of Lemma~\ref{lem:prec_lemma} to obtain the Lemma~\ref{lem:strong_prec_lemma}.

\begin{proof}[Proof of Lemma~\ref{lem:strong_prec_lemma}]
This proof is identical to that of Lemma~\ref{lem:prec_lemma}. The main difference consists of including the $\valpha$ factors in the definition of $\rr$ for edges in $e$, more specifically, we let:
\begin{align*}
\rr = \begin{cases}
\frac{\valpha^+ \ww^+}{(\ss^+)^2} + \frac{\valpha^- \ww^-}{(\ss^-)^2} + \Rp \cdot (\tff)^{p-2} & \textnormal{for edges in $E$,} \\
\Rstar  + \Rp \cdot (\tff')^{p-2} & \textnormal{for edges in $E'$.}
\end{cases}
\end{align*}
Using this new definition for $\rr$, the remaining proof carries over by following the exact same steps as before.
\end{proof}

\section{Solving the Mixed Objective}

\subsection{Invoking the Solver}

Producing a high precision solution to the regularized objective in (\ref{eq:regularized_newton}) can be done efficiently in our particular setting, where we aim to optimize a mixed $\ell_2$-$\ell_p$ objective in the space of circulations. To this extent we use the following result from~\cite{kyng2019flows}, also restated and improved in~\cite{adil2020faster}.
\begin{theorem}\label{thm:mixedsolver}
For any $p\geq 2$, given weights $\rr \in \bbR_{\geq 0}^{|E|}$, and a cost vector $\gg \in \bbR^{|E|}$ define the function defined over circulations $\ff$ in $G$:
$$val(\ff) = \sum_e g_e f_e + r_e f_e^2 + \vert f_e \vert^p\,.$$
Given any circulation $\ff$ for which all the parameters are bounded by $2^{(\log n)^{O(1)}}$ we can compute a circulation $\tff$ such that
$$val(\tff)-OPT \leq \frac{1}{2^{(\log m)^{O(1)}}} (val(\ff) - OPT) + \frac{1}{2^{(\log m)^{O(1)}}}$$
in $2^{O(p^{3/2})} m^{1+O(1/\sqrt{p})}$ time.
\end{theorem}

For large values of $p$ this solves the regularized objective defined in (\ref{eq:regularized_newton}) to high precision in almost linear time $O\left(m^{1+o(1)}\right)$, which is comparably fast to the time required to minimize a convex quadratic function in the space of circulations via fast Laplacian solvers.

\subsection{Discussion on Solver Errors}
\label{sec:solver_error}
Throughout the paper we assume that the solutions to the regularized objective are exact. This is not exactly true due to the approximate nature of the solver specified in Theorem~\ref{thm:mixedsolver}. Instead, we argue that the entire analysis we showed carries over even if the solver returns a solution which carries some small error.

Consider the equivalent problem of minimizing the negative of the objective $\Phi(\xx)$ stated in (\ref{eq:regularized_newton}). We verify that at a point $\xx$ such that the current flow $\tff = \CCstar \xx$, its Hessian
is
\[
\nabla^2 \Phi(\xx) = \CCstar^\top \DD \CCstar
\]
where $\DD$ is a diagonal matrix whose entries are defined such that
\[
(\DD)_{e,e} =
\begin{cases}
\frac{w_e^+}{(s_e^+)^2} + \frac{w_e^-}{(s_e^-)^2}
+ (p-1)\Rp  (\tf_e)^{p-2}
\quad&\textnormal{if }e \in E\,, \\
\Rstar + \Rp (p-1) (\tf_e)^{p-2} \quad &\textnormal{if } e\in E'\,.
\end{cases}
\]
Our choice of regularization parameters yield an upper bound on $\|\tff\|_\infty$ as we showed in the proof of Lemma~\ref{lem:opt_non_aug} (Equation (\ref{eq:fstar_pnorm})), which together with our choice of regularization parameters (Section~\ref{sec:parameters}) and the invariants that $\onev \leq \ww$ and  $\|\ww\|_1 = O(m)$ ensure that $\Rp (\tf_e)^{p-2} \leq \widetilde{O}((\delta^2 m)^3) = o(m)$.
Assuming that we always maintain all our slacks large enough i.e. for all $e$: $\tau^{-1} \leq f_e \leq 1-\tau^{-1}$, for $\tau = m^{O(1)}$, we see that $\DD$ is always well-conditioned in the sense that
all of its diagonal entries are between $\tau^{-2}$ and $M = O(m \tau^{2})$.
We will discuss later how to maintain this property.\footnote{Alternatively, we can enforce lower and upper bounds on the entries of $\DD$ by adding an additional small quadratic regularizer to the edges in $E$.}

Under this assumption we can therefore use basic tools from convex analysis to argue that after optimizing $\Phi$ to high precision, the gradient of the returned solution is small. Let $\xx$ be the solution returned by the high-precision solver such that $\Phi(\xx) \leq \Phi(\xx^*) + \epsilon$. Letting $\nabla \Phi(\xx) = \CCstar^\top \hh$ for some vector $\hh$, we can write

\begin{align*}
\| \nabla \Phi(\xx) \|_{\CCstar^\top \CCstar}^*
&=
\left\| \int_0^1 \nabla^2 \Phi((1-t) \xx^*+ t \xx) (\xx- \xx^*)  dt  \right\|_{\CCstar^\top \CCstar}^* \\
&\leq \max_{t \in [0,1]}
\left\|
\nabla^2 \Phi((1-t)\xx^* + t\xx  )
\right\|_{\|\cdot\|_{\CCstar^\top \CCstar} \rightarrow \|\cdot\|_{\CCstar^\top \CCstar}^*}
\cdot
\left\|\xx-\xx^*\right\|_{\CCstar^\top \CCstar}
\\
&\leq \max_{\substack{\DD \textnormal{ diagonal} \\ \DD_{ii} \in [M^{-1}, M] }}
\left\|
\CCstar^\top \DD \CCstar 
\right\|_{\|\cdot\|_{\CCstar^\top \CCstar} \rightarrow \|\cdot\|_{\CCstar^\top \CCstar}^*}
\cdot 
\left\|\xx-\xx^*\right\|_{\CCstar^\top \CCstar}\,.
\end{align*}

Here we use the notation 
\[
\|\xx\|_{\CCstar^\top \CCstar} = \|\CCstar\xx\|_2\,,
\]
\[
\|\ff\|_{\CCstar^\top \CCstar}^* = \max_{\xx : \|\xx\|_{\CCstar^\top \CCstar} \leq 1} \langle \ff, \xx \rangle\,,
\] 
and 
\[
\|\AA\|_{\|\cdot\|_{\CCstar^\top \CCstar} \rightarrow \|\cdot\|_{\CCstar^\top\CCstar}^*} =
 \max_{\xx: \|\xx\|_{\CCstar^\top \CCstar} \leq 1} \| \AA \xx \|_{\CCstar^\top \CCstar}^*\,.
\]

Using these definitions we can further write
\begin{align*}
\|\CCstar^\top \DD \CCstar\|_{\|\cdot\|_{\CCstar^\top \CCstar} \rightarrow \|\cdot\|_{\CCstar^\top \CCstar}^*} &= \max_{\xx : \|\xx\|_{\CCstar^\top \CCstar}\leq 1} \| \CCstar^\top \DD \CCstar \xx \|_{\CCstar^\top \CCstar}^* \\
&= \max_{\xx: \|\CCstar\xx\|_2 \leq 1} \max_{\yy:\|\CCstar\yy\|_2\leq 1} \langle \CCstar^\top \DD \CCstar \xx, \yy \rangle \\
&= \max_{\ff, \gg: \|\ff\|_2 \leq 1, \|\gg\|_2 \leq 1} \langle \DD \ff,  \gg\rangle  \\
&\leq \max_i  \DD_{ii} \,.
\end{align*}

Next we bound the distance $\|\xx-\xx^*\|_{\CCstar^\top \CCstar}$. To do so, we Taylor expand:
\begin{align*}
\Phi(\xx) - \Phi(\xx^*) &=   \frac{1}{2} (\xx-\xx^*)^\top   \left(\int_0^1 \nabla^2 \Phi((1-t) \xx^* + t \xx)dt\right) (\xx-\xx^*) \\
&= \frac{1}{2}\left(
 \left\|
 \int_0^1 \nabla^2 \Phi((1-t) \xx^* + t \xx)dt
 \right\|_{\|\cdot\|_{\CCstar^\top \CCstar}\rightarrow \|\cdot\|_{\CCstar^\top \CCstar}^*} \left\| \xx-\xx^* \right\|_{\CCstar^\top \CCstar}\right)^2 \\
 &\geq \frac{1}{2} \left( \min_i \DD_{ii} \cdot \|\xx-\xx^*\|_{\CCstar^\top \CCstar} \right)^2\,. 
\end{align*}
Since the solution we obtain satisfies $\Phi(\xx)-\Phi(\xx^*) \leq \epsilon$, we thus infer that
\[
\|\xx-\xx^*\|_{\CCstar^\top \CCstar} \leq \frac{\sqrt{2\epsilon}}{\min_i  \DD_{ii} }\,.
\]
Plugging into the upper bound on the gradient norm, we obtain that:
\[
\|\nabla \Phi(\xx)\|_{\CCstar^\top \CCstar}^* \leq \frac{\max_i  \DD_{ii} }{\min_i  \DD_{ii} }   \cdot \sqrt{2\epsilon}\,.
\]
Since we have that $\nabla \Phi(\xx) = \CCstar^\top \hh$, this implies an upper bound on the norm of $\hh$. Indeed we have that $\|\CCstar^\top \hh\|_{\CCstar^\top \CCstar}^* = \max_{\yy: \|\CCstar\yy\|\leq 1} \langle \CCstar^\top \hh, \yy \rangle = \max_{\yy: \|\CCstar\yy\|\leq 1} \langle \hh, \CCstar \yy \rangle = \|\hh\|_2$. Therefore the same upper bound also holds for $\|\hh\|_2$. Hence under the assumption that all slacks are lower bounded by a small polynomial, we can obtain that $\|\hh\|_2\leq M^{-1}$ by setting $\epsilon$ to an appropriately small polynomial.

Therefore after (approximately) solving the regularized objective (\ref{eq:regularized_newton}) we can set its gradient exactly to $0$ only by slightly perturbing the linear term. This extra error, in turn, gets passed to Lemma~\ref{lem:perturbation_correction} which removes the error by slightly modifying the weight vector $\ww$. Although this operation may enable some weights to decrease by a tiny amount below $1$, this can be prevented simply by uniformly upscaling all the weights. The centrality property will be then achieved for a slightly smaller parameter $\mu$ and the weight increase will be upper bounded by an 
arbitrary inverse polynomial in $m$ thus
negligible.

Finally, Lemma~\ref{lem:slack_lb} ensures that as a matter of fact all slacks are always polynomially lower bounded, and so our claim holds. Furthermore, all calls to the solver in Theorem~\ref{thm:mixedsolver} are made on instances where all parameters are well-conditioned, as required.

\begin{lemma}[Slack lower bound]
\label{lem:slack_lb}
When invoking the mixed objective solver
during the procedure described in Theorem~\ref{th:main_theorem},
at all times 
we have that $\min\left\{s_e^+, s_e^-\right\} \geq 1/m^{O(1)}$
for all edges $e\in E$.
\end{lemma}
\begin{proof}
We note that the stretch condition from Lemma~\ref{lem:prec_lemma} holds for any $\hh$.
We apply it once for $\hh = \frac{\ww^+}{\ss^+} - \frac{\ww^-}{\ss^-}$ and once for
$\hh = \frac{\cc}{\mu}$. For the former, since
\[ \left\Vert\frac{\left|\frac{\ww^+}{\ss^+} - \frac{\ww^-}{\ss^-}\right|}{\sqrt{(\ww^++\ww^-)\left(\frac{\ww^+}{(\ss^+)^2}+\frac{\ww^-}{(\ss^-)^2}\right)}} \right\Vert_\infty \leq 1\,,\]
we get
\begin{align}
\left\Vert \left( 
\frac{\ww^+}{(\ss^+)^2}+\frac{\ww^-}{(\ss^-)^2} 
+
\Rp \cdot \tff^{p-2}
\right) \tff - \left(\frac{\ww^+}{\ss^+} - \frac{\ww^-}{\ss^-}\right) \right\Vert_\infty
=  O(\delta \|\ww\|_1 \cdot  \log \|\ww\|_1) = O(m)\,, 
	\label{eq:prec_c1}
	\end{align}
and for the latter,
since $\left\Vert\frac{\frac{\cc}{\mu}}{\sqrt{(\ww^++\ww^-)\left(\frac{\ww^+}{(\ss^+)^2}+\frac{\ww^-}{(\ss^-)^2}\right)}}\right\Vert_\infty \leq \frac{W\cdot m^{4}}{2}$, we get
\begin{align}
\left\Vert \left( 
\frac{\ww^+}{(\ss^+)^2}+\frac{\ww^-}{(\ss^-)^2} 
+
\Rp \cdot \tff^{p-2}
\right) \tff - \frac{\cc}{\mu}
\right\Vert_\infty
= O(  \delta \|\ww\|_1 \cdot W \cdot m^{4} \cdot {\log \|\ww\|_1}) = m^{O(1)}\,.
	\label{eq:prec_c2}
\end{align}
Combining (\ref{eq:prec_c1}) and (\ref{eq:prec_c2}) and using the triangle inequality we get that
\begin{align*}
\left\Vert \frac{\ww^+}{\ss^+}-\frac{\ww^-}{\ss^-}\right\Vert_\infty \leq 
\left\Vert \frac{\cc}{\mu}\right\Vert_\infty + m^{O(1)} = m^{O(1)}\,.
\end{align*}
Since $\max\left\{s_e^+,s_e^-\right\} \geq \frac{1}{2}$, 
using the triangle inequality and the fact that 
$1\leq w_e^+,w_e^- \leq O(m)$
in the above
inequality implies that
$\min\left\{s_e^+,s_e^-\right\}\geq 1 / m^{O(1)}$.
\end{proof}

\section{Strengthening the Mixed Objective Solver}
\label{sec:strong_solver}

In this section we prove that the $\ell_2$-$\ell_p$ solver from~\cite{kyng2019flows} can be extended to handle a broader class of optimization problems on graphs. This will be useful in order to solve the optimization problem required by the improved method we present in Section~\ref{sec:improved_running_time}. We next state the main lemma we prove in this section.

\begin{lemma}
\label{lem:strong_solver}
Let a graph $G = (V,E)$ with $m$ edges, and a family of functions $\{g_e\}_{e\in E}$,  $g_e : \bbR \rightarrow \bbR$ such that each function satisfies $\left\vert \ln \frac{g_e''(x)}{g_e''(0)} \right\vert \leq \alpha$, for all $x\in\bbR$. Let the function defined over circulations $\ff$ in $G$:
\[
\widetilde{val}(\ff) = \sum_e g_e(f_e) + \vert f_e \vert^p\,.
\]
We can compute a circulation $\tff$ such that
\[
\widetilde{val}(\tff) - OPT 
\leq 
\frac{1}{2^{(\log m)^{O(1)}}} 
(\widetilde{val}(\ff) - OPT)
+ \frac{1}{2^{(\log m)^{O(1)}}}
\]
in $ 2^{O(\alpha + p^{3/2})}m^{1+O(1/\sqrt{p})}$ time.
\end{lemma}

The proof closely follows the lines of the iterative refinement proofs from the original paper~\cite{adil2019iterative}. Intuitively, since $g$ has bounded second order derivatives, it is always well approximated by a quadratic. Therefore the solver from Theorem~\ref{thm:mixedsolver} can be used iteratively to improve the error of the current solution. 

Iterative refinement is based on the following basic statement.
\begin{lemma}
\label{lem:basic_iterative_refinement_progress}
Let a linear subspace $\domX \subseteq \bbR^m$ , let  $h, k : \domX \rightarrow \bbR$ be convex twice-differentiable functions, and let $\xx, \xx^* \in \bbR^m$ such that $\xx^*$ minimizes $h$. Suppose that $k$ satisfies
\[
k(c \vdelta) \leq c^2 k(\vdelta)\,, 
\]
for all $c \in \bbR$, $\vdelta \in \domX$, 
such that
\[
k(\vdelta) 
\leq 
h(\xx+\vdelta) - h(\xx) - \langle \nabla h(\xx), \vdelta \rangle
\leq
\beta \cdot k(\vdelta)\,,
\]
for any $\vdelta \in \domX$.
Then letting 
\[
\vdelta^* = \arg\min_{\vdelta \in \domX} \langle  \nabla h(\xx), \vdelta \rangle + k(\vdelta)
\]
and $\vdelta^{\sharp} \in \domX$ such that
\[
\langle \nabla h(\xx), \vdelta^{\sharp} \rangle + k(\vdelta^{\sharp})
\leq 
\frac{1}{\gamma}\left(
\langle \nabla h(\xx), \vdelta^{*} \rangle + k(\vdelta^{*})
\right) + \epsilon\,,
\]
one has that
\[
h(\xx) - h(\xx+\vdelta^{\sharp}) \geq \frac{1}{\beta\gamma}\left( h(\xx) - h(\xx^*) \right) - \frac{\epsilon}{\beta}\,.
\]
\end{lemma}
\begin{proof}
Using the hypothesis, we have that
\begin{align*}
h(\xx + \vdelta^{\sharp}/\beta) 
&\leq 
h(\xx) + \langle \nabla h(\xx), \vdelta^{\sharp}/\beta \rangle + \beta\cdot k(\vdelta^{\sharp}/\beta)\\
&\leq h(\xx) +\frac{1}{\beta}\left( \langle \nabla h(\xx), \vdelta^{\sharp} \rangle +  k(\vdelta^{\sharp})\right)\,,
\end{align*}
where we used the right hand side of the sandwiching inequality from the hypothesis, and the fact that $k(\vdelta^{\sharp}/\beta) \leq k(\vdelta^{\sharp})/\beta^2$.

Next we plug in the relation between $\vdelta^*$ and $\vdelta^{\sharp}$ to obtain:
\begin{align*}
h(\xx + \vdelta^{\sharp}/\beta) 
&\leq
h(\xx) + \frac{1}{\beta \gamma} \left(  \langle  \nabla h(\xx), \vdelta^* \rangle + k(\vdelta^*) \right) + \frac{\epsilon}{\beta} \\
&\leq
h(\xx) + \frac{1}{\beta \gamma} \left(  \langle  \nabla h(\xx), \xx^*-\xx \rangle + k(\xx^*-\xx) \right) + \frac{\epsilon}{\beta} 
\,,
\end{align*}
where the latter inequality follows from the fact that $\vdelta^*$ minimizes $\langle \nabla h(\xx), \vdelta^* \rangle + k(\vdelta^*)$, so plugging in $\xx^*-\xx$ as an argument can only increase the sum of these two terms. 
Finally, we plug in the left hand side of the sandwiching inequality, i.e.
\[
\langle \nabla h(\xx), \xx^* -\xx  \rangle+ k(\xx^* - \xx)  \leq \hh(\xx^*) - \hh(\xx)\,,
\]
which combined with the previous inequality yields,
\begin{align*}
h(\xx) - h(\xx+\vdelta^{\sharp}/\beta) 
\geq 
\frac{1}{\beta\gamma}
\left(
h(\xx) - h(\xx^*)
\right) - \frac{\epsilon}{\beta}\,,
\end{align*}
which is what we needed.
\end{proof}

We apply Lemma~\ref{lem:basic_iterative_refinement_progress} for a custom choice of $k$, which is dictated by the specific family of functions $\{g_e\}_{e\in E}.$.

In order to do so, we also require a sandwiching inequality for the $\sum f_e^p$ term. We use the following inequality from~\cite{kyng2019flows}.
\begin{lemma}[\cite{kyng2019flows}, p. 11]
\label{lem:lp_ineq}
Let $x, \delta \in \bbR$ and $p \geq 2$. Then
\[
2^{-O(p)} (x^{p-2} \delta^2 + \delta^p)
\leq
(x+\delta)^p  - x^p - px^{p-1}
\leq 
2^{O(p)} (x^{p-2} \delta^2 + \delta^p)
\,.
\]
\end{lemma}
Using Lemmas~\ref{lem:basic_iterative_refinement_progress} and~\ref{lem:lp_ineq}, we can now define an appropriate function $k$ for each circulation $\ff$. In order to do so, we use the following simple lemma.

\begin{lemma}
\label{lem:general_obj_lp_basic}
Let $\widetilde{val}(\ff)$ specified as in Lemma~\ref{lem:strong_solver}, and let a circulation $\ff$. Then one has that for any circulation $\vdelta$:
\begin{align*}
\sum_e 
\left(
\frac{e^{-\alpha} g_e''(0)}{2} \delta_e^2 + 2^{-O(p)}(  f_e^{p-2} \delta_e^2 + \delta_e^p)
\right)
\leq 
\widetilde{val}(\ff + \vdelta) - \widetilde{val}(\ff) - \langle \nabla \widetilde{val}(\ff), \vdelta \rangle
\\
\hfill\leq
\sum_e 
\left(
\frac{e^{\alpha} g_e''(0)}{2} \delta_e^2 + 2^{O(p)} ( f_e^{p-2} \delta_e^2 + \delta_e^p)
\right)\,.
\end{align*}
\end{lemma}
\begin{proof}
We require lower bounding and upper bounding the terms of order higher than $1$, after expanding $\widetilde{val}(\ff + \vdelta)$ around $\widetilde{val}(\ff)$. To do so, we first notice that by the hypothesis in Lemma~\ref{lem:strong_solver} we have
\[
\frac{e^{-\alpha} g_e''(0)}{2} \delta^2
\leq
g_e(x+\delta) - g_e(x) - g_e'(x) \delta 
\leq 
\frac{e^{\alpha} g_e''(0)}{2} \delta^2\,.
\]
Similarly, we use Lemma~\ref{lem:lp_ineq} to lower and upper bound the higher order terms of $f_e^p$ for each $e$. Combining, we obtain the desired claim.
\end{proof}

Now we can prove that we can decrease $\widetilde{val}(f) - OPT$ very fast, which in turn enables us to prove the main lemma in this section.
\begin{proof}[Proof of Lemma~\ref{lem:strong_solver}]
We use Lemma~\ref{lem:basic_iterative_refinement_progress} where we define the functions $k_e$ based on Lemma~\ref{lem:general_obj_lp_basic}. More precisely, 
for each $e \in E$ we let 
\[
k(\vdelta) = \sum_e k_e(\delta_e)\,,
\]
where
\[
k_e(\delta_e) = \frac{e^{-\alpha} g_e''(0)}{2} \delta_e^2 + 2^{-O(p)}( f_e^{p-2} \delta_e^2 +  \delta_e^p)\,.
\]
Using Lemma~\ref{lem:general_obj_lp_basic} we verify that
\[
k(\vdelta) \leq
\widetilde{val}(\ff+\vdelta)-\widetilde{val}(\ff) - \langle \nabla \widetilde{val}(\ff), \vdelta \rangle
\leq
\max\{e^{2\alpha} , 2^{O(p)}\} \cdot k(\vdelta)\,.
\]
We use the solver from Theorem~\ref{thm:mixedsolver} to approximately minimize $k(\vdelta)$ plus the corresponding linear term. 
Then, for our specific setting, we can apply Lemma~\ref{lem:basic_iterative_refinement_progress} with

\begin{align*}
\gamma &= {1 + \frac{1}{2^{(\log m)^{O(1)}}}}\,, \\
\epsilon &= \frac{1}{2^{(\log m)^{O(1)}}}\,,\\
\beta &= \max\{e^{2\alpha} , 2^{O(p)}\}\,,
\end{align*}
to get that the newly obtained iterate $\ff'$ satisfies
\begin{align*}
\widetilde{val}(\ff)-\widetilde{val}(\ff') 
\geq \frac{1}{\beta \gamma} (  \widetilde{val}(\ff) - OPT  ) - \frac{\epsilon}{\beta}\,.
\end{align*}
Therefore executing this step for $T$ iterations, we obtain $\ff_T$ such that
\begin{align*}
\widetilde{val}(\ff_T) - OPT 
&\leq
 \left(1-\frac{1}{\beta\gamma}\right)^T
\left(\widetilde{val}(\ff) - OPT \right) + \frac{\epsilon}{\beta} \left( \sum_{t=0}^{T-1} \left(1-\frac{1}{\beta\gamma}\right)^t \right)
\\
&\leq
 \left(1-\frac{1}{\beta\gamma}\right)^T
\left(\widetilde{val}(\ff) - OPT \right) + {\epsilon}{\gamma}\,.
\end{align*}
Hence setting $T = \beta \gamma (\log m)^{O(1)}$ we make
\[
\widetilde{val}(\ff_T)-OPT \leq \frac{1}{2^{(\log m)^{O(1)}}}(\widetilde{val}(\ff) - OPT) + \frac{1}{2^{(\log m)^{O(1)}}}\,.
\]
Hence we require $T = (e^{2\alpha}+ 2^{O(p)}) \log^{O(1)} m$ iterations to obtain the target accuracy. Together with the running time guarantee from  Theorem~\ref{thm:mixedsolver} we obtain the claim.
\end{proof}

\section{Deferred Proofs}

\subsection{Proof of Lemma~\ref{lem:centrality_condition}}
\label{sec:centrality_condition_proof}
\begin{proof}
Writing first order optimality conditions for (\ref{eq:barrier_formulation}), we have
$$\nabla F_\mu^{\ww} (\xx) = \frac{\CC^\top \cc}{\mu} + \CC^\top \left(\frac{\ww^+}{\ss^+}-\frac{\ww^-}{\ss^-}\right)\,.$$ Setting this to $\zerov$ yields the conclusion. Now letting 
$\yy = \mu \cdot \ww/\ss$, we see that $\yy \geq \zerov$ since both weights $\ww$ and slacks $\ss$ are non-negative, and $\CC^\top \left(\yy^+-\yy^-\right) = -\CC^\top c$, from which we conclude that $\yy$ is a feasible dual vector.
Finally, we can write the duality gap as
\begin{align}
& \langle\cc, \CC\xx\rangle 
+ \langle \onev - \ff_0, \yy^+ \rangle  + \langle \ff_0, \yy^-\rangle
= -\langle\CC^\top (\yy^+-\yy^-), \xx\rangle
+ \langle \onev - \ff_0, \yy^+ \rangle  + \langle \ff_0, \yy^-\rangle\notag\\
&= \langle \onev - \ff_0 - \CC\xx, \yy^+ \rangle  + \langle \ff_0 + \CC\xx, \yy^-\rangle\notag
=  \langle \ss^+, \mu\cdot \ww^+/\ss^+ \rangle
+  \langle \ss^-, \mu\cdot \ww^-/\ss^- \rangle
\notag\\
&= \mu \|\ww\|_1\,.\notag
\end{align}
\end{proof}

\subsection{Proof of Lemma~\ref{lem:equiv_energy}}
\label{sec:lem_equiv_energy_proof}
\begin{proof}
We write 
\begin{align*}
\energy{\ww,\ss}{\hh} 
& = \min_{\tyy: \CC^\top (\tyy + \hh) = \zerov} \frac{1}{2} \sum\limits_{e\in E} \left(\frac{w_e^+}{(s_e^+)^2} + \frac{w_e^-}{(s_e^-)^2}\right)^{-1} (\ty_e)^2\\
& = \min_{\tyy}\max_{\txx}\, \frac{1}{2} \sum\limits_{e\in E} \left(\frac{w_e^+}{(s_e^+)^2} + \frac{w_e^-}{(s_e^-)^2}\right)^{-1} (\ty_e)^2 - \langle \txx, \CC^\top \left(\tyy+\hh\right)\rangle \\
& = \max_{\txx} \min_{\tyy}\, \frac{1}{2} \sum\limits_{e\in E} \left(\frac{w_e^+}{(s_e^+)^2} + \frac{w_e^-}{(s_e^-)^2}\right)^{-1} (\ty_e)^2 - \langle \CC \txx, \tyy+\hh\rangle
\end{align*}
The solution to the inner optimization problem is $\tyy = \left(\frac{\ww^+}{(\ss^+)^2} + \frac{\ww^-}{(\ss^-)^2} \right) \CC \txx$, therefore by substituting we get
\begin{align*}
\energy{\ww,\ss}{\hh} 
& = \max_{\txx}\, \langle \hh, \CC \txx\rangle - \frac{1}{2} \sum\limits_{e\in E} \left(\frac{w_e^+}{(s_e^+)^2} + \frac{w_e^-}{(s_e^-)^2}\right) (\CC\txx)_e^2\\
& = \max_{\tff = \CC \txx}\, \langle \hh, \tff\rangle - \frac{1}{2} \sum\limits_{e\in E} \left(\frac{w_e^+}{(s_e^+)^2} + \frac{w_e^-}{(s_e^-)^2}\right) (\tf_e)^2\,,
\end{align*}
proving (\ref{eq:energy}). Now the first order optimality condition of this problem is given by 
\[ \CC^\top \left(\frac{\ww^+}{(\ss^+)^2} + \frac{\ww^-}{(\ss^-)^2}\right)\cdot (\CC \txx) = \CC^\top \hh \]
or equivalently, by setting $\tff = \CC \txx$,
\[ \CC^\top \left(\frac{\ww^+}{(\ss^+)^2} + \frac{\ww^-}{(\ss^-)^2}\right)\cdot \tff = \CC^\top \hh\,. \]
If $\txx, \tff$ are solutions to the above linear system, this implies
\begin{align*}
\energy{\ww,\ss}{\hh} 
& = \langle \hh, \CC \txx\rangle - \frac{1}{2} \sum\limits_{e\in E} \left(\frac{w_e^+}{(s_e^+)^2} + \frac{w_e^-}{(s_e^-)^2}\right) (\CC\txx)_e^2\\
& = \langle \CC^\top \hh, \txx\rangle - \frac{1}{2} \sum\limits_{e\in E} \left(\frac{w_e^+}{(s_e^+)^2} + \frac{w_e^-}{(s_e^-)^2}\right) (\CC\txx)_e^2\\
& = \frac{1}{2} \sum\limits_{e\in E} \left(\frac{w_e^+}{(s_e^+)^2} + \frac{w_e^-}{(s_e^-)^2}\right) (\CC\txx)_e^2\\
& = \frac{1}{2} \sum\limits_{e\in E} \left(\frac{w_e^+}{(s_e^+)^2} + \frac{w_e^-}{(s_e^-)^2}\right) (\tf_e)^2\,.
\end{align*}
The equations in terms of $\vrho$ follow by simple substitution.
\end{proof}

\subsection{Proof of Lemma~\ref{lem:energy_contract}}
\label{sec:energy_contract_proof}
\begin{proof}
We explicitly write the new residual after performing the update. We have:
\begin{align*}
-\CC^\top \hh' := \nabla F_\mu^{\ww}(\xx') 
&=  \CC^\top \left(\frac{\ww^+}{(\ss^+)'} - \frac{\ww^-}{(\ss^-)'} +  \frac{\cc}{\mu} \right) \\
&=  \CC^\top \left(\frac{\ww^+}{\ss^+} - \frac{\ww^-}{\ss^-} +  \frac{\cc}{\mu} \right) 
 + \CC^\top \left(\frac{\ww^+}{(\ss^+)'} - \frac{\ww^+}{\ss^+} 
 - \frac{\ww^-}{(\ss^-)'} + \frac{\ww^-}{\ss^-} \right)\\
&\overset{(1)}{=}  -\CC^\top \left(\frac{\ww^+\vrho^+}{\ss^+} - \frac{\ww^-\vrho^-}{\ss^-}\right)
 + \CC^\top \left(\frac{\ww^+\vrho^+}{\ss^+(\onev-\vrho^+)}
 - \frac{\ww^-\vrho^-}{\ss^-(\onev-\vrho^-)} \right)\\
&=  \CC^\top \left(\frac{\ww^+(\vrho^+)^2}{\ss^+(\onev-\vrho^+)}
 - \frac{\ww^-(\vrho^-)^2}{\ss^-(\onev-\vrho^-)} \right)\\
& =  \CC^\top \left(\frac{\ww^+(\vrho^+)^2}{(\ss^+)'}
 - \frac{\ww^-(\vrho^-)^2}{(\ss^-)'} \right)\,,
\end{align*}
where $(1)$ follows from 
Equations (\ref{eq:rhodef}) and (\ref{eq:linsys}) 
and from writing
\begin{align*}
&(\ss^+)' = \ss^+ - \CC \txx = \ss^+ \left(\onev - \frac{\CC \txx}{\ss^+}\right) = \ss \left(\onev-\vrho^+\right)\\
&(\ss^-)' = \ss^- + \CC \txx = \ss^- \left(\onev + \frac{\CC \txx}{\ss^+}\right) = \ss \left(\onev - \vrho^-\right)\,.
\end{align*}
 Now we can upper bound the energy required to route the new residual with resistances determined by $(\ww, \ss')$, by substituting 
\[
\tyy = \frac{\ww^+ (\rho^+)^2}{(\ss^+)'} - \frac{\ww^- (\rho^-)^2}{(\ss^-)'}
\] into
Definition~\ref{def:energy},
after noting that 
\[
\CC^\top \tyy = \CC^\top \left(\frac{\ww^+ (\vrho^+)^2}{(\ss^+)'} - \frac{\ww^- (\vrho^-)^2}{(\ss^-)'}\right) = -\CC^\top \hh'\,.
\]
We thus obtain
\begin{align*}
\energy{\ww,\ss'}{\hh'} 
& \leq \frac{1}{2} 
\sum\limits_{e\in E} \frac{\left(\frac{w_e^+ (\rho_e^+)^2}{(s_e^+)'} - \frac{w_e^-(\rho_e^-)^2}{(s_e^-)'}\right)^2}{\frac{w_e^+}{(s_e^+)^{'2}}+\frac{w_e^-}{(s_e^-)^{'2}}}
 \leq \frac{1}{2} \sum\limits_{e\in E} \frac{\left(\frac{w_e^+ (\rho_e^+)^2}{(s_e^+)'}\right)^2 + \left(\frac{w_e^-(\rho_e^-)^2}{(s_e^-)'}\right)^2}{\frac{w_e^+}{(s_e^+)^{'2}}+\frac{w_e^-}{(s_e^-)^{'2}}}\\
& \leq \frac{1}{2} \sum\limits_{e\in E} 
\frac{\left(\frac{w_e^+ (\rho_e^+)^2}{(s_e^+)'}\right)^2 }
{\frac{w_e^+}{(s_e^+)^{'2}}}
+ \frac{\left(\frac{w_e^-(\rho_e^-)^2}{(s_e^-)'}\right)^2}
{\frac{w_e^-}{(s_e^-)^{'2}}}
 = \frac{1}{2} \sum\limits_{e\in E} \left(w_e^+ (\rho_e^+)^4 + w_e^-(\rho_e^-)^4\right)\,.
\end{align*}
\end{proof}

\subsection{Proof of Corollary~\ref{cor:energy_contract}}
\label{sec:cor_energy_contract_proof}
\begin{proof}
From Lemma~\ref{lem:equiv_energy} we have that 
\[
\energy{\ww,\ss}{\hh} = \frac{1}{2} \sum\limits_{e\in E} \left(w_e^+ (\rho_e^+)^2+w_e^-(\rho_e^-)^2\right)\,.
\]
Since $\ww \geq 1$ we have that 
\[
\|\vrho\|_\infty^2 \leq 2\cdot\energy{\ww,\ss}{\hh}\,.
\]
Using Lemma~\ref{lem:energy_contract} we upper bound
\begin{align*}
\energy{\ww,\ss'}{\hh'} \leq \frac{1}{2} \sum\limits_{e\in E} \left(w_e^+ (\rho_e^+)^4+w_e^-(\rho_e^-)^4\right)
\leq \frac{1}{2} \|\vrho\|_\infty^2 \cdot 
\sum\limits_{e\in E}\left( w_e^+ (\rho_e^+)^2 + w_e^- (\rho_e^-)^2\right)
\leq 2\cdot \energy{\ww,\ss}{\hh}^2\,.
\end{align*}
\end{proof}

\subsection{Proof Lemma~\ref{lem:fine_correction}}
\label{sec:fine_correction_proof}
\begin{proof}
By Lemma~\ref{lem:equiv_energy} we know that there exists a vector $\vrho$  such that 
\[
\CC^\top \hh = \CC^\top \left(\frac{\ww^+ \vrho^+}{\ss^+}-\frac{\ww^-\vrho^-}{\ss^-}\right)
\]
and
\[
\energy{\ww,\ss}{\hh} = \frac{1}{2} \sum_{e\in E} \left(w_e^+ (\rho_e^+)^2 +w_e^-(\rho_e^-)^2\right)\leq \epsilon\,.\]
Therefore $\|\vrho\|_\infty \leq \sqrt{{2 \epsilon}}$.
Now let 
$\ww' = \frac{\ww(\onev+\vrho)}{1-\|\vrho\|_\infty}$
and
$\mu' = \frac{\mu}{1-\|\vrho\|_\infty}$. We have that
\begin{align*}
\nabla F_{\mu'}^{\ww'}(\xx) &= \CC^\top \left(\frac{(\ww^+)'}{\ss^+} -\frac{(\ww^-)'}{\ss^-}-\frac{\cc}{\mu'} \right)\\
&= \frac{1}{1-\|\vrho\|_\infty} \cdot
\left(
 \CC^\top \left(\frac{\ww^+}{\ss^+} -\frac{\ww^-}{\ss^-}-\frac{\cc}{\mu} \right)
+ \CC^\top \left(\frac{\ww^+\vrho^+}{\ss^+} -\frac{\ww^-\vrho^-}{\ss^-}\right)
\right)\\
&= \frac{1}{1-\|\vrho\|_\infty} \cdot \left( -\CC^\top \hh + \CC^\top \hh \right)\\
&= \zerov\,.
\end{align*}
Furthermore, by this construction we see that $\ww \leq \ww'$, and that
\[
\ww' \leq \ww \cdot \frac{1+\|\vrho\|_\infty}{1-\|\vrho\|_\infty} \leq \ww \cdot \frac{1+\sqrt{2\epsilon}}{1-\sqrt{2\epsilon}} \leq \ww \cdot (1+4\sqrt{\epsilon})\,,
\]
whenever $\epsilon \leq 1/100$. Therefore the total increase in weight is at most $4\sqrt{\epsilon}\|\ww\|_1$. Furthermore the loss of duality gap is determined by $\mu' \leq \mu (1+2\sqrt{\epsilon})$.
\end{proof}

\subsection{Proof of Lemma~\ref{lem:correction}}
\label{sec:lem_correction_proof}
\begin{proof}
We first perform $O(\log \log \|\ww\|_1)$ vanilla residual correction steps 
as described in Definition~\ref{def:residual_correction}
to obtain a new solution $\ff' = \ff_0 + \CC \xx'$ with residual 
$\nabla F_{\mu}^{\ww}(\xx')= -\CC^\top \gg'$ and low energy
$\energy{\ww,\ss'}{\gg'} \leq \|\ww\|_1^{-22}/16$.
Then we apply the perfect correction step from Lemma~\ref{lem:fine_correction} to 
eliminate the residual $-\CC^\top \gg'$ 
by
obtaining a new set of weights $\ww'\geq \ww$
such that $\|\ww' -\ww\|_1\leq \|\ww\|^{-10})$
and $\nabla F_{\mu'}^{\ww'}(\xx') = \zerov$
where $\mu' \leq \mu (1 + \frac{1}{2} \|\ww\|_1^{-11})$.
\end{proof}

\subsection{Proof of Lemma~\ref{lem:dumb_predictor}}
\label{sec:dumb_predictor_proof}
\begin{proof}
Since $\ff$ is $\mu$-central, we can write:
\begin{align*}
\nabla F_{\mu'}^{\ww}(\xx) 
&= \CC^\top \left(\frac{\ww^+}{\ss^+} - \frac{\ww^-}{\ss^-}+ \frac{\cc}{\mu'}\right)
=(1+\delta) \CC^\top \left(\frac{\ww^+}{\ss^+} - \frac{\ww^-}{\ss^-}+ \frac{\cc}{\mu} \right) 
-\delta \CC^\top \left(\frac{\ww^+}{\ss^+}-\frac{\ww^-}{\ss^-}\right)\\
&=-\delta \CC^\top \left(\frac{\ww^+}{\ss^+}-\frac{\ww^-}{\ss^-}\right)\,,
\end{align*}
and so we can set $\hh' = \delta \left(\frac{\ww^+}{\ss^+} - \frac{\ww^-}{\ss^-}\right)$.
Using Definition~\ref{def:energy} we can upper bound the energy required to route this residual by exhibiting the solution
$\tyy = -\delta \left(\frac{\ww^+}{\ss^+}-\frac{\ww^-}{\ss^-}\right)$
after noting that $\CC^\top \left(\tyy+\hh'\right) = \zerov$:
\begin{align*}
\energy{\ww,\ss}{\hh'} 
& \leq \frac{1}{2} \delta^2 
\sum\limits_{e\in E} \frac{\left(\frac{w_e^+}{s_e^+} - \frac{w_e^-}{s_e^-}\right)^2}{\frac{w_e^+}{(s_e^+)^{2}}+\frac{w_e^-}{(s_e^-)^{2}}}
 \leq \frac{1}{2} \delta^2 \sum\limits_{e\in E} \frac{\left(\frac{w_e^+}{s_e^+}\right)^2 + \left(\frac{w_e^-}{s_e^-}\right)^2}{\frac{w_e^+}{(s_e^+)^{2}}+\frac{w_e^-}{(s_e^-)^{2}}}
 \leq \frac{1}{2} \delta^2 
 \left(
\sum\limits_{e\in E} \frac{\left(\frac{w_e^+}{s_e^+}\right)^2}{\frac{w_e^+}{(s_e^+)^{2}}}
+ \sum\limits_{e\in E} \frac{\left(\frac{w_e^-}{s_e^-}\right)^2}{\frac{w_e^-}{(s_e^-)^{2}}}
\right)
\\
& = \frac{1}{2} \delta^2 \left\Vert w\right\Vert_1\,,
\end{align*}
which is at most $1/4$ as long as $\delta \leq \frac{1}{(2\|\ww\|_1)^{1/2}}$.
\end{proof} 

\subsection{Proof of Lemma~\ref{lem:vanilla_mcf_final}}
\label{sec:vanilla_mcf_final_proof}
\begin{proof}
Given a $\mu$-central flow and
using Lemma~\ref{lem:dumb_predictor} we see that setting $\delta = \frac{1}{(2 \|\ww\|_1)^{1/2}}$ and $\mu' = \mu(1-\delta)$
we obtain $\energy{\ww,\ss}{\hh} \leq 1/4$, where $-\CC^\top \hh = \nabla F_{\mu'}^{\ww}(\xx)$.
Hence applying Corollary~\ref{cor:energy_contract} for $O(\log \log m)$ iterations
we obtain a new flow $\ff' = \ff_0 + \CC \xx'$
with slacks $\ss'$ and residual $-\CC^\top \hh' = \nabla F_{\mu'}^{\ww}(\xx')$
such that 
$\energy{\ww,\ss'}{\hh'} \leq m^{-20}/4$.

Finally, applying the perfect correction step from Lemma~\ref{lem:fine_correction}
we obtain a new set of weights $\ww' \geq \ww$, such that 
$\ww' \leq \ww (1+m^{-10})$
and 
${\nabla F_{\mu''}^{\ww'}(\xx')} = \zerov$
for $\mu'' \leq \mu'(1+m^{-10})$.
In other words,  $\ff'$ is $\mu''$-central with respect to $\ww'$.

Since the increase in weights is very small, iterating this procedure for 
$O(m^{1/2} \log m)$ steps maintains the  invariant that $\|\ww\|_1 \leq 2m+1$. Furthermore, 
in each iteration the parameter $\mu$ gets scaled down by a factor of $1+\frac{1}{(4\|\ww\|_1)^{1/2}} \geq 1+\frac{1}{3m^{1/2}}$, after which it gets slightly scaled up by at most $1+m^{-10}$ due to the perfect correction step. Hence within $O(m^{1/2} \log \frac{m \mu^0}{\epsilon}) = O(m^{1/2}\log m)$ iterations, the parameter $\mu$ gets scaled down by a factor of 
$\Omega(m \mu^0/ \epsilon)$, which thus implies that the final duality gap will be $O\left( \frac{ 3m \cdot \mu^0}{m \mu^0/\epsilon} \right) = O(\epsilon)$.
\end{proof}

\subsection{Proof of Lemma~\ref{lem:emax_upperbound}}
\label{sec:emax_upperbound_proof}
\begin{proof}
For the first inequality, we use the test vector $\tyy = -\hh$ into 
Definition~\ref{def:energy}.
For the second one, we have
\begin{align*}
\energymax(\hh,\ww,\ss) 
	& = \frac{1}{2} \sum\limits_{e\in E}\frac{\delta^2\left(\frac{w_e^+}{s_e^+}-\frac{w_e^-}{s_e^-}\right)^2}{\frac{w_e^+}{(s_e^+)^2} + \frac{w_e^-}{(s_e^-)^2}}
	 \leq \frac{1}{2} \delta^2\sum\limits_{e\in E}\frac{\left(\frac{w_e^+}{s_e^+}\right)^2+\left(\frac{w_e^-}{s_e^-}\right)^2}{\frac{w_e^+}{(s_e^+)^2} + \frac{w_e^-}{(s_e^-)^2}}\\
	& \leq \frac{1}{2} \delta^2\sum\limits_{e\in E}
	\left(\frac{\left(\frac{w_e^+}{s_e^+}\right)^2}
{\frac{w_e^+}{(s_e^+)^2}} + 
\frac{\left(\frac{w_e^-}{s_e^-}\right)^2}
{\frac{w_e^-}{(s_e^-)^2}}\right)
	= \frac{1}{2}\delta^2 \sum\limits_{e\in E}(w_e^++w_e^-)\\
	& = \frac{1}{2}\delta^2 \left\Vert\ww\right\Vert_1 \,.
\end{align*}
\end{proof}

\subsection{Proof of Corollary~\ref{cor:l_p_upperbound}}
\label{sec:cor_l_p_upperbound_proof}
\begin{proof}
We can upper bound using Lemma~\ref{lem:opt_non_aug}:
\begin{align}
\|\tff_\star\|_p
&\leq \notag
\left(\frac{p\cdot \energymax(\hh,\ww,\ss)}{\Rp}\right)^{1/p} 
\leq \notag
\left(\frac{p  \cdot\frac{1}{2} \delta^2 \|\ww\|_1}{p \cdot\left(10^6 \cdot \delta^2 \|\ww\|_1 \cdot {\log\|\ww\|_1}\right)^{p+1} }\right)^{1/p} \\
&\leq \frac{1}{ 10^6\cdot\delta^2 \|\ww\|_1\cdot{\log\|\ww\|_1}}
\,. 
\end{align}
\end{proof}

\subsection{Proof of Corollary~\ref{cor:gamma}}
\label{sec:cor_gamma_proof}
\begin{proof}
We have
\begin{align*}
\hat{\gamma}
& = 
\left(\Rstar
+\Rp \cdot \|\tff_\star\|_\infty^{p-2}\right)^{1/2} 
\cdot \left\| \frac{ \hh }{  \sqrt{(\ww^+ + \ww^-)\left( \frac{\ww^+}{(\ss^+)^2}+\frac{\ww^-}{(\ss^-)^2} \right) }} \right\|_\infty \cdot {32 \log \|\ww\|_1}\,.
\end{align*}
In particular, for sufficiently large $m$ and since 
$\|\tff_\star\|_\infty \leq \|\tff_\star\|_p$ we have
\begin{equation}\label{eq:Rp_Rstar_comparison}
\begin{aligned}
R_p\left\Vert\tff_\star\right\Vert_\infty^{p-2} 
& \leq \frac{p\left(10^6\cdot\delta^2\left\Vert\ww\right\Vert_1\cdot{\log\left\Vert\ww\right\Vert_1}\right)^{p+1}}
{\left(10^6\cdot\delta^2\left\Vert\ww\right\Vert_1\cdot{\log\left\Vert\ww\right\Vert_1}\right)^{p-2}} 
 \leq p \cdot \left(10^6 \cdot \delta^2 \left\Vert\ww\right\Vert_1 \cdot{\log \left\Vert\ww\right\Vert_1}\right)^3\\
& = \left(\delta^2 \|\ww\|_1\right)^3 \cdot p \cdot \left(10^6 \cdot {\log \left\Vert\ww\right\Vert_1}\right)^3
 < 3\delta^2 \left\Vert \ww\right\Vert_1^2
 = R_\star\,,
\end{aligned}
\end{equation}
where the last inequality follows from the assumption on $\delta$. 
Furthermore, we bound the term under the $\ell_\infty$ norm as:
\begin{equation}\label{eq:residual_linfty_bound}
\begin{aligned}
\frac{\left|h_e\right|}{\sqrt{(w_e^++w_e^-)\left(\frac{w_e^+}{(s_e^+)^2} + \frac{w_e^-}{(s_e^-)^2}\right)}}
&= \delta \frac{\left|\frac{w_e^+}{s_e^+} - \frac{w_e^-}{s_e^-}\right|}{\sqrt{(w_e^++w_e^-)\left(\frac{w_e^+}{(s_e^+)^2} + \frac{w_e^-}{(s_e^-)^2}\right)}}\\
&= \delta \left(\frac{\left(\frac{w_e^+}{s_e^+} - \frac{w_e^-}{s_e^-}\right)^2}{(w_e^++w_e^-)\left(\frac{w_e^+}{(s_e^+)^2} + \frac{w_e^-}{(s_e^-)^2}\right)}\right)^{1/2}\\
&\leq \delta \left(\frac{1}{w_e^++w_e^-} \left(\frac{\left(w_e^+/s_e^+\right)^2}{w_e^+/(s_e^+)^2} + \frac{\left(w_e^-/s_e^-\right)^2}{w_e^-/(s_e^-)^2}\right)\right)^{1/2}\\
& =\delta\,.
\end{aligned}
\end{equation}
Combining (\ref{eq:Rp_Rstar_comparison}) and (\ref{eq:residual_linfty_bound}), we get the upper bound
\begin{align*}
\hat{\gamma}
&\leq \left(2\cdot R_\star\right)^{1/2} \cdot \delta \cdot {32\log \left\Vert\ww\right\Vert_1}
= \delta^2\left\Vert\ww\right\Vert_1\cdot 32\sqrt{6}\cdot \log {\left\Vert\ww\right\Vert_1}
= \gamma\,.
\end{align*}
\end{proof}

\subsection{Proof of Lemma~\ref{lem:imb_cong_bound}}
\label{sec:lem_imb_cong_bound_proof}

\begin{proof}
We restate
(\ref{eq:precond_guarantee2}) in terms of $\gamma \geq \hat{\gamma}$:
\begin{align*}
\left(\frac{w_e^+}{(s_e^+)^2} + \frac{w_e^-}{(s_e^-)^2}\right) \cdot \left\vert\tf_e  \right\vert
&\leq 
\left\vert h_e \right \vert + \gamma\,.
\end{align*}
More specifically we will use the following, which the above implies for the setting of $\hh = \delta \left(\frac{\ww^+}{\ss^+} - \frac{\ww^-}{\ss^-}\right)$:
\begin{align*}
\left(\frac{w_e^+}{(s_e^+)^2} + \frac{w_e^-}{(s_e^-)^2}\right) \cdot \left\vert\tf_e  \right\vert
&\leq 
\frac{w_e^+\delta}{s_e^+}+\frac{w_e^-\delta}{s_e^-} + \gamma
\end{align*}
or equivalently since $\rho_e^+$ and $\rho_e^-$ have opposite signs:
\begin{align}
\left|\frac{w_e^+\rho_e^+}{s_e^+}\right| + \left|\frac{w_e^-\rho_e^-}{s_e^-}\right|
&\leq 
\frac{w_e^+\delta}{s_e^+}+\frac{w_e^-\delta}{s_e^-} + \gamma\,.
\label{eq:precond_guarantee_balancing}
\end{align}
Assume without loss of generality that $s_e^+\leq s_e^-$, and so $\left|\rho_e^+\right| \geq \left|\rho_e^-\right|$. 
Now, for the sake of contradiction
we suppose that
$\left|\rho_e^+\right| \geq C_\infty$ and $\left|\frac{w_e^+\rho_e^+}{s_e^+}\right|+\left|\frac{w_e^-\rho_e^-}{s_e^-}\right| > 6\gamma$. 
We consider the two cases of Definition~\ref{def:imbalance}:
\paragraph{(1) $\max\left\{w_e^+,w_e^-\right\} \leq \delta \left\Vert \ww\right\Vert_1$:}
Since $0 < s_e^+ \leq s_e^- < 1$ and $s_e^++s_e^-=1$, we have that $s_e^- \geq \frac{1}{2}$.
Therefore $\frac{w_e^-\delta}{s_e^-} \leq 2\cdot \delta^2 \left\Vert \ww\right\Vert_1 \leq 2 \cdot \gamma$.
We conclude that
\begin{align*}
\left|\frac{w_e^+\rho_e^+}{s_e^+}\right|
+\left|\frac{w_e^-\rho_e^-}{s_e^-}\right|\leq 
\frac{w_e^+\delta}{s_e^+}
+\frac{w_e^-\delta}{s_e^-} + \gamma \leq
\frac{1}{2} \left|\frac{w_e^+\rho_e^+}{s_e^+}\right| 
+ 3\gamma\,,
\end{align*}
since $\left|\rho_e^+\right| \geq C_\infty = \frac{1}{2\delta \sqrt{2\|\ww\|_1}} \geq 2\delta$ by our assumption on $\delta$.
Thus
\begin{align*}
\left|\frac{w_e^+\rho_e^+}{s_e^+}\right|
+\left|\frac{w_e^-\rho_e^-}{s_e^-}\right|\leq
6\gamma\,,
\end{align*}
a contradiction.

\paragraph{(2) $\min\left\{w_e^+,w_e^-\right\} \geq 96 \cdot \delta^4 \left\Vert\ww\right\Vert_1^2$:}
We will first prove that $\left|\tf_e\right| < 2 \delta$. Suppose to the contrary.
We have that $\left|\rho_e^-\right| = \frac{\left|\tf_e\right|}{s_e^-} \geq \left| \tf_e\right| \geq 2\delta$, so
\begin{align*}
\left|\frac{w_e^+\rho_e^+}{s_e^+}\right|
+\left|\frac{w_e^-\rho_e^-}{s_e^-}\right|\leq 
\frac{w_e^+\delta}{s_e^+}
+\frac{w_e^-\delta}{s_e^-} + \gamma \leq
\frac{1}{2} \left|\frac{w_e^+\rho_e^+}{s_e^+}\right| 
+ \frac{1}{2} \left|\frac{w_e^-\rho_e^-}{s_e^-}\right| 
+ \gamma
\end{align*}
and thus 
\begin{align*}
\left|\frac{w_e^+\rho_e^+}{s_e^+}\right|
+\left|\frac{w_e^-\rho_e^-}{s_e^-}\right|\leq
2\gamma\,,
\end{align*}
a contradiction, therefore $\left|\tf_e\right| < 2\delta$. This 
immediately implies that $2\delta > \left| \tf_e\right| = \left|\rho_e^+\right|s_e^+ \geq C_\infty s_e^+$,
and so $s_e^+ \leq 2\delta/C_\infty$.
Therefore we get
\begin{align*}
\left|\frac{w_e^+ \rho_e^+}{s_e^+}\right| \geq 
\frac{(96\cdot \delta^4 \left\Vert\ww\right\Vert_1^2) C_\infty}{2\delta/C_\infty} 
= 48\cdot \delta^3 \left\Vert\ww\right\Vert_1^2 \cdot C_\infty^2\,.
\end{align*}
Since 
$C_\infty^2 = \frac{1}{8\delta^2 \left\Vert\ww\right\Vert_1}$, we conclude that

\begin{align*}
\left|\frac{w_e^+ \rho_e^+}{s_e^+}\right| \geq 
6\cdot\delta \left\Vert \ww\right\Vert_1 \geq 3\frac{w_e^-\delta}{s_e^-}\,,
\end{align*}
where we used the fact that $\left\Vert w\right\Vert_\infty \leq \left\Vert w\right\Vert_1$
and the fact that $s_e^-\geq\frac{1}{2}$.
Therefore
\begin{align*}
\left|\frac{w_e^+\rho_e^+}{s_e^+}\right|
+\left|\frac{w_e^-\rho_e^-}{s_e^-}\right|\leq 
\frac{w_e^+\delta}{s_e^+}
+\frac{w_e^-\delta}{s_e^-} + \gamma \leq
\frac{1}{2} \left|\frac{w_e^+\rho_e^+}{s_e^+}\right| 
+ \frac{1}{3} \left|\frac{w_e^+\rho_e^+}{s_e^+}\right| 
+\gamma\,.
\end{align*}
Therefore we conclude that
\begin{align*}
\left|\frac{w_e^+\rho_e^+}{s_e^+}\right|
+\left|\frac{w_e^-\rho_e^-}{s_e^-}\right|\leq 
6\gamma\,,
\end{align*}
again a contradiction.

\end{proof}

\subsection{Proof of Lemma~\ref{lem:improved_correction}}
\label{sec:lem_improved_correction_proof}

\begin{proof}
Let us analyze the new residual after performing the update described in (\ref{eq:s1_update}-\ref{eq:s2_update}). Just like in the standard correction step which we analyzed in
Section~\ref{sec:correct_res},
here we can write:
\begin{align*}
\nabla F_{\mu}^{\ww}(\xx') 
&= \frac{\CC^\top \cc}{\mu} 
+ \CC^\top \left( \frac{\ww^+}{(\ss^+)'} - \frac{\ww^-}{(\ss^-)'} \right) 
\\
&= \nabla F_\mu^{\ww}(\xx) 
+ \CC^\top \left( \frac{\ww^+}{(\ss^+)'} - \frac{\ww^-}{(\ss^-)'} \right) 
- \CC^\top \left( \frac{\ww^+}{\ss^+} - \frac{\ww^-}{\ss^-} \right) 
\\
&=\CC^\top \Delta \hh 
- \CC^\top \left(  \frac{\ww^+ \vrho^+}{\ss^+}-\frac{\ww^-\vrho^-}{\ss^-} \right)
+ \CC^\top \left( \frac{\ww^+}{(\ss^+)'} - \frac{\ww^-}{(\ss^-)'} \right) 
- \CC^\top \left( \frac{\ww^+}{\ss^+} - \frac{\ww^-}{\ss^-} \right) 
\\
&=\CC^\top \Delta \hh 
- \CC^\top \left(  \frac{\ww^+(\onev + \vrho^+)}{\ss^+}-\frac{\ww^-(\onev+\vrho^-)}{\ss^-} \right)
+ \CC^\top \left( \frac{\ww^+}{(\ss^+)'} - \frac{\ww^-}{(\ss^-)'} \right) 
\\
&=\CC^\top \Delta \hh 
- \CC^\top \left(  \frac{\ww^+(\onev + \vrho^+)(\onev-\vrho^+)}{(\ss^+)'}-\frac{\ww^-(\onev+\vrho^-)(\onev-\vrho^-)}{(\ss^-)'} \right)
+ \CC^\top \left( \frac{\ww^+}{(\ss^+)'} - \frac{\ww^-}{(\ss^-)'} \right) 
\\
&=\CC^\top \Delta\hh +
\CC^\top\left(  \frac{\ww^+ (\vrho^+)^2}{(\ss^+)'} - \frac{\ww^- (\vrho^-)^2}{(\ss^-)'}\right)\,.
\end{align*}
Therefore 
\begin{align*}
\nabla F_\mu^{\ww}(\xx') - \CC^\top \Delta \hh = 
\CC^\top\left(  \frac{\ww^+ (\vrho^+)^2}{(\ss^+)'} - \frac{\ww^- (\vrho^-)^2}{(\ss^-)'}\right)\,.
\end{align*}
Next we show that modifying the weights from $\ww$ to $\ww'$ sets a subset of these entries to $0$, which will enable us to correct this new perturbed residual. 
Using the weight update described in (\ref{eq:w1}-\ref{eq:w2}) we obtain
\begin{align*}
& -\CC^\top\left(\gg + \Delta \hh\right) \\
&= 
\nabla F_\mu^{\ww'}(\xx') - \CC^\top \Delta \hh \\
&= \CC^\top \frac{\cc}{\mu} -\CC^\top\Delta\hh + \CC^\top\left(\frac{(\ww^+)'}{(\ss^+)'} - \frac{(\ww^-)'}{(\ss^-)'}\right) \\
&= \nabla F_\mu^{\ww}(\xx') -\CC^\top\Delta\hh+ \CC^\top\left(\frac{(\ww^+)'-\ww^+}{(\ss^+)'} - \frac{(\ww^-)'-\ww^-}{(\ss^-)'}\right) \\
&= \nabla F_\mu^{\ww}(\xx') -\CC^\top\Delta\hh+ \CC^\top\left(\frac{\ww^-\left(\vrho^-\right)^2}{(\ss^-)'}\cdot\onev_{\vert\vrho^-\vert\geq \Cinf} - 
		\frac{\ww^+\left(\vrho^+\right)^2}{(\ss^+)'}\cdot\onev_{\vert\vrho^+\vert\geq \Cinf}\right) \\
&= 
\CC^\top\left(\left(  \frac{\ww^+ (\vrho^+)^2}{(\ss^+)'} \cdot \onev_{\vert\vrho^+\vert < \Cinf}\right)- \left(\frac{\ww^- (\vrho^-)^2}{(\ss^-)'} \cdot \onev_{\vert\vrho^-\vert < \Cinf}\right)\right)\,.
\end{align*}
Finally, using Lemma~\ref{lem:equiv_energy} we certify an upper bound on
\begin{align*}
\energy{\ww',\ss'} { \gg +  \Delta \hh}
& \leq \energymax(\gg+\Delta\hh, \ww', \ss') \\
& = 
\frac{1}{2} \sum\limits_{e\in E} 
\frac{ \left(\left(  \frac{\ww^+ (\vrho^+)^2}{(\ss^+)'} \cdot \onev_{\vert\vrho^+\vert < \Cinf}\right)- \left(\frac{\ww^- (\vrho^-)^2}{(\ss^-)'} 
			\cdot \onev_{\vert\vrho^-\vert < \Cinf}\right)\right)_e^2
	}{\frac{w_e^{'+}}{(s_e^{+})^{'2}}+\frac{w_e^{'-}}{(s_e^{-})^{'2}}} \\
& \leq 
\frac{1}{2}\sum\limits_{\substack{e\in E\\ \vert\rho_e^+\vert<\Cinf}} \frac{ \left(\frac{w_e^+ (\rho_e^+)^2}{(s_e^+)'}\right)^2
	}{\frac{w_e^{'+}}{(s_e^{+})^{'2}}} 
+ \frac{1}{2}\sum\limits_{\substack{e\in E\\ \vert\rho_e^-\vert<\Cinf}} \frac{ \left(\frac{w_e^- (\rho_e^-)^2}{(s_e^-)'}\right)^2
	}{\frac{w_e^{'-}}{(s_e^{-})^{'2}}} \\
& \leq 
\frac{1}{2}\sum\limits_{\substack{e\in E\\ \vert\rho_e^+\vert<\Cinf}} w_e^+ (\rho_e^+)^4
+ \frac{1}{2}\sum\limits_{\substack{e\in E\\ \vert\rho_e^-\vert<\Cinf}} w_e^- (\rho_e^-)^4\\
&\leq \frac{1}{2}\Cinf^2 \cdot \sum_{e\in E} \left( w_e^+ (\rho_e^+)^2 + w_e^- (\rho_e^-)^2 \right)  \\
& =
 \frac{1}{2}\Cinf^2 \cdot \sum_{e\in E} (\tf_e)^2\left( \frac{w_e^+}{(s_e^+)^2} + \frac{w_e^-}{(s_e^-)^2} \right)  \\
&\leq \frac{1}{2}\Cinf^2 \cdot  8\cdot\energymax(\hh,\ww,\ss) \\
&\leq \frac{2\delta^2 \|\ww\|_1}{8\delta^2 \|\ww\|_1} \\
&= \frac{1}{4}\,,
\end{align*}
where we used the fact that $\ww'\geq \ww$, (\ref{eq:energy_new_residual}), and Lemma~\ref{lem:emax_upperbound}.
\end{proof}

\subsection{Proof of Lemma~\ref{lem:strong_opt_non_aug}}
\label{sec:lem_strong_opt_non_aug}
\begin{proof}
Let $\tffstar = \tff + \tff'$ where $\tff'$ is the restriction of $\tffstar$ to the edges incident to $\vstar$. 
To prove (\ref{eq:first_order_opt_mixed_obj_barrier}),
note that from Lemma~\ref{lem:optimality_avg_hessian}
there exists a vector $\valpha = (\valpha^+; \valpha^-)$, $(1+\theta)^{-2} \cdot \onev \leq \valpha \leq (1-\theta)^{-2} \cdot \onev$, 
such that for any circulation $\gg = \CCstar \zz_\star$ in $G_\star$, 
\begin{align}
\label{eq:avg_optimality_cond}
\left\langle \gg, \left[\begin{array}{c} 
\hh
- \tff \left(
\frac{\valpha^+ \ww^+}{(\ss^+)^2} 
+
\frac{\valpha^- \ww^-}{(\ss^-)^2} 
\right)
- \Rp \cdot (\tff)^{p-1}  \\ 
- \Rstar \cdot \tff' 
-\Rp  \cdot (\tff')^{p-1} 
\end{array}\right] \right\rangle 
= 0\,.
\end{align}
Restricting ourselves to circulations supported only in the non-preconditioned graph $G$, one has that for any circulation in $\gg' = \CC\zz$ in $G$:
\begin{align*}
\left\langle \zz, 
\CC^\top \left(\hh + \Delta\hh
-\tff
\left( \frac{\valpha^+\ww^+}{(\ss^+)^2} 
+
\frac{\valpha^-\ww^-}{(\ss^-)^2} \right)
\right)
\right\rangle 
= 0
\end{align*}
Since this holds for any test vector $\zz$, it must be that the second term in the inner product is $\zerov$.
Rearranging, it yields the identity from (\ref{eq:first_order_opt_mixed_obj_barrier}).
Next, we notice that (\ref{eq:avg_optimality_cond}) is the optimality condition of the following objective:
\begin{align}
\max_{\substack{\tffstar = \CCstar \txx}}\, 
\left\langle \hh, \tff \right\rangle 
&- \frac{1}{2} \sum_{e \in E} (\tf_e)^2\cdot \left( \frac{ 
\alpha_e^+ w_e^+}{(s_e^+)^2} + \frac{\alpha_e^- w_e^-}{(s_e^-)^2}\right)\notag\\
&- \frac{\Rstar}{2} \sum_{e\in E'} (\tf_e')^2
- \frac{\Rp}{p} \sum_{e\in E\cup E'} (\tfstar)_e^p\,.
\label{eq:mixed_obj_avg_resistance}
\end{align}

Let us proceed to bound the norm of the demand routed by $\tff$.
Consider the value of the objective in (\ref{eq:mixed_obj_avg_resistance})
after truncating it to only the first two terms, which we can write as:
\begin{align}
&\left\langle \hh, \tff \right\rangle 
- \frac{1}{2} \sum_{e \in E} (\tf_e)^2\cdot \left( \frac{ 
\alpha_e^+ w_e^+}{(s_e^+)^2} + \frac{\alpha_e^- w_e^-}{(s_e^-)^2}\right)\\
& \leq \sum_{e \in E} h_e \cdot \tf_e
- \frac{1}{3} \sum_{e \in E} (\tf_e)^2 \cdot \left( \frac{w_e^+}{(s_e^+)^2} + \frac{w_e^-}{(s_e^-)^2}\right) \label{eq:barrier_energy}\\
&\leq 
\frac{3}{4} \sum_{e\in E} 
h_e^2 \cdot \left( \frac{ w_e^+}{(s_e^+)^2} + \frac{w_e^-}{(s_e^-)^2}\right)^{-1} \\
&=\frac{3}{2} \energymax(\hh,\ww,\ss) \,, \label{eq:upper_bound_obj2}
\end{align}
where we used 
$\valpha \geq \frac{1}{(1+\theta)^2}\cdot \onev > \frac{2}{3}\cdot \onev$
and 
the fact that $\langle \aa,\bb\rangle \leq \frac{1}{2} \|\aa\|^2 + \frac{1}{2}\|\bb\|^2$. 

Note that the value of the regularized objective 
(\ref{eq:mixed_obj_avg_resistance})
is at least $0$ since we can always substitute $\txx = \zerov$ and obtain exactly $0$. By re-arranging,\begin{align}\label{eq:rearranged_regularized_barrier}
&\frac{\Rstar}{2} \sum_{e \in E'} (\tf_e')^2\\
& \leq 
\left\langle \hh , \tff\right\rangle
- \frac{1}{2} \sum_{e \in E} (\tf_e)^2\cdot \left( \frac{ \alpha_e^+ w_e^+}{(s_e^+)^2} + \frac{\alpha_e^-w_e^-}{(s_e^-)^2}\right)
- \frac{\Rp}{p} \sum_{e\in E\cup E'} (\tfstar)_e^p
\label{eq:partial_regularized_barrier}\\
& \leq \left\langle \hh , \tff \right\rangle
- \frac{1}{2} \sum_{e \in E} (\tf_e)^2\cdot \left( \frac{ \alpha_e^+w_e^+}{(s_e^+)^2} + \frac{\alpha_e^-w_e^-}{(s_e^-)^2}\right) \\
&\leq \frac{3}{2} \energymax(\hh,\ww,\ss) \,, \label{eq:upper_bound_energy_barrier}
\end{align}
where we also used the fact that the last term of (\ref{eq:partial_regularized_barrier}) is non-positive
and (\ref{eq:upper_bound_obj2}).
Therefore (\ref{eq:upper_bound_energy_barrier}) 
enables us to upper bound
\begin{align}
\sum_{e \in E'}(\tf_e')^2\leq
\frac{3}{\Rstar} 
\energymax(\hh,\ww,\ss) \,,
\end{align}
which implies that
\begin{align}
\sum_{e\in E'} \left| \tf_e'\right| 
& \leq|E'|^{1/2}
\cdot
\left( \sum_{e\in E'} 
(\tf_e')^2
\right)^{1/2}\\
&\leq \left( 3\|\ww\|_1 \cdot \frac{3}{\Rstar}
\energymax(\hh,\ww,\ss)
\right)^{1/2}\\
&= 3\left( \frac{\|\ww\|_1 \cdot 
\energymax(\hh,\ww,\ss)}{\Rstar}
\right)^{1/2}\,,
\end{align}
a quantity that upper bounds the demand perturbation.
Using a similar argument we can upper bound $\left\Vert \tffstar\right\Vert_p$. We have
\begin{align}\label{eq:rearranged_regularized_barrier2}
&\frac{\Rp}{p} \sum_{e\in E\cup E'} (\tf_\star)_e^p\\
& \leq 
\left\langle \hh , \tff \right\rangle
- \frac{1}{2} \sum_{e \in E} (\tf_e)^2\cdot \left( \frac{ \alpha_e^+ w_e^+}{(s_e^+)^2} + \frac{\alpha_e^-w_e^-}{(s_e^-)^2}\right)
-\frac{\Rstar}{2} \sum_{e \in E'} (\tf_e')^2
\label{eq:partial_regularized_barrier2}\\
& \leq \left\langle \hh , \tff \right\rangle
- \frac{1}{2} \sum_{e \in E} (\tf_e)^2\cdot \left( \frac{ \alpha_e^+w_e^+}{(s_e^+)^2} + \frac{\alpha_e^-w_e^-}{(s_e^-)^2}\right) \\
&\leq \frac{3}{2} \energymax(\hh,\ww,\ss) \,, \label{eq:upper_bound_energy_barrier2}
\end{align}
thus concluding that
\begin{align}
\left\Vert \tffstar\right\Vert_p
= \left(\sum\limits_{e\in E\cup E'} (\tfstar)_e^p\right)^{1/p}
\leq \left(\frac{p\cdot \frac{3}{2} \energymax(\hh,\ww,\ss)}{R_p}\right)^{1/p}\,.\label{eq:fstar_pnorm_barrier}
\end{align}
\end{proof}
 
\bibliographystyle{abbrv}
\bibliography{main}

\end{document}